\documentclass[aps,pra,twocolumn,a4paper,10pt,floatfix,amsmath,amsfonts,amssymb,accepted=2020-07-07]{quantumarticle}
\pdfoutput=1
\usepackage[utf8]{inputenc}
\usepackage[T1]{fontenc}
\usepackage{graphicx}
\usepackage{array}
\usepackage{bm}
\usepackage{multirow}
\usepackage{longtable}
\usepackage{amsthm}
\usepackage{enumerate}
\usepackage{mathtools}
\usepackage{xcolor}
\usepackage{marvosym}
\usepackage{float}
\usepackage{graphicx}
\usepackage{hyperref}

\usepackage{float}

\newtheorem{theorem}{Theorem}

\newtheorem{observation}{Observation}
\newtheorem{corollary}{Corollary}[theorem]

\newtheorem{lemma}{Lemma}
\newtheorem{definition}{Definition}

\newcommand{\tr}{\text{tr}}

\newcommand{\id}{\ensuremath{\mathbbm{1}}} 
\newcommand{\one}{\id}

\DeclarePairedDelimiter{\abs}{\lvert}{\rvert}
\DeclarePairedDelimiter{\norm}{\lVert}{\rVert}
\def\acts{\curvearrowright}
\newcommand{\bra}[1]{\langle #1|}
\newcommand{\ket}[1]{|#1\rangle}
\newcommand{\braket}[2]{\langle #1|#2\rangle}

\newcommand{\C}{\ensuremath{\mathbbm C}}

\newcommand{\bi}{\begin{itemize}}
\newcommand{\ei}{\end{itemize}}

\newcommand{\be}{\begin{equation}}
\newcommand{\ee}{\end{equation}}
\newcommand{\bea}{\begin{eqnarray}}
\newcommand{\eea}{\end{eqnarray}}

\newcommand{\kommentar}[1]{}

\newcommand{\identity}{\mathbbm{1}}

\newcommand*{\ditto}{---\texttt{"}---}

\newcommand{\forget}[1]{}

\begin{document}

\title{A link between symmetries of critical states and the structure of SLOCC classes in multipartite systems}

\author{Oskar Słowik$^1$}
\orcid{0000-0003-4138-3063}
\author{Martin Hebenstreit$^2$}
\orcid{0000-0002-5841-7082}
\author{Barbara Kraus$^2$}
\orcid{0000-0001-7246-6385}
\author{Adam Sawicki$^1$}
\orcid{0000-0003-4906-2459}

\affiliation{$^1$ Center for Theoretical Physics, Polish Academy of Sciences, 02-668 Warsaw, Poland\\%
$^2$ Institute for Theoretical Physics, University of Innsbruck, A–6020 Innsbruck, Austria}

\begin{abstract}
Central in entanglement theory is the characterization of local transformations among pure multipartite states. As a first step towards such a characterization, one needs to identify those states which can be transformed into each other via local operations with a non-vanishing probability. The classes obtained in this way are called SLOCC classes. They can be categorized into three disjoint types: the null-cone, the polystable states and strictly semistable states. Whereas the former two are well characterized, not much is known about strictly semistable states. We derive a criterion for the existence of the latter. In particular, we show that there exists a strictly semistable state if and only if there exist two polystable states whose orbits have different dimensions. We illustrate the usefulness of this criterion by applying it to tripartite states where one of the systems is a qubit. Moreover, we scrutinize all SLOCC classes of these systems and derive a complete characterization of the corresponding orbit types. We present representatives of strictly semistable classes and show to which polystable state they converge via local regular operators.
\end{abstract}

\maketitle

\section{Introduction}

Due to the relevance of multipartite entanglement in many areas of physics, entanglement theory has developed into an important research topic 
over the last decades \cite{NiCh09, HoHo09}. Multipartite entanglement does not only play an essential role in quantum communication \cite{HiBu99,Go00}, in quantum computation, e.g., measurement-based quantum computation \cite{RaBr01}, or quantum metrology (see e.g. \cite{GiLl11}), but it does also play a role in condensed matter physics. There, a connection between entanglement present in the wave function and quantum phase transitions can be drawn (for a review see e.g. \cite{AmFa08} and references therein). Moreover, tensor network methods such as matrix product states \cite{FaNa92,PeVe07,Vi03} or projected entangled pair states \cite{VeCi04,VeWo06} embody powerful numerical tools in condensed matter physics (see e.g. \cite{Or14} and references therein).

Entanglement theory is a resource theory, which becomes apparent in the scenario of spatially separated parties sharing a joint quantum state, an often considered scenario in quantum information theory. There, it is natural to restrict the allowed operations to Local quantum Operations assisted by Classical Communication (LOCC). Then, entanglement arises as a resource allowing to achieve certain tasks that are not possible by LOCC alone. Moreover, it is apparent that any state $\ket{\psi}$, which can be deterministically converted to another state $\ket{\phi}$ via LOCC, must be at least as entangled as $\ket{\phi}$ according to any reasonable measure of entanglement. The reason for that is that any task starting out with the state $\ket{\phi}$ could as well be implemented starting out with the state $\ket{\psi}$ instead  (transforming it to $\ket{\phi}$ as a first step). The partial order that LOCC imposes on the state space must thus be taken into account by any entanglement measure, which is also known as the LOCC-monotonicity condition for entanglement measures \cite{HoHo09}.

Studying multipartite entanglement is hence intimately connected to understanding which state transformations are possible via LOCC. A first step in studying whether a state transformation is possible via LOCC deterministically, is studying whether the state transformation is possible via stochastic LOCC (SLOCC), i.e., with a non-vanishing probability of success. Clearly, if a state transformation is not possible via SLOCC, it is not possible via LOCC. It has been shown that if one considers fully entangled states, i.e., states $\ket{\psi}$ for which all single-particle reduced density matrices $\rho_i = \tr_{1,\ldots, i-1, i+1, \ldots, N} \ket{\psi}\bra{\psi}$ have full rank (here, $N$ denotes the number of parties), then SLOCC is an equivalence relation \cite{DuVi00}. Mathematically stated, two $N$-partite states $\ket{\psi}$ and $\ket{\phi}$ are SLOCC-equivalent if and only if there exist regular operators $A_1, \ldots, A_N$ such that $\ket{\psi} = A_1 \otimes \ldots \otimes A_N \ket{\phi}$ \cite{DuVi00}. The state space thus partitions into so-called SLOCC classes. Another often considered class of operations are Local Unitaries (LUs) (for pure states see e.g. \cite{Kr10PRL,Kr10PRA}). Such operations are reversible and, hence, do not alter the entanglement of a given state.

For bipartite pure states, necessary and sufficient conditions for when LOCC transformations are possible have been derived \cite{Ni99}. Moreover, for such states, the concept of SLOCC classes is trivial in the sense that there exists exactly one SLOCC class containing all fully entangled states, i.e., states that have full Schmidt rank. 
All this is no longer true in the multipartite scenario. Deciding which state transformations are possible via LOCC becomes a notoriously difficult problem due to the intricate structure of multipartite LOCC protocols \cite{ChLe14,Ch11,ChHs17}. Moreover, considering the simplest multipartite quantum system, three qubits, there exist two distinct SLOCC classes (considering fully entangled states), the well-known GHZ- and W-class \cite{DuVi00}. Furthermore, already for systems as simple as four qubits, the Hilbert space partitions into an infinite number of SLOCC classes \cite{VeDe02}. Besides small system sizes, SLOCC classes have been studied for special types of states, such as permutation-symmetric states \cite{MaKr10,BaKr09,MiRo13}. The notion of SL-invariant polynomials has been used to distinguish certain SLOCC classes (see \cite{LuTh03,ViEl11,ElSi14} and references therein). These are polynomials $f$ in the coefficients of a state $\ket{\psi}$, which fulfill that $f(g_1 \otimes \ldots \otimes g_N \ket{\psi}) = f(\ket{\psi})$ for any $g_i \in SL(d_i,\mathbb{C})$; a complete set of such polynomials can be constructed \footnote{For first complete sets of SL-invariant polynomials see \cite{DuVi00} and \cite{LuTh03}.}\cite{GoWa13,OsSi05}.

Recently, a notable step in the characterization of state transformations has been taken for homogeneous systems (systems for which all local dimensions equal). In particular, it has been shown that in such systems, generically, no LOCC transformations are possible \cite{GoKr17,SaWa18}. Important in the proof was the fact that within such systems, states for which all single-particle reduced density matrices $\rho_i$ are proportional to the identity, so-called critical states, exist. A simple criterion that allows to decide whether a critical state exists for given local dimensions $d_1, \ldots, d_N$ has been derived in \cite{BrLe19}. 

It has been realized that there is an intrinsic connection between the existence of a critical state within an SLOCC class and the geometry of that SLOCC class. Due to the Kempf-Ness theorem, an SLOCC class contains a critical state if and only if it is closed (wrt. standard complex topology) \cite{KeNe79}. We call (states within) SLOCC classes containing a critical state \emph{polystable} \footnote{Such states are also known as balanced states \cite{OsSi10} and critical states are also known as stochastic states \cite{VeDe02}.}. It moreover holds that, within an SLOCC class, critical states are unique (up to LUs). Moreover, within polystable classes, the norm $||g_1 \otimes \ldots \otimes g_N \ket{\psi}||$, where $g_i \in SL(d_i,\mathbb{C})$, does attain its minimum for $g_1 \otimes \ldots \otimes g_N \ket{\psi}$ being a critical state \cite{KeNe79}. Polystable SLOCC classes have been a subject of interest in previous works \cite{Kl02,OsSi10,GoKr17,SaWa18,GoWa11,GoWa13,SaSc18}. It has been shown that such SLOCC classes can be distinguished by ratios of SL-invariant polynomials \cite{GoWa13}. SLOCC classes (states) for which all nonconstant homogeneous SL-invariant polynomials vanish, form the so-called \emph{null-cone}, such classes have been studied e.g. in \cite{NeMu84,WaDo13,SaOs14,MaSa15,MaSa18,JoEr14}. These SLOCC classes contain the $0$-vector in their closure and, obviously, cannot be distinguished with SL-invariant polynomials. SLOCC classes (states) that are not in the null-cone are called \emph{semistable}. Thus, the polystable classes are those semistable classes that do contain a critical state. We call (states within) the remaining classes strictly semistable. These SLOCC classes do not contain a critical state, however, they do contain a critical state in their closure \cite{KeNe79}.

Numerical tools which allow distinguishing between certain of the aforementioned types of SLOCC classes have been developed. In \cite{VeDe03}, an algorithm which transforms a state $\ket{\psi}$ into its so-called normal form has been presented. In case $\ket{\psi}$ is in the null-cone, this normal form is $0$. Otherwise, it coincides with the critical state within the SLOCC class of $\ket{\psi}$ (closure of the SLOCC class of $\ket{\psi}$), in case $\ket{\psi}$ is polystable (strictly semistable), respectively. In \cite{WaDo13, SaOs14}, it has been shown that a numerical algorithm, which follows the gradient of the sum of the linear entropies of the eigenvalues of $\rho_i$, also allows to determine the critical state within a polystable SLOCC class, or the critical state within the closure of a strictly semistable SLOCC class. Moreover, the latter method also allows to distinguish certain SLOCC classes within the null-cone \cite{NeMu84,SaOs14, SaOs12,MaSa15,MaSa18,SaMa18}.

Here, we study the existence of semistable states, i.e., states that are neither in the null-cone, nor SLOCC-equivalent to a critical state using tools from Geometric Invariant Theory (GIT). One of the main results of this article is a criterion for the existence of such states in terms of the dimensions of polystable SLOCC classes. Moreover, we study three-partite systems with local dimensions $2$, $m$, and $n$ in more details. SLOCC classes in such systems have been characterized in \cite{ChMi10} (see also  \cite{HeGa18}) with the help of the theory of matrix pencils \cite{Kr90}. We apply the derived criterion for the existence of strictly semistable states in the context of such systems. Moreover, we not only answer the question of whether strictly semistable states exist, but also derive a full characterization of the orbit types of all SLOCC classes within such systems.

The outline of the remainder of the paper is the following. First we introduce our notation and some preliminary results on SLOCC classes and the normal form of multipartite states \cite{VeDe03}. We will also recall how SLOCC classes of polystable states can be distinguished via SL-invariant polynomials \cite{GoWa13}. In Sec. \ref{Sec_main} we will focus on semistable states. Using tools from GIT, we will show that strictly semistable states exist if and only if there exist two polystable states whose orbits have different dimensions. In Sec. \ref{Sec_2nm} we illustrate the usefulness of this criterion by characterizing all SLOCC classes of tripartite states in $2\times m \times n$. Moreover, we will identify the representatives of strictly semistable states and will show to which polystable states they converge via SLOCC.

\section{Notation and Preliminaries}
\label{sec:not}
In this section we first introduce our notation and then review some important concepts utilized in the characterization of SLOCC classes.

We consider the Hilbert space, ${\cal H}=\C^{d_1}\otimes \C^{d_2}\otimes \ldots \otimes \C^{d_N}$ with arbitrary local dimensions $d_i\geq 2$. We are mainly interested in pure normalized states in $\mathcal{H}$. Hence, we consider the complex projective space $\mathbb{P}(\mathcal{H})$, which is obtained from $\mathcal{H}$ by identifying any two vectors which are proportional to each other via a non-zero complex number. For any vector $\ket{\phi} \in \mathcal{H}$ the corresponding quantum state is denoted by $[\phi] \in \mathbb{P}(\mathcal{H})$. Throughout the paper we will consider actions of appropriate Lie groups on both $\mathcal{H}$ and $\mathbb{P}(\mathcal{H}$). The action of a Lie group $H$ on $\cal{H}$ will be denoted by $H \acts \cal{H}$ and the action on $\mathbb{P}(\cal{H})$ by $H \acts \mathbb{P}(\cal{H})$, respectively. Given a Lie group $H$, the $H$--orbit of a vector $\ket{\Psi}\in {\cal H}$ is defined as
\be
H\ket{\Psi}=\left\{h|\Psi\rangle \Big| h\in H\right\}.
\ee

The stabilizer of a vector $\ket{\Psi}$ with respect to the group $H$ is defined as
\be
H_{\ket{\Psi}}=\left\{h\in H \Big| h\ket{\Psi} = \ket{\Psi}\right\}.
\ee
Any orbit $H\ket{\Psi}$ is an embedded submanifold of ${\cal H}$ which  is isomorphic to the left coset of $H_{\ket\Psi}$ in $H$, namely $H\ket{\Psi} \simeq H/H_{\ket{\Psi}}$. Hence, the larger the dimension of the stabilizer of the state, which is equal to the number of linearly independent generators of the stabilizing Lie subgroup, the smaller the dimension of the orbit. More precisely, the dimensions satisfy the formula
\be
\label{eq:orbitstabilizerdim1}
\mathrm{dim}\, H\ket{\Psi} + \mathrm{dim} \, H_{\ket{\Psi}} = \mathrm{dim} \, H .
\ee
The action of a Lie group $H$ on $\mathbb{P}(\mathcal{H})$ is induced by the action $H \acts \mathcal{H}$ via
\begin{gather}
h[\Psi] \coloneqq [h \ket{\Psi}],\,\,h\in H.
\end{gather}
Hence, the orbit through any state $[\Psi]$ is given by
\begin{gather}
H[\Psi] = [H \ket{\Psi}],
\end{gather}
i.e. the $H$--orbit of any $[\Psi]$ in $\mathbb{P}(\mathcal{H})$ is the projection of the orbit $H \ket{\Psi} \subset \mathcal{H}$. Thus, in contrast to $H[\Psi]$, the $H$--orbit of $\ket{\Psi}$, $H\ket{\Psi}$, might contain unnormalized states. The stabilizer of a state $[\Psi]$ with respect to the group $H$ is defined as
\be
H_{[\Psi]}=\left\{h\in H \Big| h\ket{\Psi} \propto \ket{\Psi}\right\}.
\ee
We also have that $H[\Psi] \simeq H/H_{[\Psi]}$ and
\begin{gather}
\label{eq:orbitstabilizerdim}
\mathrm{dim}\, H[\Psi] + \mathrm{dim} \, H_{[\Psi]} = \mathrm{dim} \, H.
\end{gather}

The stabilizer $H_{[\Psi]}$ is the set of symmetries of $[\Psi]$ in $H$. Note that $H_\Psi \subset H_{[\Psi]}$ for any $\Psi$ and any Lie group $H$. The two groups we will consider are
\begin{enumerate}
\item Special local unitary operators, i.e. elements of $K:=SU(d_1)\times SU(d_2) \times \ldots \times SU(d_N)$
\item  SLOCC operators, i.e. elements of $G:=SL(d_1)\times SL(d_2)\times \ldots \times SL(d_N)$.
\end{enumerate}
Both groups act in a natural way on $\mathcal{H}$ and $\mathbb{P}(\mathcal{H})$, respectively. The action of local unitary operators does not alter entanglement and it defines an equivalence relation. Two states are in the same LU-equivalence class if they belong to the same $K$--orbit in $\mathbb{P}(\mathcal{H})$. Hence, instead of considering all states in $\mathbb{P}(\mathcal{H})$, we will consider representatives of LU--equivalence classes. Similarly the $G$--orbit of a state, $[\Psi]$, is the SLOCC--class which contains $[\Psi]$ as an element. Two states, $[\Psi],[\Phi]$ are in the same SLOCC--class, i.e. they are SLOCC--equivalent, if there exist  $g_i \in SL(d_i)$ such that $[\ket{\Psi}]=[\bigotimes_{i=1}^n g_i \ket{\Phi}]$.

There is a qualitative difference between the actions of $K$ and $G$. The group $K$ is compact and all $K$--orbits are closed. On the other hand the group $G$ is not compact and $G$--orbits may not be closed. The closure of an orbit $\overline{G \ket{\Psi}}\subset\mathcal{H}$ in the standard complex topology in $\mathcal{H}$ can obviously be obtained by adding to $G \ket{\Psi}$ the limits of all sequences of vectors in $G \ket{\Psi}$. It is well known (see e.g. \cite{Wa17}), that the closure in standard topology coincides with the closure in Zariski topology. We refer the reader to \cite{Wa17} for more details on the Zariski topology and the relations between those two topologies. Let us just mention here that a set $X\in \C^m$ is Zariski closed if there exists a set of polynomials, $S\subset \C[x_1,\ldots,x_m]$ such that $X=\C^m (S)=\{x\in \C^m \mid s(x)=0 \forall s\in S\}$, where $\C[x_1,\ldots,x_m]$ denotes the polynomial ring. A Zariski closed set is called an algebraic variety. As mentioned above, the closures of $G \ket{\Psi}$ in $\mathcal{H}$ coincide in both topologies. It will become clear later on, which orbits $G[\Psi] \subset \mathbb{P}(\mathcal{H})$ are considered to be closed.

States which play a particularly important role in the classification of multipartite states are so--called critical states. A vector $\ket\Psi\in {\cal H}$ is called critical if all single-particle reduced states are completely mixed, i.e. $\rho_i\propto \one$ for all $i$. In the following we will denote the set of critical vectors by $Crit\subset\mathcal{H}$ and the set of critical states by $[Crit]\subset\mathbb{P}(\cal{H})$.

\subsection{Normal form of multipartite states}
\label{sub:normal}
As mentioned above, $G$-orbits in $\mathcal{H}$ may not be closed. In fact, the Kempf-Ness theorem \cite{KeNe79} states that $G\ket\Psi\subset\mathcal{H}$ is closed if and only if it contains a unique (up to LUs) critical vector $\ket{\Phi}\in Crit \cup \{0\}$ \footnote{ Note that a particular example of a closed $G$-orbit is the $G$-orbit containing the $0$-vector.}. Moreover, $\|g\ket{\Phi}\|\geq \|\ket\Phi\|$ for any $g\in G$.


Let us note that $G$-orbit $G \ket{\Psi}$ is open in its closure $\overline{G \ket{\Psi}}$, therefore the boundary of $G \ket{\Psi}$ can be defined as $\overline{G \ket{\Psi}} \setminus G \ket{\Psi}$ \footnote{See e.g. \cite{Br10} p. 4, Proposition 1.11.}.

The following lemma describes the structure of the boundary of an orbit and follows from well-known facts - about the dimensions of orbits in the boundary of orbit \footnote{See e.g. \cite{Br10} p. 4, Proposition 1.11.} and basic facts in GIT theory \footnote{See e.g. \cite{Mu03} p. 161, Theorem 5.3. and p. 162, Corollary 5.5.}.

\begin{lemma}
\label{boundary}
For each $\ket\Psi \in \mathcal{H}$, the boundary of its orbit, $\overline{G \ket\Psi} \setminus G \ket\Psi$, is a union of $G$--orbits of strictly smaller dimension. Each $G$-orbit closure contains a unique closed $G$--orbit. This closed orbit is the only orbit in the closure which has minimal dimension, i.e. the closed orbit has the smallest dimension of all orbits in the closure.
\end{lemma}

Starting from a vector $\ket\Psi\in\mathcal{H}$ one can look for
\begin{gather}
\inf_{\ket{\Phi^\prime}\in \overline{G \ket\Psi}}\|\ket{\Phi^\prime}\|.
\end{gather}
According to the above, one of the following must be the case.
\begin{enumerate}[(i)]
\item The infimum is $0$.
\item The infimum is not $0$, it is attained at some critical vector $\ket\Phi\in Crit$ within the orbit $G\ket{\Psi}$, and the orbit is closed. 
\item The infimum is not $0$ and it is attained at some critical vector $\ket\Phi\in Crit$, which is not in the orbit $G\ket{\Psi}$, but within its closure, $\overline{G\ket{\Psi}}$.
\end{enumerate}

In case (i) the vectors are said to be {\it unstable} (they form the so-called null-cone, which will be denoted by ${\cal N}$). For such vectors, there exist sequences $\{g_k\}\subset G$ such that $\lim_{k \rightarrow \infty} g_k \ket{\Psi} = 0$.
Vectors that are not unstable are called {\it semistable} and they form a set ${\mathcal{H}}_{ss} \coloneqq \cal{H} \setminus {\cal N}$. In case (ii) the vectors are called {\it polystable}; they are, hence, elements of an SLOCC class which contains a critical vector. In case (ii) one can distinguish two subcases: (ii a) - when the $G$-orbit of a polystable vector is of maximal dimension; (ii b) - when the $G$-orbit of a polystable state is not of maximal dimension. Note that maximal dimension refers here and in the following to the maximal dimension among all polystable orbits. In case (ii a), we call a vector {\it stable} \footnote{Notice that in the literature concerning GIT (see e.g. \cite{Do10} Chapter 8), the usual definition of stable vector is different - a vector is stable if its $G$-orbit is closed (i.e. the state is polystable) and has a finite stabilizer. Thus, our definition is much weaker since we only require the dimension of the stabilizer to be minimal (among all polystable orbits), not necessarily zero.} and in case (ii b), we call a vector {\it strictly polystable}. In case (iii), vectors are semistable but not polystable - we call such vectors {\it strictly semistable}. For such states, an infinite sequence $\{g_k\}\subset G$ is required in order to reach a critical state, i.e., $\lim_{k \rightarrow \infty} g_k \ket{\Psi} = \ket\Phi\in Crit$.
It has been shown that any strictly semistable vector is a superposition of a polystable vector and a vector from the null-cone (see the Supplemental Material for \cite{GoWa13} p. 1, Proposition 2).
Finally, a state in the projective space $[\Psi] \in \mathbb{P}(\mathcal{H})$ is called unstable/semistable/polystable etc. if it is the projection of some unstable/semistable/polystable vector $\ket{\Phi} \in \mathcal{H}$, respectively, i.e. $[\Psi]=[\Phi]$. We denote the set of semistable states by $\mathbb{P}(\mathcal{H})_{ss}$. In case of a $2 \times 2 \times 2$ system, the W--state, $\ket{W}= 1/\sqrt{3} (\ket{100}+\ket{010}+\ket{001})$ is a prominent example of a state in the null-cone (case (i)), whereas the critical GHZ--state $\ket{GHZ}= 1/\sqrt{2} (\ket{000}+\ket{111})$ is a stable state (case (ii a)). In order to give an example of strictly semi- and polystable states we need to consider systems of larger dimension. In case of a $2 \times 4 \times 4$ system, a state $\ket{\Psi_{sps}}=\frac{1}{2} \ket{0} (\ket{22}+ \ket{33}) +  \frac{1}{2} \ket{1} (\ket{00} + \ket{11})$ is strictly polystable (case (ii b) and a state $\ket{\Psi_{sss}}= \frac{1}{\sqrt{5}} \ket{0} (\ket{01}+ \ket{22} + \ket{33}) +  \frac{1}{\sqrt{5}} \ket{1} (\ket{00} + \ket{11})$ is strictly semistable (case (iii); cf. Section \ref{Sec_2nm} and Appendix \ref{app:tables244}). We illustrate the relationship between the types of states in Figure \ref{fig:nomenclature}.

Note that there is a numeric algorithm, which for any state $\ket{\Psi}$, determines a sequence of SL-operators, which converges, if applied to $\ket{\Psi}$, to the normal form as discussed above \cite{VeDe03}. If $\ket\Psi$ is in the null-cone, the normal form vanishes. If $\ket{\Psi}$ is polystable, then the normal form coincides with the critical state within the SLOCC class of $\ket\Psi$. Finally, if $\ket\Psi$ is strictly semistable, the algorithm converges to, but does not reach, a critical state in the closure of the SLOCC class of $\ket\Psi$ \footnote{Let us remark here that for system sizes we are considering here, numerical errors become relevant. Keeping track of the state by storing its computational basis coefficients may result in jumping to another SLOCC class due to precision errors. States that are in the null-cone may be erroneously identified as semistable, e.g., $\left|\psi_{26}\right>$ and  $\left|\psi_{27}\right>$ in Appendix \ref{app:255}. A more robust implementation can be achieved by keeping track of the accumulated local operators, instead.}.

The Kempf-Ness theorem \cite{KeNe79} is a powerful tool characterizing polystable states. Any SLOCC--class of a polystable state contains a unique critical state (up to LUs). It is known that $G [\it{ Crit}]$, is open, dense, and of full measure in $\mathbb{P}({\cal H})$ \cite{GoWa11, Wa17}, provided that $[Crit]$ is non-empty. We will use this fact in the proof of our main theorem which gives a criterion for the existence of strictly semistable states.

\begin{figure}
 \begin{center}
 \includegraphics[width=1.0 \columnwidth]{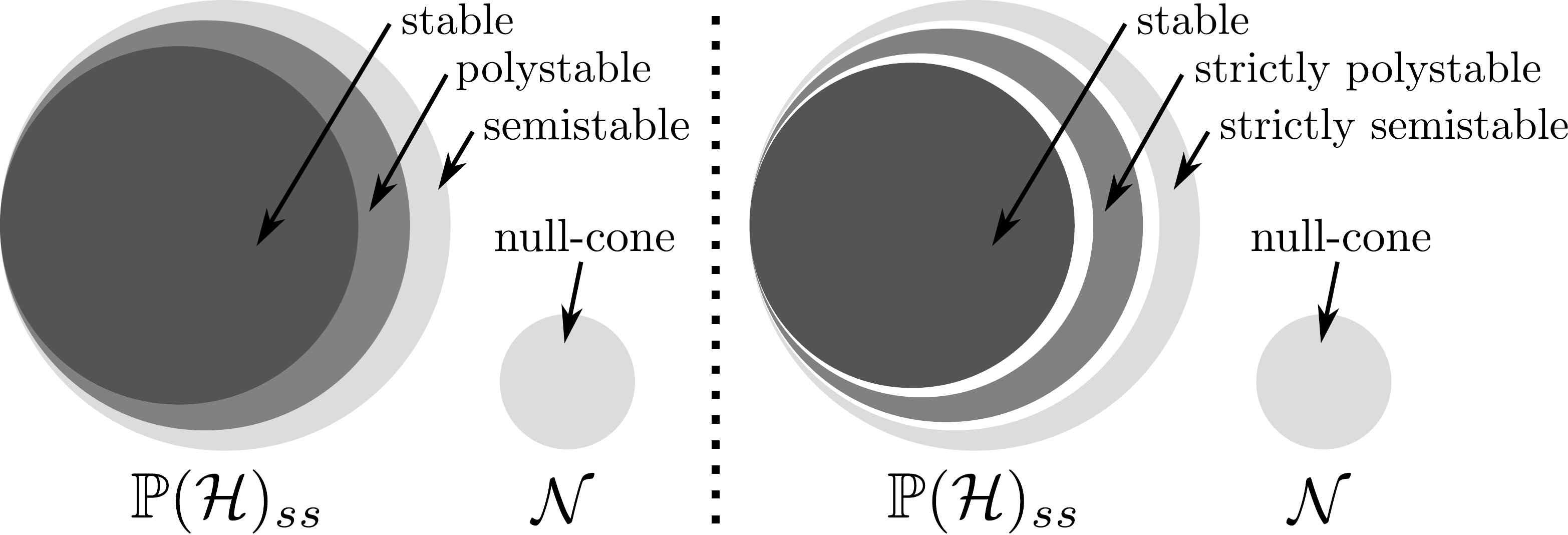}
 \caption{Set-theoretic relations between the types of states used in the paper (notice that e.g. the null-cone is not topologically separated from the semistable states, as this figure might suggest). The whole space of states $\mathbb{P}({\cal H})$ is divided into the semistable states $\mathbb{P}({\cal H})_{ss}$ and the null-cone $\mathcal{N}$. We call those semistable orbits, that are closed, polystable. We call those polystable orbits, that have maximal dimension (among all polystable orbits), stable. Moreover, we call polystable orbits, which are not stable, strictly polystable. Finally, we call semistable orbits, which are not polystable, strictly semistable. Thus, $\mathbb{P}({\cal H})$ partitions into four disjoint sets, the null-cone, the strictly semistable states, the strictly polystable states, and the stable states.
 \label{fig:nomenclature}}
 \end{center}
\end{figure}

\subsection{\texorpdfstring{$G$}{\textit{G}}-invariant polynomials}
\label{Sec_SLinvariant}

$G$--invariant polynomials constitute an important tool in the study of SLOCC--classes. A polynomial, $f:{\cal H}\longrightarrow \C$ is called $G$--invariant if $f(g\ket{\Psi})=f(\ket{\Psi})$ for any $\ket{\Psi}\in {\cal H}$ and any $g \in G$. The vector space (over $\C$) of all $G$--invariant polynomials is spanned by homogeneous polynomials of degree $k\in \mathbb{N}$. The set of $G$--invariant polynomials forms a ring which is finitely generated \cite{Mu03}.

As polynomials are continuous functions, any $G$--invariant polynomial will have the same value on any two $G$-orbits, $G\ket{\Psi_1}$ and $G\ket{\Psi_2}$ in $\mathcal{H}$, whose closures intersect. Hence,
as strictly semistable vectors converge to critical vectors, the SLOCC--class of a strictly semistable state cannot be differentiated from the one of the critical state, which is in its closure using $G$-invariant polynomials. Moreover, for any vector in the null-cone any of these polynomials vanish and vectors which are in the different SLOCC--classes within the null-cone cannot be distinguished via $G$--invariant polynomials \footnote{Note however, that some $G$-orbits in the null-cone can be distinguished by studying the gradient flow of the function \unexpanded{$-\sum_{i=1}^N\tr\{\rho_i([\Psi])^2\}$}, where $\rho_i([\Psi])$ is the $i$-th reduced state of $[\Psi]$ \cite{NeMu84,SaOs14, SaOs12,MaSa15,MaSa18}.}.
However, in \cite{GoWa13} it has been shown that SLOCC--classes of polystable states can be distinguished using ratios of $G$--invariant polynomials, as we recall in the following.

Two polystable states $[{\Psi}],[{\Phi}]$ are SLOCC equivalent if and only if \cite{GoWa13}
\bea
\frac{f_k(\ket{\Psi})}{h_k(\ket{\Psi})}=\frac{f_k(\ket{\Phi})}{h_k(\ket{\Phi})},\eea for any $G$--invariant homogeneous polynomials $f_k,h_k$ with $h_k(\ket{\Psi})\neq 0$ of degree $k$ for any $k$. As mentioned before, polystable states constitute a full measure set of states in case at least one critical state exists. Hence, if $[Crit] \neq \emptyset$, the criterion above solves the SLOCC--equivalence problem for generic states.

The necessary and sufficient condition for the existence of a critical state has been presented in \cite{BrLe19}. For arbitrary dimensional Hilbert spaces, a critical state exists if and only if
\begin{align}
\label{lmeformula}
\prod_{i=1}^N d_i - \sum_{l=1}^N (-1)^{l+1} \sum_{1 \leq i_1 < \ldots < i_l \leq N} \operatorname{gcd}(d_{i_1}^2, \ldots, d_{i_l}^2) \geq 0.
\end{align}
We will elaborate on that for the Hilbert spaces $\C^2 \otimes \C^m \otimes \C^n$, which we denote by $2\times m \times n$, in Section \ref{Sec_2nm}.

In summary, there are many important facts known about polystable states and also the null-cone has been investigated previously \cite{NeMu84,JoEr12,JoEr14,SaOs14,MaSa18}. However, little is known about strictly semistable states. In the subsequent section we will derive a simple necessary and sufficient condition for their existence. We will use this criterion for the simplest examples of multipartite systems, namely the Hilbert spaces $2\times m \times n$. There, we will not only identify the dimensions for which strictly semistable states exist, but we will scrutinize all SLOCC classes and present a complete characterization of them.

\section{The existence of strictly semistable states}
\label{Sec_main}
In this section we will show that there exists a strictly semistable state if and only if there exist two critical states whose $G$--orbits do not have the same dimensions (see Theorem \ref{Th_main}). Note that the latter means that there exist a strictly polystable state (i.e. not all polystable states are stable). To this end, we will first restrict ourselves to the set of semistable states. We will then introduce the definition of $c$--equivalence, which leads to a coarser classification than SLOCC--equivalence. Next, using tools from GIT we will prove the main theorem of the paper.

\subsection{The set of semistable states}
\label{sos}
Our aim is to derive a criterion for the existence of a strictly semistable state. Hence, we can restrict our attention to the Hilbert space excluding the null-cone, i.e. ${\cal H}\setminus {\cal N}$, or, more precisely, $\mathbb{P}({\cal H})_{ss}=\mathbb{P}({\cal H}\setminus {\cal N})$. However, we will first demonstrate that without excluding the null-cone, taking the closure in the projective space might become problematic. To give a simple example, let us consider the 3--qubit GHZ--state, $\ket{GHZ} = 1/\sqrt{2} (\ket{000}+\ket{111})$ and the regular matrix $G(\alpha)=\frac{1}{\sqrt{2}} \begin{pmatrix}e^\alpha & - e^\alpha \\ e^{-\alpha} & e^{-\alpha} \end{pmatrix}$ with $\det G(\alpha) = 1$. Clearly, the limit $\lim_{\alpha \to\infty} G^{\otimes 3}(\alpha) \ket{GHZ}$ does not exist in $\mathcal{H}$. However, in the projective space we obtain $\lim_{\alpha \to\infty} [G^{\otimes 3}(\alpha) \ket{GHZ}]=[\ket{W}] \in [{\cal N} \setminus\{0\}].$ To avoid such situations from now on we will restrict ourselves to the analysis of semistable states so the full (topological) space of quantum states is $\mathbb{P}(\mathcal{H})_{ss}$. This restriction is justified by the following observation which we prove in Appendix \ref{app:closuredifference}.

\begin{observation}
\label{obs:closuredifference}
For any vector $\ket{\Psi} \in \mathcal{H}_{ss}$, we have
\bea \overline{[\;G \ket{\Psi} \;]} \setminus [\;\overline{G \ket{\Psi}}\;] \subset [{\cal N} \setminus\{0\}].\eea
\end{observation}

Hence, considering $\mathbb{P}(\mathcal{H})_{ss}$ as the full (topological) space we have that the orbit $G[\Psi]$ is closed (in $\mathbb{P}(\mathcal{H})_{ss}$)  if and only if $G\ket\Psi$ is closed in $\mathcal{H}_{ss}$. Moreover, for any pair of states $[\Phi], [\Psi]\in \mathbb{P}(\mathcal{H})_{ss}$ it holds that $[\Phi]\in\overline{G[\Psi]}$ if and only if $\ket{\Phi}\in \overline{G\ket{\Psi}}$. Thus if we treat $\mathbb{P}(\mathcal{H})_{ss}$ and $\mathcal{H}_{ss}$ as full topological spaces, we have a well defined correspondence between closures in $\mathcal{H}_{ss}$ and $\mathbb{P}(\mathcal{}H)_{ss}$. In the following we will always consider $\mathbb{P}(\mathcal{H})_{ss}$ and $\mathcal{H}_{ss}$ as full topological spaces. 

Due to the results summarized in Sec. \ref{sec:not} we have that the closure of a $G$--orbit in $\mathbb{P}(\mathcal{H})_{ss}$ always contains a unique critical state (up to LUs). It might also contain strictly semistable states of various SLOCC classes, which converge to this critical state. In order to derive a criterion for the existence of these strictly semistable states, we adapt now Lemma \ref{boundary} for states in $\mathbb{P}(\mathcal{H})_{ss}$. To this end, we use the following observation, which follows from Corollary 10 of \cite{GoWa11}.

\begin{observation}
\label{orbitdim}
For any $\ket\Psi\in\mathcal{H}_{ss}$ it holds that $$\mathrm{dim} \, G\ket{\Psi}=\mathrm{dim} \,G[\Psi].$$ That is, the dimension of the $G$--orbit in the projective space coincides with the dimension of the $G$--orbit in the Hilbert space for any semistable state. 
\end{observation}

\begin{proof}
As we consider here, in contrast to \cite{GoWa11}, the projective space, we adapt the proof of \cite{GoWa11} to our setting. Of course $G_{\ket{\Psi}}\subset G_{[\Psi]}$. Assume next that there is $g\in G_{[\Psi]}$ such that $g\ket{\Psi}=\alpha\ket \Psi$. When $|\alpha|< 1$ this implies that $\lim_{n\rightarrow\infty}\|g^n\ket\Psi\|=0$ and when $|\alpha|> 1$ it means $\lim_{n\rightarrow\infty}\|g^{-n}\ket\Psi\|=0$. In both cases this is contradiction as $\ket\Psi$ is semistable. Thus the only possibility is $\alpha=e^{i\phi}$. We will show next that $\phi$ belongs to a finite set. Let $\{p_1,p_2,\ldots,p_k\}$ be a set
of homogeneous generators of $G$-invariant polynomials, $\mathrm{deg} \, p_j=d_j$. Then 
\begin{gather}
p_{j}(\ket\Psi)=p_{j}(g\ket\Psi)=p_{j}(e^{i\phi}\ket\Psi)=e^{id_j\phi}p_{j}(\ket\Psi).
\end{gather}
Thus $\forall j \, d_j\phi=0\,\mathrm{mod}\,2\pi$. This means that $G_{[\Psi]} / G_{\ket{\Psi}}$ is a finite set and hence $\mathrm{dim} \, G\ket{\Psi}=\mathrm{dim} \, G[\Psi]$ \footnote{Notice that in fact it suffices to consider any chosen $G$--invariant homogeneous polynomial.}.
\end{proof}

Due to the relation given in Eqs. (\ref{eq:orbitstabilizerdim1}) and (\ref{eq:orbitstabilizerdim}) the above statement is equivalent to the statement that the dimensions of the stabilizers coincide. That is, the dimension of the set of $SL$ operators, which leave the vector $\ket{\Psi}$ invariant, $G_{\ket{\Psi}}$, coincides with the dimension of the set of $SL$ operators, which leave the vector $\ket{\Psi}$ up to a proportionality factor invariant, $G_{[\Psi]}$.

Combining now Lemma \ref{boundary} with the Observations \ref{obs:closuredifference} and \ref{orbitdim} leads to the following  adaptation of Lemma \ref{boundary} for states in $\mathbb{P}(\mathcal{H})_{ss}$.
\begin{lemma}
\label{boundary1}
For each $[\Psi] \in \mathbb{P}(\mathcal{H})_{ss}$, the boundary of an orbit $\overline{G [\Psi]} \setminus G [\Psi]$  in $\mathbb{P}(\mathcal{H})_{ss}$ is a union of $G$--orbits in $\mathbb{P}(\mathcal{H})_{ss}$ of strictly smaller dimension.  Each $G$-orbit closure contains a unique polystable $G$--orbit. This polystable orbit is the only orbit in the closure which has minimal dimension.
\end{lemma}

\subsection{c--equivalence}
\label{ceq}
The number of SLOCC-classes, i.e. $G$--orbits, is in general infinite \cite{DuVi00}. Furthermore, as mentioned before, $G$-orbits may not be closed. This motivates the introduction of closure-equivalence ($c$-equivalence) which is used in GIT. By the Theorem of Nagata and Mumford it is known that $c$--equivalence can be defined using three equivalent definitions (see \cite{Mu03}). As we are concerned with semistable states, it will suffice to consider the following definition (see also Figure \ref{fig:fig1}).

\begin{definition}\label{def:ceq}
Two states $[\Psi],[\Phi] \in \mathbb{P}(\mathcal{H})_{ss}$ are called $c$--equivalent if \bea \overline{G [\Psi]}\cap \overline{G [\Phi]} \neq \emptyset.\eea
\end{definition}

We note that $c$--equivalence is indeed an equivalence relation. 
As the closure of a semistable state contains a critical state, which does not belong to its $G$--orbit, we have that $c$--equivalence induces a coarser classification than SLOCC-equivalence. As explained in Sec. \ref{Sec_SLinvariant}, $c$--classes can be distinguished using the ratios of SL--invariant polynomials. Each $c$--class contains exactly one closed $G$--orbit. Due to Lemma 2 this $G$--orbit is the only orbit of minimal dimension (in the closure of the $G$--orbit). This closed orbit lies in the closure of every $G$--orbit from the $c$--class \cite{Mu03} and is the one mentioned in Lemma \ref{boundary1}.  Moreover, due to the Kempf-Ness theorem \cite{KeNe79}, each $c$--class is uniquely characterized by the unique (up to LU) critical state. Hence, $c$--classes are in 1-1 correspondence with  critical states (up to LU) (see Figure \ref{fig:fig1}).

In order to state these facts mathematically rigorously, several definitions would be in order. As these definitions are cumbersome, we do not state them here, but refer the reader to Appendix \ref{app:regularmap}. We consider the set of classes $\mathbb{P}(\mathcal{H})_{ss}/\sim_c$ given by the $c$--equivalence relation in $\mathbb{P}(\mathcal{H})_{ss}$. It turns out that the set $\mathbb{P}(\mathcal{H})_{ss}/\sim_c$ can be equipped with the structure of a projective algebraic variety denoted by $\mathbb{P} (\mathcal{H})_{ss} // G$. The GIT construction gives a regular (see Definition \ref{reg}), hence continuous mapping $\pi : \mathbb{P} (\mathcal{H})_{ss} \rightarrow \mathbb{P} (\mathcal{H})_{ss} // G$ called the GIT quotient \cite{Br10}. Thus, $\pi([\Psi])$ is a point in a quotient variety corresponding to the $c$--class of $[\Psi]$ and $\pi([\Psi])=\pi([\Phi])$ if and only if $[\Psi]$ and $[\Phi]$ are $c$--equivalent. The Kempf-Ness theorem gives us the bijection $\mathbb{P} (\mathcal{H})_{ss} // G \cong [Crit]/K$. Hence, the GIT quotient maps a semistable state $[\Psi]$ to the unique (up to LUs) critical state in $\overline{G [\Psi]}$.  

In what follows we call a $c$--class trivial if it is a singleton, i.e. if it consists of exactly one (thus closed/polystable) $G$--orbit.  We will say that the $c$--equivalence in $\mathbb{P}(\mathcal{H})_{ss}$ is trivial if each $c$--class (if there is any) is trivial, i.e. the GIT quotient does not identify any two distinct $G$--orbits or there are no critical states (so no $c$--classes). Otherwise the $c$--equivalence is called non-trivial (see Figure \ref{fig:fig1}). For example, one may check that in case of $2 \times 3 \times 5$ system, the left hand side of Formula \ref{lmeformula} is negative, what indicates that there are no critical states in this system (cf. Example \ref{nolme} from Appendix F). Hence, $c$--equivalence in $2 \times 3 \times 5$ system is trivial. Stated differently, the $c$--equivalence in $\mathbb{P}(\mathcal{H})_{ss}$ is trivial if each $c$--class is also a  SLOCC--class, i.e. $c$--equivalence is not a coarser classification than SLOCC. Note that a $c$--class is non-trivial if and only if there exists a strictly semistable state whose $G$--orbit belongs to this class.
Similarly the $c$--equivalence is non-trivial if and only if there exists a strictly semistable state (given by some non-trivial $c$--class).

\begin{figure}
 \begin{center}
 \includegraphics[scale=0.35]{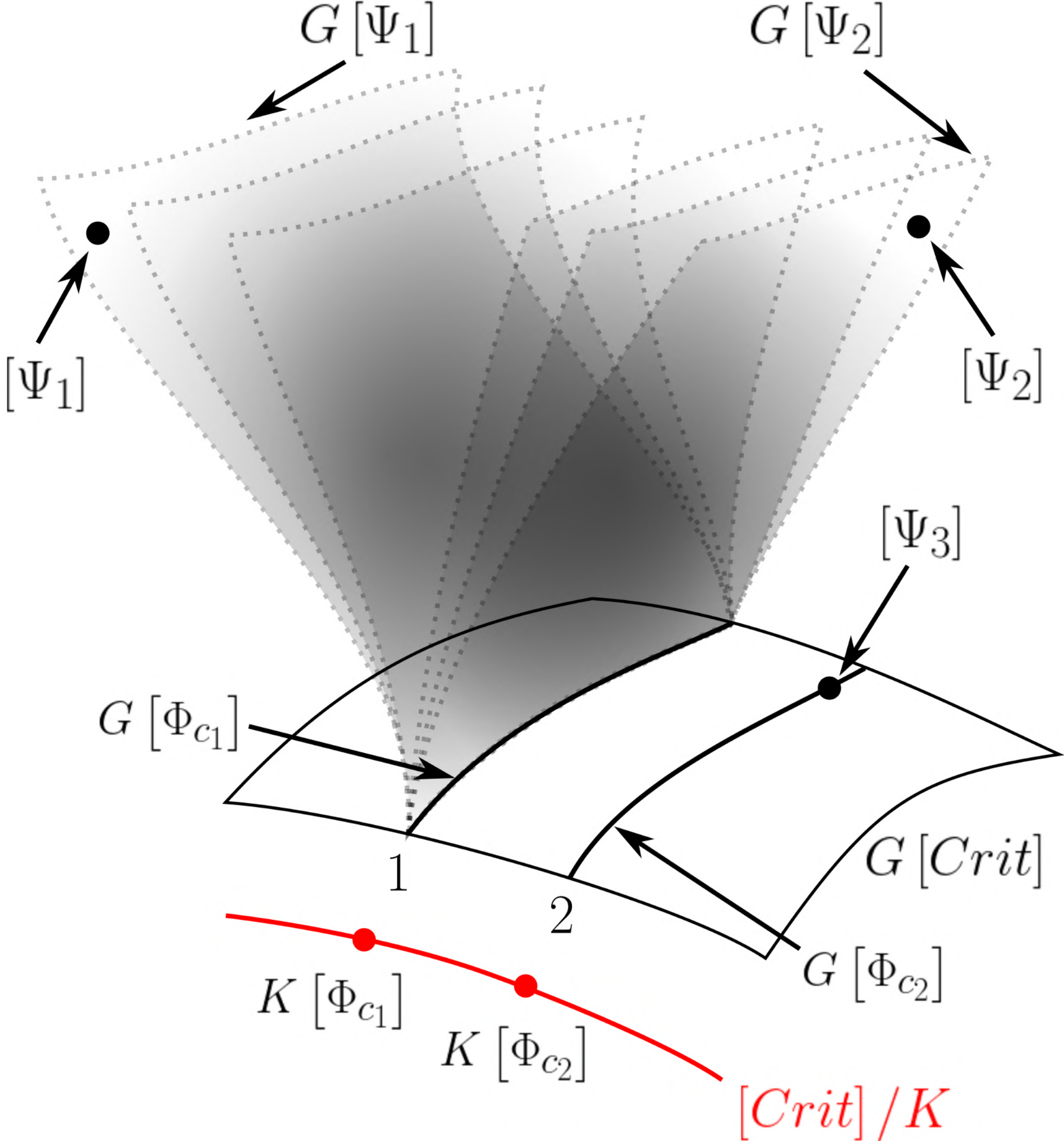}
 \caption{Example of the partition of semistable states into $c$-classes parametrized by LU--classes of critical states. Two sample $c$--classes 1 and 2 have been drawn. Points of the red line are LU-classes of critical states ($K$--orbits). For each such point we have a unique closed/polystable $G$--orbit containing it (black horizontal lines) which is also the unique orbit of smallest dimension in a given class (see Lemma \ref{boundary1}). Closed/polystable orbits correspond one-to-one to $c$-classes of $G$-orbits. Non-closed orbits (which are of higher dimension) are denoted by dashed gray surfaces. Class 2 contains only polystable states. Class 1 contains strictly semistable states and are witnesses of the non-triviality of $c$--equivalence (see Definition \ref{def:ceq}). States $[\Psi_1]$ and $[\Psi_2]$ are strictly semistable and belong to the $c$--class 1. Hence,  they can be transformed in the limit to the critical states of the LU--class $K[\Phi_{c_1}]$ (normal form). The stabilizer of this critical state has larger dimension, i.e. the $G$--orbit has smaller dimension than the one of the corresponding strictly semistable state. State $[\Psi_3]$ is polystable and can be transformed by SLOCC--operations to a critical state in $K[\Phi_{c_2}]$ (normal form). As we will show later (see Observation \ref{obs:localmainth}), the dimension of $G[\Phi_{c_2}]$ must be maximal and hence larger than $\mathrm{dim} \, G[\Phi_{c_1}]$. \label{fig:fig1}}
 \end{center}
\end{figure}

\subsection{Stability of the action of a group}
\label{chapter2}

We briefly discuss here the fact that, in case critical states exist, the action of $G$ on $\mathbb{P}(\mathcal{H})$ is stable, which is required in the proof of the main theorem of this paper. Let us first introduce the definition of stability of an action of a group.
\begin{definition}
\label{stabmax}
An action $H \acts \mathbb{P}(\mathcal{H})$ is stable if there exists an open dense set $U \subset \mathbb{P}(\mathcal{H})$ that is a union of polystable orbits.
\end{definition}

We show in Appendix \ref{app:stab_eq} that the existence of an open dense set $U \subset \mathbb{P}(\mathcal{H})$ that is a union of polystable orbits implies the existence of an open dense set $V \subset U$ that is a union of polystable orbits of maximal dimension (among all orbits). In particular, $V$ is a union of stable orbits so in case the action is stable,
the maximal dimension among polystable orbits (i.e. the dimension of stable orbits) coincides with the maximal dimension among all orbits, which we will denote in the following by $d_G$.

The action $G \acts \mathbb{P}(\mathcal{H})$ is stable if and only if critical states exist, as the union of the SLOCC--orbits of critical states (i.e. polystable orbits) is an open dense subset
$U \subset \mathbb{P}(\mathcal{H})$, whenever $U \neq \emptyset$ \cite{GoWa11}. In particular, we have \footnote{The reason for that is that the dimension of an open and dense set is the same as the dimension of the whole space.}
\begin{equation} \label{eq:gendim}
\mathrm{dim} \, G [\it{Crit}]=\mathrm{dim} \, \mathbb{P} (\mathcal{H}).
\end{equation}

Let us note here that $G [\Psi]$ having maximal dimension does not necessarily mean that $G_{[\Psi]}$ is finite, i.e. that the dimension of the local stabilizer is zero. An example would be the $G$--orbit of the  3--qubit GHZ state, which is of maximal dimension, as it is of full measure. However, its dimension does not coincide with the dimension of $G$ since it admits continuous symmetries and its stabilizer has a positive dimension \cite{BrLe19} \footnote{More precisely, we have that $\operatorname{dim} G [GHZ] = 7$ and $\operatorname{dim}G_{[GHZ]}= 2$ \cite{VeDe02ghz}.}. However, we have that $\mathrm{dim} \, G [GHZ]=\mathrm{dim} \, \mathbb{P}(\mathcal{H})$, as the $GHZ$--state is the only critical state for three qubits.

\subsection{Existence of strictly semistable states}
\label{sec:mainsubsection}

In this subsection, we prove the main theorem (Theorem \ref{Th_main}), which provides a necessary and sufficient condition for the existence of strictly semistable states and discuss its consequences. 

In order to do so, we will restate a lemma from \cite{Mu99}, which we will apply later in the context of the GIT quotient $\pi : \mathbb{P}(\mathcal{H})_{ss} \rightarrow  \mathbb{P}(\mathcal{H})_{ss}//G$ and which is a direct consequence of Corollary 15.5.4 from \cite{TaYu15}. 
\begin{lemma}[On the fibres of the mapping \cite{Mu99}]
\label{fibres}
Let $X$ and $Y$ be quasi-projective irreducible varieties and $\pi: X \rightarrow Y$ a regular surjective map. Then, there exists an open and dense set $V \subset Y$, such that for any points $\alpha \in V$ and $\gamma \in Y \setminus V$, we have
\begin{equation} \label{eq:otfotm}
\mathrm{dim } \, \pi^{-1} (\alpha)=s \leq \mathrm{dim } \, \pi^{-1} (\gamma) \end{equation}
and $s=\mathrm{dim} \, X - \mathrm{dim} \, Y$.
\end{lemma}

We defer a more detailed discussion of the assumptions of the lemma as well as proper definitions of quasi-projective varieties, irreducible varieties, and regular maps to Appendix \ref{app:regularmap}. However, for our purposes here, note that Lemma \ref{fibres} is applicable to the GIT quotient $\pi: \mathbb{P}(\mathcal{H})_{ss} \rightarrow \mathbb{P}(\mathcal{H})_{ss}//G$. Let us briefly outline the reason for that. First, $\mathbb{P}(\mathcal{H})_{ss}$ and $\mathbb{P}(\mathcal{H})_{ss}//G$ are indeed irreducible quasi-projective varieties (see Appendix \ref{app:regularmap}). Moreover, $\pi$ is a regular map, which follows from the fact that a finite number of rational polynomials can be used to decide the $c$-class of any given semistable state $[\psi]$, as explained in Section \ref{ceq}.

 The sets $\pi^{-1} (\alpha)$ in Lemma \ref{fibres} are called fibres and the dimension $s$ is called the minimal fibre dimension. Notice that $\pi^{-1}(V)$ is an open and dense subset of $\mathbb{P}(\mathcal{H})_{ss}$ and generic fibres have minimal fibre dimension $s$.
We are now in the position to prove the main theorem, which presents a criterion for when a $c$--equivalence is trivial (i.e. all GIT quotient fibres are single $G$--orbits).

\begin{theorem}[ \textbf{Main Theorem} ]
\label{Th_main}
There exist no strictly semistable state if and only if the dimensions of all polystable $G$--orbits are the same (i.e. they are all stable).
\end{theorem}

\begin{proof} First, note that if there exists no critical state, then there also does not exist a strictly semistable state and thus the statement holds trivially.

Let us now assume that there does exist a critical state. As mentioned before, we denote by $d_G$ the maximal dimension among all $G$--orbits. As discussed in Section \ref{chapter2}, $d_G$ coincides with the maximal orbit dimension among all polystable $G$--orbits, i.e., the dimension of stable orbits. We prove both directions by contradiction. 

\textit{If:} If there exists a strictly semistable state, $[\Psi]$, we have that $\overline{G [\Psi]} \setminus G [\Psi] \neq \emptyset$ contains a polystable $G$--orbit of strictly smaller dimension than $\operatorname{dim} G [\Psi]$ (Lemma \ref{boundary1}). Therefore, there has to exist a polystable state whose orbit is of strictly smaller dimension than $d_G$, i.e., there has to exist a strictly polystable state.

\textit{Only if:}
 As explained earlier in this subsection, $\pi : \mathbb{P}(\mathcal{H})_{ss} \rightarrow  \mathbb{P}(\mathcal{H})_{ss}//G$ fulfills the conditions of Lemma \ref{fibres}. Let us first show that the minimal fibre dimension of $\pi$, $s$, equals the maximal orbit dimension, $d_G$.
To this end, let $U$ be the union of all stable orbits.
Suppose that $G [\Psi] \subset U$. Obviously $\mathrm{dim} (G [\Psi])=d_G$. Next, by Lemma \ref{boundary1} there exists no state $[\Phi]\notin G[\Psi]$ such that  $G [\Psi] \subset \overline{G [\Phi]}$ as this would imply that $\mathrm{dim} (G [\Phi])> \mathrm{dim} (G [\Psi])$, which contradicts the fact that $d_G$ is the maximal dimension. This means that the $c$--class of any $[\Psi] \in U$ is given by $G [\Psi]$, i.e. for each $[\Psi] \in U$ we have $\pi^{-1}(\pi([\Psi]))=G [\Psi]$. In particular, for each $\alpha \in \pi (U)$ we have $\mathrm{dim} \, \pi^{-1} (\alpha) = d_G$. Recall that $U$ is open and dense. Let $V$ be the open and dense set from Lemma \ref{fibres}. From continuity of the GIT quotient $\pi$ we have that $\pi^{-1} (V)$ is an open and dense set so $O \coloneqq \pi^{-1} (V) \cap U \neq \emptyset$, as two open and dense sets need to intersect. Thus, we have that $\pi(O) \cap V \neq \emptyset$ and therefore $s=d_G$ [recall that $s$ is the minimal fibre dimension from Eq (\ref{eq:otfotm})]. Hence, the dimension of any $c$--class must be greater than or equal to $d_G$.

We are now in the position to prove that if there exists a polystable $G$--orbit of smaller than maximal dimension (i.e. a strictly polystable orbit), then there exists a strictly semistable state. Suppose that there exists a $[\Phi_c] \in [Crit]$ such that $\mathrm{dim} \, G [\Phi_c] < d_G$. From the fact that the minimal fibre dimension is $d_G$ we know that $ \mathrm{dim} \, \pi^{-1} (\pi([\Phi_c])) \geq d_G$. Hence, the fibre $\pi^{-1} (\pi([\Phi_c]))$ must contain another $G$--orbit (which cannot be polystable due to the Kempf-Ness theorem), i.e. there must exist a strictly semistable state.
\end{proof}

As the dimension of the orbits is directly related to the dimension of the stabilizer (see Eq. (\ref{eq:orbitstabilizerdim})) and as the $c$--equivalence is trivial iff there exists no strictly semistable state, the following statements are equivalent to Theorem \ref{Th_main}. 

\begin{enumerate}[(i)]
\item The $c$-equivalence is non-trivial if and only if there exists a critical state whose stabilizer has strictly larger dimension than the stabilizer of a generic state.
\item There exists a strictly semistable state if and only if there exists a polystable state whose $G$--orbit has smaller than the maximal dimension $d_G$.
\item The existence of strictly semistable states is equivalent to the existence of strictly polystable states.
\end{enumerate}

Moreover, note that the proof of Theorem \ref{Th_main} can also be applied to individual $c$--classes, as stated in the following 

\begin{observation}
\label{obs:localmainth}
\leavevmode
The $c$--class of a given polystable orbit $G[\Phi_c]$ is non-trivial if and only if  $\mathrm{dim} \,G[\Phi_c]$ is smaller than the dimension of a generic orbit (i.e. $[\Phi_c]$ has more continuous symmetries). If this is the case then the dimension of the $G$--orbit of every strictly semistable state in such a $c$--class is larger than the dimension of $G[\Phi_c]$ (see also Figure \ref{fig:fig1}).
\end{observation}

Let us remark here that Theorem \ref{Th_main} can be extended to other groups whose action is stable and $G$ is the complexification of $K$, i.e. $G=K^{\mathbb{C}}$. 
In the subsequent section we demonstrate the usefulness of this criterion for Hilbert spaces of dimension $2\times m\times n$.

\section{The categorization of SLOCC classes of states in \texorpdfstring{$2\times m\times n$}{2 x \textit{m} x \textit{n}}}
\label{Sec_2nm}

In this section, we apply the derived criterion (Theorem \ref{Th_main}) for the existence of strictly semistable states to three partite systems in a  Hilbert space $\mathbb{C}^2 \times \mathbb{C}^m \times \mathbb{C}^n$ ($2 \leq m \leq n$). We show that strictly semistable states exist (and thus, c-equivalence is non-trivial) if and only if $m=n$ and $m \geq 4$. Then, we go beyond that and provide a full characterization of orbit types of all SLOCC classes, i.e., $G$-orbits, for $2 \times m \times n$ systems.

The outline of this section is as follows.
First, we will briefly recall the characterization of SLOCC classes in $2\times m\times n$ systems using the theory of matrix pencils \cite{ChMi10} in Subsection \ref{sec:pencils}.
Then, in Subsection \ref{sec:geometric}, we will recall a characterization of critical states in $2\times m\times n$  systems obtained in \cite{BrLe19}.
In Subsection \ref{sec:2mncequivalence}, we apply Theorem \ref{Th_main} to $2\times m\times n$ systems.
In Subsection \ref{sec:2mncharacterization}, we perform a complete characterization of the orbit types of SLOCC classes in $2\times m\times n$ systems, which we summarize in form of a flowchart in Subsection \ref{sec:flowchart}. Finally, we provide tables of SLOCC classes in $2\times m \times n$ systems equipped with orbit types for small system sizes (up to $2\times 5 \times 5$) in Appendix \ref{app:tables}.

\subsection{SLOCC-classes of states in \texorpdfstring{$2\times m\times n$}{2 x \textit{m} x \textit{n}}}
\label{sec:pencils}
The SLOCC classification in $2\times m\times n$ systems was obtained via the Kronecker Canonical Form (KCF) of linear matrix pencils \cite{ChMi10}. As we will use this characterization, we recall here briefly the idea. The reader is referred to \cite{Ga59,Kr90,ChMi10,HeGa18} for the definitions used in the categorization of matrix pencils.
A linear $m \times n$ matrix pencil $\mathcal{P}$ is a homogeneous matrix polynomial of degree 1 in variables $\mu$ and $\lambda$,
\begin{align}
\label{eq:pencil}
\mathcal{P}=\mu R+\lambda S,
\end{align}
where $R$ and $S$ are $m \times n$ matrices.
The matrix pencil associated to an arbitrary state $\ket{\psi}$ in $\mathbb{C}^2 \times \mathbb{C}^m \times \mathbb{C}^n$, with
\begin{align}
	\ket{\psi} &=  \ket{0}_A\ket{R}_{BC} + \ket{1}_A\ket{S}_{BC}  \nonumber \\
	&=\left[ \ket{0}_A (R \otimes \identity) + \ket{1}_A (S \otimes \identity) \right] \ket{\phi^+_n}_{BC},
\end{align}
where $\ket{\phi^+_n} = 1/\sqrt{n} \sum_{i=0}^{n-1} \ket{i i }$, is given by Eq. (\ref{eq:pencil}). Let us now recall how the matrix pencil corresponding to a state is transformed under SLOCC operators $A \otimes B \otimes C$. Let 
\begin{align}
\label{eq:A}
A=\begin{pmatrix} \alpha & \beta \\ \gamma & \delta \end{pmatrix},
\end{align}
then the matrix pencil transforms as
\begin{align}
\mathcal{P} &=\mu R+\lambda S \rightarrow \mathcal{P'}, \text{ where}\\
\mathcal{P'} &= B [(\alpha \mu + \gamma \lambda) R +  (\beta \mu + \delta \lambda)S] C^T.
\end{align}

A normal form of matrix pencils, the KCF, under regular operators $B$ and $C^T$ as in Eq. (\theequation) has been derived in \cite{Kr90}. It is always possible to find operators $B$ and $C^T$ which transform a matrix pencil into its KCF.
In the following, we only consider fully entangled states, i.e., states for which all single-particle reduced density matrices have full rank. Restricting to matrix pencils that correspond to fully entangled states, the KCF of a matrix pencil is the (generalized) block-diagonal form
\begin{align}
\label{eq:kcf}
\mathcal{P} = \operatorname{blockdiag} \left\{ L_{\epsilon_1},\dots, L_{\epsilon_a}, L_{\nu_1}^T,\dots, L_{\nu_b}^T,J	\right\},
\end{align}
where
\begin{equation}
\label{eq:lblock}
L_{\epsilon}=\begin{pmatrix}\lambda & \mu &  &  & \\
 & \lambda & \mu &   &  \\
 &  & \ddots & \ddots &  \\
 &  &  & \lambda & \mu
\end{pmatrix}
\end{equation}
has size $\epsilon \times (\epsilon + 1)$ \footnote{Here and in the following, matrix elements that are not displayed are $0$.} and
\begin{align}
\label{eq:jblock}
J= \operatorname{blockdiag}\{&M^{e_1^1}(x_1),\dots,  M^{e_r^l}(x_l), \nonumber\\
                                  & \qquad \qquad N^{e_1^\mu},N^{e_2^\mu},\dots, N^{e_r^\mu} \},
\end{align}
where the pairwise different $x_i$ are called the finite eigenvalues of a matrix pencil, $e_j^i$ are called size signatures of the eigenvalue $x_i$, and $l$ denotes the number of (distinct) eigenvalues. Moreover, an eigenvalue may occur several times in Eq. (\theequation) and $m_i = \sum_j e^i_j$ is called the multiplicity of the eigenvalue $x_i$. $N$-blocks correspond to an infinite eigenvalue (which we denote by $\infty$). The blocks take the form
\begin{align}
\label{eq:mblock}
N^{e_j^\mu}&=\begin{pmatrix}\mu & \lambda &  &  & \\
 & \mu & \lambda &  & \\
 &  & \ddots & \ddots &\\
 &   &  & \mu & \lambda\\
 &  &  &  & \mu
\end{pmatrix} \text{ and}\nonumber\\
M^{e_j^i}(x_i)&=\begin{pmatrix}x_i \mu + \lambda & \mu &  &  & \\
 & x_i \mu + \lambda & \mu &  &  \\
 &   & \ddots & \ddots & \\
 &   &  & x_i \mu + \lambda & \mu\\
 &  &  &  & x_i \mu + \lambda
\end{pmatrix}.
\end{align}
The sizes of these matrices are $e^\mu_j \times e^\mu_j$ and $e_j^i \times e_j^i$, respectively.
The KCF of a matrix pencil is invariant under the operators $B$, $C$, which makes it a useful tool for studying SLOCC classes. However, the operator $A$ [given in Eq. (\ref{eq:A})] changes the finite eigenvalues according to
 \begin{align}
   x_i \rightarrow \left\{\begin{array}{lr}
        \frac{\alpha x_i + \beta}{\gamma x_i + \delta}, & \text{if } \gamma x_i + \delta \neq 0\\
       \infty, & \text{if } \gamma x_i + \delta = 0\\
        \end{array}\right.,
 \end{align}
and the infinite eigenvalue according to
  \begin{align}
   \infty \rightarrow \left\{\begin{array}{lr}
        \frac{\alpha}{\gamma}, & \text{if } \gamma \neq 0\\
       \infty, & \text{if } \gamma = 0\\
        \end{array}\right..
 \end{align}
Other than that, $A$ does not change the KCF of a matrix pencil \footnote{In particular, the occurrence of $L$- and $L^T$-blocks as well as their sizes are unchanged, as are the size signatures of the eigenvalues.}. Note that with an operator $A$ it is always possible to set three of the distinct eigenvalues to arbitrary new (distinct) values. Moreover, it is always possible to bring a matrix pencil into a form without infinite eigenvalues \cite{ChMi10}. Therefore, in the following, we only consider finite eigenvalues for convenience. Degeneracies of eigenvalues cannot be changed with any finite $A$. However, we will see later on, it is indeed possible to create degeneracy when considering the limits of sequences of $SL$-operators.

All SLOCC classes of states in $2 \times m \times n$ can be characterized by the KCF of the corresponding matrix pencils \cite{ChMi10}. Moreover, the problem of listing all SLOCC classes becomes a combinatorial problem of listing all possible KCFs. Let us restate here the according theorem from \cite{ChMi10}.
\begin{theorem}[\cite{ChMi10}]
 \label{theo:fractional}
Two  $2 \times m \times n$-dimensional pure states $\ket{\psi}$ and $\ket{\phi}$ for which their corresponding matrix pencils have only finite eigenvalues $\{x_i\}$ and $\{x_i'\}$, respectively, are SLOCC equivalent if and only if the KCFs of the matrix pencils agree up to a linear fractional transformation of the eigenvalues,
\begin{equation}
\frac{\alpha x_i+ \gamma}{\beta x_i+\delta}=x_i',
\end{equation}
for some $\alpha, \beta, \gamma, \delta \in \mathbb{C}$, where $\alpha \delta - \beta \gamma \neq 0$.
\end{theorem}

\subsection{Critical states}
\label{sec:geometric}
The question of the existence of a critical state has been answered for general multipartite quantum systems in \cite{BrLe19}. For $2 \times m \times n$ systems, at least one SLOCC class containing a critical state exists if and only if $m=n$ or $n-m$ divides $m$ \cite{BrLe19}.
Moreover, in $2 \times m \times (m+1)$ systems, there is a unique (up to LUs) critical state \cite{BrLe19}, which is given by
\begin{align}
\ket{\psi_{2,m,m+1}} = \frac{1}{\sqrt{m+1}} \sum_{k=0}^{m-1} &\left( \sqrt{\frac{m-k}{m}} \ket{0, k, k } \right. \nonumber\\
                                                      &\left. \quad  + \sqrt{\frac{k+1}{m}} \ket{1, k, (k+1)} \right).
\end{align}
More generally, in all $2 \times m \times n$ systems for which $(n-m)$ divides $m$, there is a unique (up to LUs) critical state which is LU-equivalent to \cite{BrLe19}
\begin{align}
 \ket{\psi_{2,\frac{m}{n-m},\frac{m}{n-m}+1}}_{A\,B_1\, C_1} \otimes \ket{\phi^+_{n-m}}_{B_2\,C_2},
\end{align}
where the system held by the second (third) party is composed of the subsystems $B_1$ and $B_2$ ($C_1$ and $C_2$), respectively.
The corresponding matrix pencil is $\mathcal{P} = \bigoplus_{i=1}^{(n-m)} L_{\frac{m}{n-m}}$.

Moreover, for $2 \times n \times n$ systems, critical states are (up to LUs) of the form \cite{BrLe19}
\begin{align}
\label{eq:crit}
\frac{1}{\sqrt{n}} \sum_{i=1}^n \left(e^{i \phi_i} \sin(\theta_i/2) \ket{0} + \cos(\theta_i/2) \ket{1} \right)\ket{i-1, i-1}, \nonumber \\
= \ket{0} (D_0 \otimes \identity) \ket{\phi^+_n} + \ket{1} (D_1 \otimes \identity) \ket{\phi^+_n},
\end{align}
where $D_0 = 1/\sqrt{n} \operatorname{diag}(e^{i \phi_1} \sin(\theta_1/2), \ldots )$ and  $D_1 = 1/\sqrt{n} \operatorname{diag}(\cos(\theta_1/2), \ldots )$
for $(\theta_i, \phi_i)$ such that
\begin{align}
\sum_i \cos(\theta_i) &= 0 \nonumber\\
\sum_i \sin(\theta_i) \cos(\phi_i) &= 0 \nonumber\\
\sum_i \sin(\theta_i) \sin(\phi_i) &= 0.
\end{align}
The critical state is unique in cases $n=2$ and $n=3 $ \cite{BrLe19}. Note that it is possible to associate a geometrical meaning to the conditions in Eq. (\theequation). The pairs $(\theta_i,\phi_i)$ can be interpreted as spherical coordinates of unit vectors
\begin{align}
\label{eq:spherical}
\vec{v}_i = \begin{pmatrix}
\sin \theta_i \cos \phi_i \\
\sin \theta_i \sin \phi_i\\
\cos \theta_i
\end{pmatrix}
\end{align}
and the state in Eq. (\ref{eq:crit}) is critical if and only if the vectors $\vec{v}_i$ sum to $\vec{0}$ \cite{BrLe19}.

Note that matrix pencils corresponding to the states given in Eq. (\ref{eq:crit}) are diagonal, i.e. direct sums of $M^1(x_i)$ blocks and the eigenvalues of the matrix pencil, $x_i$, are related to the unit vectors $\vec{v}_i$ by \cite{BrLe19}
\begin{align}
x_i = e^{i \phi_i} \tan{\frac{\theta_i}{2}}.
\end{align}

Recall, however, that the operator on the first system ($A$) may change the eigenvalues of the matrix pencil \footnote{Operators $B$ and $C$ do not change the eigenvalues, though, as the KCF of a matrix pencil and thus the eigenvalues are invariant under $B$ and $C$.}. It will become important later on to understand how the vectors $\vec{v}_i$ are transformed under such an operator $A$. Let us consider the singular value decomposition, $A = V \operatorname{diag}\left(\sqrt{\frac{1}{1-\alpha}}, \sqrt{1-\alpha} \right) U$, where $\alpha \in [0,1)$. Writing $U$ in Euler decomposition and absorbing the third ($\sigma_z$) rotation into $V$ (as the rotation commutes with the diagonal matrix), we have $U = e^{- i\frac{\Theta}{2} \sigma_y} e^{- i \frac{\Phi}{2} \sigma_z}$. We are now interested in how this transformation changes the eigenvalues and thus the vectors $\vec{v}_i$. Note that we will disregard $V$, as it is only a LU transformation of the corresponding states.

Let us first consider the application of $A$ in case $U = V = \identity$, i.e., we only have the transformation $\operatorname{diag}\left(\sqrt{\frac{1}{1-\alpha}}, \sqrt{1-\alpha} \right)$, which acts on all vectors simultaneously as
\begin{align}
\label{eq:hinging}
f_{\alpha}: [0, \pi]\times [0, 2\pi] &\rightarrow [0, \pi]\times [0, 2\pi] \nonumber\\
(\theta_i, \phi_i) &\mapsto ( 2 \arctan \left(\frac{1}{1-\alpha} \tan(\theta_i / 2) \right), \phi_i).
\end{align}
 Note that---under this non-linear transformation $f_\alpha$---$\theta_i$ increases strictly with $\alpha$ for all $i$, unless $\theta_i =0$ or $\theta_i = \pi$. In the latter cases $\theta_i$ is left invariant. The transformation can be pictorially understood as hinging all $\vec{v}_i$ (except those pointing in direction $(0,0,1)^T$ or $(0,0,-1)^T$) down towards the direction $(0,0,-1)^T$ without changing their azimuthal angle, $\phi_i$.

Let us now consider the full operator $A$ also taking the unitary $U$ into account. It can be easily seen that unitary transformations by $A$ correspond to $SO(3)$ rotations of the $\vec{v_i}$ \cite{BrLe19}. Taking the $SO(3)$ operation stemming from $U$ into account, the vectors $\vec{v}_i$ are simply rotated \emph{before} hinging them down by $f_\alpha$ according to Eq. (\ref{eq:hinging}) along the $z$-axis. Equivalently, the $SO(3)$ operation stemming from $U$ can also be understood as choosing an axis $(\Theta, \Phi)$ along which the ``hinging''-transformation $f_\alpha$ is performed (on static vectors). 

\subsection{c-equivalence and existence of strictly semistable states}
\label{sec:2mncequivalence}
In this subsection, we will apply Theorem \ref{Th_main} to $2 \times m \times n$ systems to characterize those $m$, $n$ for which strictly semistable states exist.

Recall that in the cases $m \neq n$, or $m=n \leq 3$ \cite{BrLe19} there exists either no critical state, or a unique critical state (up to LUs). Combining this fact with Theorem \ref{Th_main} we obtain the following corollary.
\begin{corollary}
In $2 \times m \times n$ systems, there exists no strictly semistable state if $m \neq n$, or $m=n \leq 3$.
\end{corollary}

In the following we will show that in all other cases, i.e. $2\times m \times n$ for $m=n \geq 4$, strictly semistable states exist and thus c-equivalence is non-trivial. To this end, it suffices to identify two critical states whose dimensions of the stabilizer group differ (see Theorem \ref{Th_main}).

It can be easily verified that the (complex) dimension of the stabilizer of states $\ket{\psi}$ of the form given in Eq. (\ref{eq:crit}) is given in terms of the degeneracies of the eigenvalues of the corresponding (diagonal) matrix pencil and reads
\begin{align}
\dim(G_{\ket{\psi}}) = \max\{3-l, 0\} + \left( -1 + \sum_{i=1}^l m_i^2 \right),
\end{align}
where $l$ denotes the number of distinct eigenvalues and $m_i$ their degeneracies. To this see this, note that symmetries of the form $\identity \otimes B \otimes C$ for diagonal matrix pencils must satisfy $C^T = B^{-1}$ as it must hold that $B \identity C^T = \identity$. Moreover, $B$ must commute with $\operatorname{blockdiag}\{x_1 \identity_{m_1}, \ldots, x_l \identity_{m_l}\}$. Hence, the number of free parameters is $ -1 + \sum_{i=1}^l m_i^2 $, which gives rise to the second summand on the righthand side of Eq. (\theequation). Moreover, any operator $A$ that gives rise to a symmetry must act in such a way that all sets of eigenvalues with coinciding multiplicity are mapped into themselfs. This gives rise to some finite freedom in choosing a permutation of the eigenvalues in an appropriate way. Then, the image of three distinct complex numbers uniquely determines an operator $A$, while any number less then three allows free complex parameters in $A$ as in the first summand on the righthand side of Eq. (\theequation).

It is clear that $\dim(G_{\ket{\psi}})$ is always minimized by $m_i=1 \ \forall i$, i.e., by having all distinct eigenvalues, as summarized in the following observation.

\begin{observation}
\label{obs:degeneracy}
The dimension of the SLOCC orbit of a state corresponding to a diagonal matrix pencil with all distinct eigenvalues is strictly larger than the orbit dimension of a state corresponding to a diagonal matrix pencil with any degeneracy in the eigenvalues.
\end{observation}

Using this observation, it is easy to prove the following theorem.

\begin{theorem}
\label{thm:2mnsemistable}
In $2 \times m \times n$ systems, there exist striclty semistable states if and only if $m=n \geq 4$.
\end{theorem}
In other words, c-equivalence is non-trivial  if and only if $m=n \geq 4$.
\begin{proof}
For $m=n\geq 4$, there exist critical states corresponding to diagonal matrix pencils that do and those which do not have degenerate eigenvalues. This guarantees the existence of strictly semistable states according to Theorem \ref{Th_main}. This can be easily seen in the geometrical picture reviewed in Section \ref{sec:geometric}. For $n \geq 2$, it is always possible to find $n$ distinct unit vectors in $\mathbb{R}^3$ with multiplicity $m_i = 1$, which sum to $\vec{0}$. Moreover, for $m=n\geq 4$, it is always possible to construct one vector $\vec{v}_1$ with $m_1 = 2$ as well as up to $n-2$ additional unit vectors with multiplicities summing up to $n-2$, such that the total weighted sum of vectors yields $\vec{0}$.
\end{proof}

\subsection{Characterization of orbit types}
\label{sec:2mncharacterization}

In this subsection, we go beyond Theorem \ref{Th_main} and derive a full characterization of the orbit types of SLOCC classes in $2 \times m \times n$ systems with the help of matrix pencils. More precisely, we provide a systematic way of checking properties of a matrix pencil (corresponding to a $2 \times m\times n$ state) in order to conclude the orbit type of the corresponding SLOCC class. Let us remark here that for some choices of $m$ and $n$, the considerations above as well as Theorem \ref{thm:2mnsemistable} (see also \cite{BrLe19}) already give the full characterization, e.g., cases in which $m \neq n$, or $m=n=3$. However, in case $m=n \geq 4$, the considerations above guarantee the existence of semistable states, but do not yet give a full characterization of the orbit type of all SLOCC classes for the considered system sizes.
In the following, we will thus employ two observations and a theorem which will allow extending Theorem \ref{thm:2mnsemistable} to a full characterization of the orbit types.
Remarkably, we will show that the orbit type does not depend on the eigenvalues of the matrix pencil, but only on their multiplicities. If this had not been the case, it would have seemed unlikely that a self-contained characterization of orbit types is possible.
In the remainder of the section, we will wlog always consider matrix pencils in KCF and states such that the corresponding matrix pencil is in KCF. As mentioned before, this is possible as operators $B$ and $C^T$ which transform the matrix pencil to KCF always exist.

Let us start by considering the case $m \neq n$. In this case, Theorem \ref{thm:2mnsemistable} states that there do not exist strictly semistable states. Moreover, as discussed in Section \ref{sec:geometric}, critical states do only exist if $n-m$ divides $m$ \cite{BrLe19}. Furthermore, if $n-m$ does divide $m$, there exists a unique stable SLOCC class corresponding to the matrix pencil $\mathcal{P} = \bigoplus_{i=1}^{(n-m)} L_{\frac{m}{n-m}}$. Thus, for $n \neq m$, there is either exactly one stable SLOCC class or no stable SLOCC class, which implies that there neither exist any strictly polystable, nor strictly semistable classes. Thus, all SLOCC classes except the aforementioned stable class are in the null-cone. Hence, the characterization of orbit types is already complete for the case $m \neq n$.

In order to tackle the open cases (within $2 \times m \times n$ systems such that $n=m$) let us continue with an observation about matrix pencils containing an $L_\epsilon$- and an $L^T_\nu$-block.

\begin{observation}
\label{obs:epsilon_nu_null-cone}
An SLOCC class that corresponds to a matrix pencil whose KCF contains an $L_\epsilon$ and an $L^T_\mu$ block is in the null-cone.
\end{observation}
\begin{proof}
Let us denote the  matrix pencil (in KCF) representing the SLOCC class in question by $\mathcal{P}$ and the corresponding state by $\ket{\psi}$. We will prove the observation by explicitly constructing a sequence of SL-operators of the form $\identity \otimes B_\alpha \otimes C_\alpha$ such that $\lim_{\alpha \rightarrow \infty} \identity \otimes B_\alpha \otimes C_\alpha \ket{\psi} = 0$. Wlog let us assume an ordering of the blocks such that the matrix pencil has the form $\mathcal{P} = L_\epsilon \oplus L^T_\nu \oplus \tilde{\mathcal{P}}$, where $\tilde{\mathcal{P}}$  is of size $[m- (\epsilon +\nu +1)] \times [n- (\epsilon +\nu +1)]$ and contains the remaining blocks, whose form does not matter.

We define the $(\epsilon+\nu+1) \times (\epsilon+\nu+1)$ matrices
\begin{align}
B_\alpha' &= \sum_{i=0}^{\epsilon + \nu} e^{\left( i - \frac{\epsilon + \nu}{2} \right) \alpha}  \ket{i}\bra{i} \text{ and} \\
C_\alpha' &= \sum_{i=0}^{\epsilon + \nu}  e^{\left(  \frac{\epsilon + \nu}{2} - i \right) \alpha}  \ket{i}\bra{i}
\end{align}
and the $n \times n$ matrix
\begin{align}
C_\alpha'' = e^{-\frac{\alpha}{2(n-1)}} \identity_{n} + e^{\frac{n \alpha}{2(n-1)}} \ket{\epsilon}\bra{\epsilon}.
\end{align}
The matrices of interest are then defined as $B_\alpha =  B_\alpha' \oplus \identity_{m-(\epsilon + \nu + 1)} $ and $C_\alpha =C_\alpha'' \left( C_\alpha' \oplus \identity_{n-(\epsilon + \nu + 1)} \right)$. It can be easily verified that $\det(B_\alpha) = \det(C_\alpha) = 1$. Moreover, $\lim_{\alpha \rightarrow \infty} B_\alpha \mathcal{P} C_\alpha^T = 0$. Thus, also $\lim_{\alpha \rightarrow \infty} \identity \otimes B_\alpha \otimes C_\alpha \ket{\psi} = 0$. This proves that any such SLOCC class is in the null-cone.
\end{proof}

The intuition behind the construction is the following. The purpose of $B'$ and $C'$ is to damp the off-diagonal entries of the $L_\epsilon \oplus L^T_\nu$ part in $\mathcal{P}$, leaving the diagonal entries invariant. Thus, the whole $\epsilon$th column of $\mathcal{P}$ is damped, i.e., multiplied by $e^{-\alpha}$, as in this column, only a non-diagonal entry is present. The purpose of $C_\alpha''$ is to damp all the remaining columns through multiplying them with $e^{-\frac{\alpha}{2(n-1)}}$ at the cost of boosting the $\epsilon$th column through multiplying it by $e^{\alpha/2}$. It can be easily seen that $B_\alpha' \otimes C_\alpha'$ damps the entry in the $\epsilon$th column faster than it is boosted by $C_\alpha''$. Hence, overall, the whole matrix pencil $\mathcal{P}$ is damped.

Let us remark here, the also for the previously considered case ($m \neq n$), it can be analogously seen that matrix pencils which contain an $L_{\epsilon_1}$- and an $L_{\epsilon_2}$-block with $\epsilon_1 \neq \epsilon_2$ correspond to SLOCC classes in the null-cone (the same holds for $L^T$-blocks of different size). Moreover, the same holds true for matrix pencils that contain an $M$-block (with arbitrary eigenvalue and arbitrary size $e$) as well as either an $L_\epsilon$- or an $L^T_\nu$-block. Explicit constructions of sequences of SL-operations that, in the limit, map states corresponding to such matrix pencils to 0 are provided in Appendix \ref{app:nullcone2mn}.

For $m=n$, however, Observation \ref{obs:epsilon_nu_null-cone} already allows going beyond previous results. There, whenever an $L$-block is present, it must be accompanied by an $L^T$-block (and vice versa) due to the dimension of the matrix pencil. Hence, Observation \ref{obs:epsilon_nu_null-cone} applies and shows that any class corresponding to a matrix pencil which contains such blocks is in the null-cone. It thus only remains to classify matrix pencils that are solely composed of $M$-blocks. The following observation reduces the problem to characterizing diagonal matrix pencils.

\begin{observation}
\label{obs:diagonalreduction}
Let $C$ be an SLOCC class that corresponds to a KCF $\mathcal{P} = \bigoplus_{i,j} M^{e^i_j}(x_i)$ with at least one $e^i_j > 1$. Consider the matrix pencil  $\mathcal{P'} = \bigoplus_{i}  \bigoplus_{j'=1}^{m_i} M^{1}(x_i)$, which is obtained from $\mathcal{P}$ by deleting all non-diagonal elements of $\mathcal{P}$ (recall that $m_i = \sum_j e^i_j$). Then, the following holds.
\begin{enumerate}[(i)]
\item If the SLOCC class corresponding to $\mathcal{P'}$ is semistable, then $C$ is strictly semistable.
\item If the SLOCC class corresponding to $\mathcal{P'}$ is in the null-cone, then $C$ is in the null-cone.
\end{enumerate}
\end{observation}
\begin{proof}
 Let us denote the state corresponding to $\mathcal{P}$ by $\ket{\psi}$. As before, we will prove the two statements by explicitly constructing a sequence of SL-operators of the form $\identity \otimes B_\alpha \otimes C_\alpha$ such that $\lim_{\alpha \rightarrow \infty} \identity \otimes B_\alpha \otimes C_\alpha \ket{\psi} \propto \ket{\psi'}$, where by $\ket{\psi'}$ we denote the state corresponding to $\mathcal{P'}$.
To this end consider the $n \times n$ matrices
\begin{align}
B_\alpha &= \sum_{i=0}^{n-1} e^{\left( i - \frac{n-1}{2} \right) \alpha}  \ket{i}\bra{i} \text{ and} \\
C_\alpha &= \sum_{i=0}^{n-1}  e^{\left(  \frac{n-1}{2} - i \right) \alpha}  \ket{i}\bra{i}.
\end{align}
It can be easily verified that $\det(B_\alpha) = \det(C_\alpha) = 1$ and that $\lim_{\alpha \rightarrow \infty} B_\alpha \mathcal{P} C^T_\alpha = \mathcal{P'}$.

Let us first proof statement (i).
Due to the considerations above, we have $G \ket{\psi'} \subseteq \overline{G \ket{\psi}}$. Moreover, as the size signatures of the eigenvalues of $\mathcal{P}$ and of $\mathcal{P'}$ are different, $\ket{\psi}$ and $\ket{\psi'}$ are not in the same SLOCC class. Thus, $G \ket{\psi'} \subseteq \overline{G \ket{\psi}} \setminus G \ket{\psi}$ and therefore the latter set is not empty and $\ket{\psi}$ is strictly semistable.

Let us now prove statement (ii). $\ket{\psi'}$ being in the null-cone implies that all SL-invariant polynomials vanish on $\ket{\psi'}$. However, as before we have that $\det(B_\alpha) = \det(C_\alpha) = 1$ and $\lim_{\alpha \rightarrow \infty} B_\alpha \mathcal{P} C^T_\alpha = \mathcal{P'}$. Thus, it also holds that all SL-invariant polynomials vanish on $\ket{\psi}$, i.e., $\ket{\psi}$ is in the null-cone.
\end{proof}

Due to the results above, it remains to characterize the SLOCC classes corresponding to diagonal matrix pencils. In order to obtain such a characterization, let us first state a lemma, which we prove in Appendix \ref{app:geometric}.

\begin{lemma}
\label{lemma:geometric}
Let $\{\vec{v}_i\}_{i=1}^l$ be a set of $l$ distinct unit vectors in $\mathbb{R}^3$ with associated multiplicities $m_i \in \mathbb{N}$ sorted in descending order such that $\sum_{i=1}^l m_i = n$ for some $n \geq 3$ and $m_1 < \sum_{i =2}^l m_i$. Let $(\theta_i, \phi_i)$ denote the spherical coordinates of $\vec{v}_i$, i.e.,
\begin{align}
\vec{v_i} = \begin{pmatrix}
\sin \theta_i \cos \phi_i \\
\sin \theta_i \sin \phi_i\\
\cos \theta_i
\end{pmatrix}.
\end{align}
Moreover, let $f_\alpha$ be a (non-linear) transformation acting on the spherical coordinates of the vectors as
\begin{align}
f_{\alpha}: [0, \pi]\times [0, 2\pi] &\rightarrow [0, \pi]\times [0, 2\pi] \nonumber\\
(\theta_i, \phi_i) &\mapsto   ( 2 \arctan \left(\frac{1}{1-\alpha} \tan(\theta_i / 2) \right), \phi_i),
\end{align}
where $\alpha \in [0,1)$. Then, it is always possible to transform the vectors $\{\vec{v}_i\}_{i=1}^l$ to new set of vectors $\{\vec{v'}_i\}_{i=1}^l$ such that $\sum_i m_i \vec{v'}_i = 0$ through some simultaneous $SO(3)$ rotation of the vectors $\vec{v}_i$ followed by $f_{\alpha}$.
\end{lemma}

Note that the operation considered in the lemma is exactly how an operator $A$ (up to a final local unitary) transforms a state whose corresponding matrix pencil in KCF has eigenvalues that correspond to the vectors $\{\vec{v}_i\}_{i=1}^l$  (as in the geometrical picture discussed in Section \ref{sec:geometric}) into a new state whose corresponding matrix pencil in KCF has eigenvalues that correspond to the vectors $\{\vec{v'}_i\}_{i=1}^l$. Let us remark here, that this lemma applied in our context shows that diagonal matrix pencils whose eigenenvalues with degeneracies $m_i$ are such that every $m_i$ is strictly smaller than the sum of the remaining $m_j$, correspond to SLOCC classes that contain a critical state. With the help of this lemma, we can now perform the characterization of SLOCC classes corresponding to diagonal matrix pencils as shown in the following theorem. Note, that as we only consider fully entangled states, matrix pencils that we consider have at least two distinct eigenvalues (otherwise, the state would be bi-separable).

\begin{theorem}
\label{thm:diagtype}
Let $C$ be an SLOCC class that corresponds to a diagonal matrix pencil of size $n \times n$ with $l \geq 2$ distinct eigenvalues $x_i$. Let $m_i$ denote the degeneracies of the eigenvalues $x_i$ in descending order. Then, the following holds.
\begin{enumerate}[(i)]
\item If $m_1 > \sum_{i=2}^l m_i$, then $C$ is in the null-cone.
\item If $l = 2$ and $m_1 = m_2$, then $C$ is strictly polystable (except when $n=2$, which corresponds to the 3-qubit GHZ state, which is stable).
\item If $l > 2$ and $m_1 = \sum_{i=2}^l m_i$, then $C$ is strictly semistable.
\item Otherwise, i.e., if $l>2$ and $m_1 < \sum_{i=2}^l m_i$, $C$ is polystable. Moreover, $C$ is stable if and only if all eigenvalues are distinct ($l = n$) and $C$ is strictly polystable if and only if there is any degeneracy ($l < n$).
\end{enumerate}
\end{theorem}
\begin{proof}
We will first prove statements (i), (ii), and (iii) by explicit construction of sequences of $SL$-operators transforming the state of interest into $0$, a critical state, or asymptotically into a critical state, respectively. Utilizing Lemma \ref{lemma:geometric}, we will then show statement (iv).

In order to prove statement (i), let us first apply an operator $A$, which brings the state to a form such that the eigenvalue with highest degeneracy is $x_1 = 0$. Then the matrix pencil is of the form
\begin{align}
\label{eq:kcfi}
\mathcal{P} =
\begin{pmatrix}
 \lambda &   &   &  &  \\
   &  \ddots   &   &  &  \\
   &   &  \lambda &  &  \\
   &   &   &    \lambda + x_2 \mu & \\
   &   &   &    & \ddots
\end{pmatrix}.
\end{align}
Let us denote the corresponding state by $\ket{\psi}$, which is given by $\ket{\psi} = \ket{0} (\operatorname{diag}(0, \ldots, 0, x_2, \ldots) \otimes \identity) \ket{\phi^+_n} + \ket{1}\ket{\phi^+_n}$. Note that $n-m_1 < m_1$ due to the assumption. Consider the matrices
\begin{align}
A_\alpha &= \operatorname{diag}(e^{\alpha},e^{-\alpha})   \text{ and}\\
B_\alpha &=  e^{\frac{2 (n-m_1)}{n} \alpha} \identity_{m_1} \oplus e^{-\frac{2 m_1}{n} \alpha } \identity_{n-m_1}.
\end{align}
Then it holds that $\det(A_\alpha) = \det(B_\alpha) = 1$ and moreover $\lim_{\alpha \rightarrow \infty} A_\alpha \otimes B_\alpha \otimes \identity \ket{\psi} = 0$, which proves that the corresponding SLOCC class is in the null-cone.

Let us now prove statement (ii). In this case there are exactly two eigenvalues and these have the same multiplicity. It is possible to apply an operator $A$ such that the matrix pencil takes the form $\mathcal{P} = \lambda \identity_{n/2} \oplus \mu \identity_{n/2}$, which is easy to see as up to three eigenvalues of a matrix pencil can be chosen freely by an operator $A$ \cite{ChMi10}. It can be easily verified that the corresponding state is critical, i.e., all single-particle reduced density matrices are proportional to $\identity$. Thus, the SLOCC class in question is polystable. Moreover, unless $n=2$, the multiplicity of the eigenvalues is larger then 1. In Observation \ref{obs:degeneracy} we have seen that this implies that the dimension of the orbit is strictly smaller than the maximal dimension. This implies that, unless $n = 2$, the considered SLOCC class is strictly polystable. In case $n=2$, the class in question is the 3-qubit GHZ class, which is stable.

Let us now prove statement (iii). As in the proof of statement (i) above, let us first bring the matrix pencil to a form such that $x_1 = 0$. The matrix pencil is then of the form as in Eq. (\ref{eq:kcfi}). Let us denote the corresponding state (after the operation) by $\ket{\psi}$.
Considering now the matrices
\begin{align}
A_\alpha &= \operatorname{diag}(e^{\alpha},e^{-\alpha})   \text{ and}\\
B_\alpha &= \left(\prod_{i=2}^l x_i^{m_i} \right)^{\frac{1}{n}} \left( e^{\alpha} \identity_{m_1} \oplus  e^{-\alpha} \bigoplus_{i=2}^l  \frac{1}{x_i} \identity_{m_i} \right),
\end{align}
we have  $\det(A_\alpha) = \det(B_\alpha) = 1$ and moreover $\lim_{\alpha \rightarrow \infty} A_\alpha \otimes B_\alpha \otimes \identity \ket{\psi} =  \left(\prod_{i=2}^l x_i^{m_i} \right)^{\frac{1}{n}} \ket{\psi'}$, where $\ket{\psi'}$ is a critical state corresponding to an SLOCC class belonging to a matrix pencil of the form given in (ii). This proves statement (iii).

For proving statement (iv), let us now consider the mapping of the problem to a geometrical problem in $\mathbb{R}^3$ as in Section \ref{sec:geometric}.
Recall that rewriting the eigenvalues of the matrix pencil as $x_i = e^{i \phi_i} \tan{\frac{\theta_i}{2}}$  and constructing unit vectors with $(\theta_i,\phi_i)$ as spherical coordinates [see Eq. (\ref{eq:spherical})], then the state $\ket{\psi}$ corresponding to the matrix pencil is related to a critical state via operators $B$ and $C$ iff $\sum_i m_i \vec{v}_i = \vec{0}$ \cite{BrLe19}. At this point we can apply Lemma \ref{lemma:geometric}, which guarantees the existence of an operator $A$ applied on the first system, which transforms the eigenvalues such that the condition is met. Let us remark here, that the assumptions in the lemma are indeed such that the lemma only applies for $m_i$ as in case (iv). As the existence of a critical state is guaranteed, the considered SLOCC class is polystable. Moreover, strictly polystable classes can be distinguished from stable classes by utilizing Observation \ref{obs:degeneracy} in Section \ref{sec:2mncequivalence}, stating that the dimension of the orbit is maximal if and only if there is no degeneracy in the eigenvalues. Hence, statement (iv) follows.
\end{proof}

This completes the classification of orbit types in $2 \times m \times n$ systems. Let us remark that cases (ii) and (iii) can only occur in case $n$ is even.

Let us also remark here that in the geometrical picture utilized to prove statement (iv) of Theorem \ref{thm:diagtype}, statements (i)-(iii), which we have proven differently, also become very transparent. In case (i) we are dealing with a vector $\vec{v}_1$, which is pointing with a majority of the weight, $m_1$, into a single direction. Even if all the remaining vectors could be freely rotated independent of each other, their cumulative weight is at most $m_1-1$. Hence the weighted sum of the vectors must always maintain a distance from the origin of at least one. Hence, in case (i), $C$ must be in the null-cone. In case (ii), we are dealing with exactly two vectors $\vec{v_1}$ and $\vec{v_2}$ with equal weight $m_1 = m_2$. These vectors can be transformed in such a way that they point in opposite directions and hence can clearly be transformed such that they sum up to $\vec{0}$. Hence, a critical state exists. In case (iii) we are dealing with a vector $\vec{v}_1$ pointing in some direction with weight $m_1$ and at least two more vectors $\vec{v}_i$ with cumulative weight $m_1$ pointing in different directions. In order to make the weighted sum of the vectors equal $\vec{0}$, all the remaining vectors would have to point in opposite direction of $\vec{v}_1$. As $A$ cannot change the multiplicities of the eigenvalues, this is not possible by some finite $A$. However, applying the non-linear transformation from Lemma \ref{lemma:geometric},  by first rotating all vectors such that $\vec{v}_1$ points in z-direction and then performing the non-linear hinging operation (see also Figure \ref{fig:operation}), it can be seen that in the limit $\alpha \rightarrow 1$ all vectors $\vec{v}_i$ for $i \in \{2, \ldots, l\}$ point into $(-z)$-direction. Thus the weighted vector sum equals $\vec{0}$ (only) in the limit and thus, case (iii) corresponds to a strictly semistable SLOCC class.

\subsection{Algorithm to determine the orbit type}
\label{sec:flowchart}

In this subsection, we present a flowchart (Figure \ref{fig:recipe}) which can be followed in order to decide the orbit type of an SLOCC class in $\mathbb{C}^2 \otimes \mathbb{C}^m \otimes \mathbb{C}^n$ given in terms of the representing matrix pencil $\mathcal{P}$. The flowchart summarizes the results obtained in this section. We present tables of the orbit types of SLOCC classes for system sizes up to $2\times 5 \times 5$ in Appendix \ref{app:tables}.

\begin{figure*}[!ht]
\includegraphics[width=1.0 \linewidth]{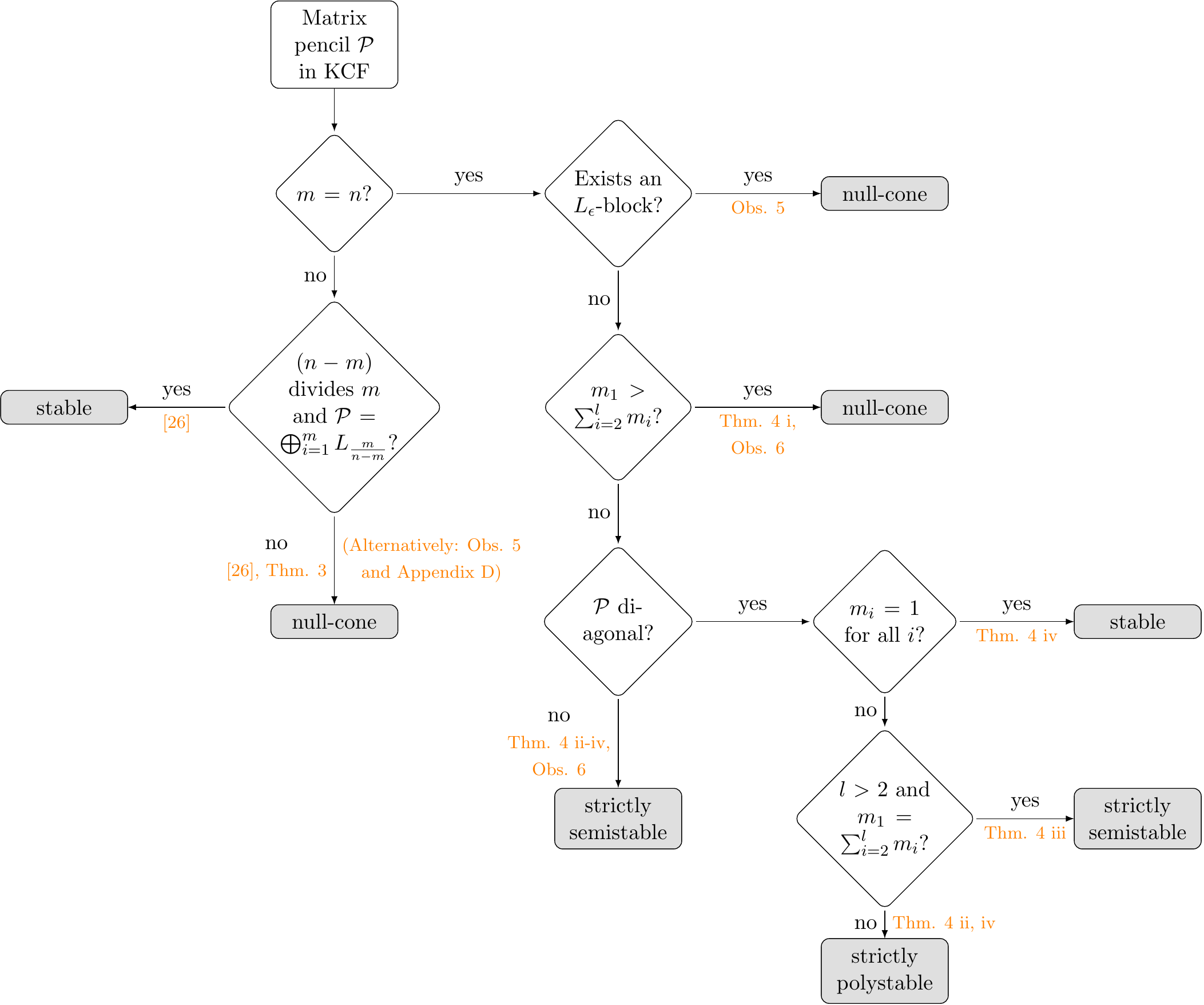}
  \caption{Characterization of orbit type of any SLOCC class in $\mathbb{C}^2 \otimes \mathbb{C}^m \otimes \mathbb{C}^n$ (wlog $m \leq n$). $\mathcal{P}$ denotes the $m \times n$ matrix pencil whose corresponding SLOCC class we are interested in. Wlog $\mathcal{P}$ is in KCF with $0 \leq l \leq n$ distinct eigenvalues. Their multiplicities are denoted by $m_i$, and are sorted in descending order. Then the orbit type can be decided by following the flowchart. Remarkably, the concrete values which the eigenvalues of the matrix pencil take do not play a role when deciding the orbit type, only their multiplicities matter.
  A non-rectangular matrix pencil (i.e., $m \neq n$) either equals $\mathcal{P} = \bigoplus_{i=1}^m L_{\frac{m}{n-m}}$ (which is only possible if $n-m$ divides $m$), in which case the class is stable, as shown in \cite{BrLe19}, or is in the null-cone otherwise, as either \cite{BrLe19} together with Theorem \ref{thm:2mnsemistable}, or, alternatively,
   Observation \ref{obs:epsilon_nu_null-cone} and the considerations in Appendix \ref{app:nullcone2mn} show. In case of square matrix pencils (i.e. $m=n$), the class is in the null-cone if any $L$-block is present. The reason for that is that such a block must always be accompanied by an $L^T$-block and thus Observation \ref{obs:epsilon_nu_null-cone} applies. Theorem \ref{thm:diagtype} and Observation \ref{obs:diagonalreduction} then allow classification of square matrix pencils that do not contain any $L$- or $L^T$-block depending on the multiplicities of the eigenvalues.}
  \label{fig:recipe}
\end{figure*}

\section{Conclusion}

In this work we have derived a necessary and sufficient criterion for the existence of strictly semistable states in terms of the orbit dimensions of polystable states using GIT methods. In particular, we have shown that strictly semistable states exist if and only if there exist polystable states whose orbit dimensions differ. We have applied the criterion to three-partite states of local dimensions $2$, $m$, and $n$ and characterized those $m$ and $n$, for which strictly semistable states do exist. Going beyond that, we have scrutinized all SLOCC classes in such systems and characterized their orbit types. 
It turns out that (for different values of $m$ and $n$), a rich variety of behaviors manifests. More precisely, there are examples where only one SLOCC class containing fully entangled states exists (which is stable); there are examples in which all SLOCC classes are in the null-cone (as characterized in \cite{BrLe19}); in some systems, stable states and states in the null-cone exist, but no strictly semistable states; and finally, there are systems in which all of the discussed orbit types are present. Moreover, all of these situations become already apparent within small system sizes, that is $m \leq n \leq 5$ as shown in Appendix \ref{app:tables}.

It would be interesting to see whether one can identify a difference between stable states, strictly polystable states, strictly semistable states, and states in the null-cone for practical applications. A well-known result in this vein is that the entanglement in the W-state, which is in the null-cone, is more robust than the entanglement in the three-qubit GHZ-state, which is stable \cite{DuVi00}.

Finally, let us remark that $2 \times m \times n$-states find application as fiducial states in the context of  matrix product states \cite{SaMo19}. In particular, properties of $2 \times m \times n$-states can be used to characterize properties of the matrix product state they give rise to, such as local symmetries, or SLOCC classes.

\acknowledgments
O.S. and A.S. acknowledge financial support from National Science Centre, Poland under the  grant SONATA BIS: 2015/18/E/ST1/00200. M.H. and B.K. acknowledge financial support from the Austrian Science Fund (FWF) grant DK-ALM: W1259-N27 and the SFB BeyondC (Grant No. F7107-N38). Furthermore, B.K. acknowledges support of the Austrian Academy of Sciences via the Innovation Fund ``Research, Science and Society''.


\bibliography{paper} 


\appendix

\section{Proof of Observation \ref{obs:closuredifference}}
\label{app:closuredifference}

In this appendix we prove Observation \ref{obs:closuredifference}. In order to do so, we will first state and prove a lemma and another observation.

\begin{lemma}
\label{lemmainfty}
Suppose $\ket{\Psi} \in \mathcal{H}_{ss}$ and we have a sequence $(g_n)_n,\, g_n \in G$ such that $\norm{g_n \ket{\Psi}} \xrightarrow{n} \infty$.
Then for any sequence $(\phi_n)_n$, $$\lim_{n \to \infty} \frac{g_n \ket{\Psi}}{\norm{g_n \ket{\Psi}}}e^{i\phi_n} \in \mathcal{N}$$ whenever this limit exists.
\end{lemma}

\begin{proof}
Let $\{p_1, \ldots, p_k\}$ be the set of homogeneous polynomials $p_i: \mathcal{H} \rightarrow \mathbb{C}$ of degree $\mathrm{\deg} \, p_i$ that generate the ring of $G$--invariant polynomials. Then for $\ket{\Phi}=\displaystyle\lim_{n \to \infty} \frac{g_n \ket{\Psi}}{\norm{g_n \ket{\Psi}}}e^{i\phi_n}$ we have
\begin{align}
    \forall_i \; p_i (\ket{\Phi})&=\lim_{n \to \infty} p_i \left( \frac{g_n \ket{\Psi}}{\norm{g_n \ket{\Psi}}}e^{i\phi_n} \right) \nonumber\\
     &= p_i ( \ket{\Psi}) \lim_{n \to \infty} \left (\frac{e^{i\phi_n}}{\norm{g_n \ket{\Psi}}} \right)^{\mathrm{deg} \, p_i} =0,\nonumber
\end{align}
since the limit exists and all $p_i$ are continuous, homogeneous and $G$--invariant. Thus $\forall_i \; p_i(\ket{\Phi})=0$ so $\ket{\Phi} \in \mathcal{N}$.
\end{proof}

\begin{observation}
\label{obs:lift}
If $[\Psi_n] \xrightarrow{n \to \infty} [\Phi]$ in $\mathbb{P}(\mathcal{H})$, then the sequence $(\ket{\tilde{\Psi}_n})_n$ in $\mathcal{H}$ with $\ket{\tilde{\Psi}_n} = \frac{\ket{\Psi_n}}{\norm{\ket{\Psi_n}}}e^{i\phi_n}$ and $\phi_n = \mathrm{arg}(\braket{\Psi_n}{\Phi})$ converges to $\ket{\tilde{\Phi}}= \frac{\ket{\Phi}}{\norm{\ket{\Phi}}}$, i.e.,
\begin{align*}
\ket{\tilde{\Psi}_n} \xrightarrow{n \to \infty} \ket{\tilde{\Phi}}.
\end{align*}
\end{observation}

\begin{proof}
Let us introduce the Fubini-Study distance in $\mathbb{P}(\mathcal{H})$,
$$d([\Psi],[\Phi])=\mathrm{arccos}\frac{\abs{\braket{\Psi}{\Phi}}}{\norm{\ket{\Psi}}\norm{\ket{\Phi}}}\mathrm{,}$$
which defines a metric compatible with the standard topology on $\mathbb{P}(\mathcal{H})$.
The convergence $[\Psi_n] \xrightarrow{n \to \infty} [\Phi]$ implies that $d([\Psi_n],[\Phi]) \xrightarrow{n \to \infty} 0$ which means that $$\frac{\abs{\braket{\Psi_n}{\Phi}}}{\norm{\ket{\Psi_n}}\norm{\ket{\Phi}}} \xrightarrow{n \to \infty} 1 \text{.}$$ On the other hand 

\begin{align*}
\norm{\ket{\tilde{\Psi}_n}-\ket{\tilde{\Phi}}}&=2\left(1- \frac{\Re \left( e^{-i\phi_n}\braket{\Psi_n}{\Phi}\right)}{\norm{\ket{\Psi_n}}\norm{\ket{\Phi}}}\right)=\\
&=2\left(1- \frac{\abs{\braket{\Psi_n}{\Phi}}}{\norm{\ket{\Psi_n}}\norm{\ket{\Phi}}}\mathrm{cos}(\psi_n-\phi_n)\right),
\end{align*}
where $\psi_n = \mathrm{arg}(\braket{\Psi_n}{\Phi})$. Hence, with the choice $\phi_n =\psi_n$ one obtains
$$\norm{\ket{\tilde{\Psi}_n}-\ket{\tilde{\Phi}}}=2\left(1- \frac{\abs{\braket{\Psi_n}{\Phi}}}{\norm{\ket{\Psi_n}}\norm{\ket{\Phi}}}\right) \xrightarrow{n \to \infty} 0 \text{.}$$

\end{proof}

Now we are ready to prove Observation \ref{obs:closuredifference}, which we restate here in order to increase readability.

\noindent {{\bf Observation \ref{obs:closuredifference}} {\bf.}}\textit{
For any vector $\ket{\Psi} \in \mathcal{H}_{ss}$, we have
\bea \overline{[\; G \ket{\Psi} \;]} \setminus [\; \overline{G \ket{\Psi}} \;] \subset [{\cal N}\setminus\{0\}].\eea
}

\begin{proof}
Notice that $\overline{[\; G \ket{\Psi} \;]}$ is $G$--invariant and the projection $[\cdot]$ is continuous in the standard topologies on $\mathcal{H}$ and $\mathbb{P}(\mathcal{H})$, thus
$$[\; \overline{G \ket{\Psi}} \;] \subset \overline{[\; G \ket{\Psi} \;]} \text{.}$$
Let $\ket{\Psi} \in \mathcal{H}_{ss}$. Choose any state $[\Phi] \in \overline{[\; G \ket{\Psi} \;]}$. We will show that either $[\Phi] \in [\mathcal{N}\setminus \{0\}]$ or $[\Phi] \in [\; \overline{G \ket{\Psi}} \;]$. Let $(g_n)_n, g_n \in G$, be a sequence of elements in $G$ such that $[g_n \ket{\Psi}]\xrightarrow{n \to \infty}[\Phi]$. Due to Observation \ref{obs:lift} there exists a sequence $(\ket{\Psi_n})_n$ such that $\ket{\Psi_n} \xrightarrow{n \to\infty} \ket{\tilde{\Phi}}$ where $$\ket{\Psi_n} = \frac{g_n \ket{\Psi}}{\norm{g_n \ket{\Psi}}}e^{i\phi_n}, \quad \ket{\tilde{\Phi}}=\frac{\ket{\Phi}}{\norm{\ket{\Phi}}}  \text{.}$$
Let us denote $$r_n \coloneqq \norm{g_n \ket{\Psi}}\in \mathbb{R}, \quad c_n \coloneqq r_n e^{-i\phi_n}\in \mathbb{C}\text{,}$$
so $c_n \ket{\Psi_n} = g_n \ket{\Psi}$ and $\abs{c_n}=r_n$. 
Suppose that the sequence $(r_n)_n$ is not bounded from above. Then we can pick a subsequence $(r_{n_k})_k$ such that $r_{n_k} \xrightarrow{k \to \infty} \infty$. It holds that $\ket{\Psi_{n_k}}\xrightarrow{k \to\infty}\ket{\tilde{\Phi}}$ and from Lemma \ref{lemmainfty} we get that $[\Phi]=[\tilde{\Phi}] \in [\mathcal{N} \setminus \{0\}]$ \footnote{Note that the projection of this diverging sequence of vectors (i.e. the corresponding sequence of states) may actually converge in $\mathbb{P}(\mathcal{H})$ so that this sequence of states corresponds to a state from the closure of $[G\ket{\Psi}]$ in $\mathbb{P}(\mathcal{H})$}. Suppose that $(r_n)_n$ is bounded from above by some $b \in \mathbb{R}_{+}$. Because $\ket{\Psi} \in \mathcal{H}_{ss}$, the sequence $(r_n)_n$ has to be bounded from below by a non-zero norm infimum $a \in \mathbb{R}_{+}$. Thus, for every $n$ we have $r_n \in [a, b]$, i.e., $(c_n)_n$ is bounded with $\abs{c_n}=r_n \in [a,b] \subset \mathbb{R}_{+}$. From the Bolzano-Weierstrass theorem for $\mathbb{C}$, we can pick a convergent subsequence $(c_{n_k})_k$ such that $c_{n_k} \xrightarrow{k \to \infty} c $ and the corresponding subsequence $(\ket{\Psi_{n_k}})_k$.  Moreover $c\in \mathbb{C}\setminus\{0\}$, as $\abs{c}\in[a,b] \subset \mathbb{R}_{+}$. Since $\ket{\Psi_n} \xrightarrow{n \to \infty} \ket{\tilde{\Phi}}$ then also $\ket{\Psi_{n_k}} \xrightarrow{k \to \infty} \ket{\tilde{\Phi}}$. The product  $c_{n_k} \ket{\Psi_{n_k}}$ of convergent sequences $c_{n_k}$ and $\ket{\Psi_{n_k}}$ converges to $c\ket{\tilde{\Phi}}$. Thus $g_{n_k}\ket{\Psi} \xrightarrow{k \to \infty} c \ket{\tilde{\Phi}} \in \overline{G  \ket{\Psi}}$. But $[\Phi]=[c \ket{\tilde{\Phi}}]$ since $c \neq 0$,  so $[\Phi] \in [\; \overline{G \ket{\Psi}} \;]$. This completes the proof.
\end{proof}

\section{An equivalent definition of stability}
\label{app:stab_eq}

In this appendix, we show that the stability of the action can be defined in an alternative way, which is equivalent to the definition introduced in the main text, Definition \ref{stabmax}. In the course of that, it will become clear that the maximal orbit dimension among all polystable orbits equals the maximal orbit dimension among all orbits, $d_G$, in case the action is stable.

Let us now consider the following alternative definition of the stability of the action.
\begin{definition}
\label{stabstable}
An action $H \acts \mathbb{P}(\mathcal{H})$ is stable if there exists an open dense set $U \subset \mathbb{P}(\mathcal{H})$ that is a union of stable orbits.
\end{definition}

We introduce the definition of a function $\mathrm{dim} \, G \textunderscore $ as follows: the value of $\mathrm{dim} \, G \textunderscore$ at $x$ is equal to $\mathrm{dim} \, G_x$. In order to prove the equivalence of Definition \ref{stabstable} to Definition \ref{stabmax}, we use the following lemma, which is a particular version of Proposition 2.27. from \cite{Ho12}.

\begin{lemma}[\cite{Ho12}]
\label{ios}
In our setting, the stabilizer dimension function $\mathrm{dim} \, G \textunderscore : \mathbb{P}(\mathcal{H}) \rightarrow \mathbb{N}_0$ is upper semi-continuous, i.e. for each $n \in \mathbb{N}_0$, the set
\begin{equation} \label{eq:uppsemicont}
X_n:=\{x \in \mathbb{P}(\mathcal{H}) : \mathrm{dim} \, G_x \geq n\} \text{,}
\end{equation}
is closed in $\mathbb{P}(\mathcal{H})$.
\end{lemma}

Below we prove the equivalence of Definitions \ref{stabmax} and \ref{stabstable}.

\begin{proof}[Proof of the equivalence of Defs. \ref{stabmax} and \ref{stabstable}.]
Let $n_G$ denote the minimal stabilizer dimension for $G \acts \mathbb{P}(\mathcal{H})$ and suppose that the action is stable. Consider a set $X:=X_{n_G} \setminus X_{n_{G}+1}$. The set $X$ contains exactly those orbits whose stabilizers have dimension $n_G$, i.e., orbits with maximal dimension. As $X_{n_G}=\mathbb{P}(\mathcal{H})$ and $X_{n_{G}+1}$ is closed (Lemma \ref{ios}), the set $X$ is open (and dense). From the stability of action we know that there is a set $U$ which is the open (and dense) set of closed orbits mentioned in Definition \ref{stabmax}. Consider the set $Y:= X \cap U$. Then, $Y$ has all of the defining properties of the set $U$ and, in addition, it contains only orbits of the maximal dimension. Thus, we can always choose the set $U$ in Definition \ref{stabmax} in such a way that it contains orbits of the maximal dimension.
\end{proof}

Moreover, from the proof it follows that in case of stable action, the dimension of stable orbits is maximal among all orbits, not only among polystable ones.

\section{Regularity of the GIT-quotient \texorpdfstring{$\pi$}{}}
\label{app:regularmap}

In this appendix, we elaborate in more details on the assumptions of Lemma \ref{fibres} from Section \ref{sec:mainsubsection}. In particular, we show that the lemma applies to the GIT quotient $\pi : \mathbb{P}(\mathcal{H})_{ss} \rightarrow  \mathbb{P}(\mathcal{H})_{ss}//G$.

Let us first introduce the following definitions. Here and in the following, $\mathbb{P}^q$ denotes a $q$-dimensional projective space, e.g., $\mathbb{P}(\mathcal{H})$ for $q=\mathrm{dim}(\mathcal{H})-1$.

\begin{definition}[\cite{SmKa00} Ch. 4, p. 49]
\label{qproj}
An algebraic variety is quasi-projective if it is an intersection of open and/or closed subsets of $\mathbb{P}^q$ considered with the Zariski topology induced from $\mathbb{P}^q$.
\end{definition}

\begin{definition}
\label{irred}
\leavevmode
\begin{enumerate}[(i)]
\item A topological space $X$ is irreducible if $X$ cannot be written as the union of two closed and proper subsets $F_1 \subsetneq X$, $F_2 \subsetneq X$, i.e., $X \neq F_1 \cup F_2$.
\item A subset $A \subset X$ of a topological space $X$ is irreducible if it is irreducible as a topological space with topology induced from $X$.
\end{enumerate}
\end{definition}

With these definitions, the following observation can be obtained.

\begin{observation}
\label{obsirred}
\leavevmode
\begin{enumerate}[(i)]
\item A topological space $X$ is irreducible if and only if any non-empty open subset $U \subset X$ is dense in $X$.
\item Any open subset $U \subset X$ of an irreducible topological space $X$ is irreducible.
\item If $f:X \rightarrow Y$ is a continuous map between topological spaces then if $X$ is irreducible so is $f(X)\subset Y$.
\end{enumerate}
\end{observation}

Central in Lemma \ref{fibres} is also the definition of a regular map, which is given in the following.

\begin{definition}[\cite{SmKa00}]
\label{reg}
We say that a map $f: X \rightarrow Y$ between a quasi-projective variety $X \subset \mathbb{P}^q$ and a projective variety $Y \subset \mathbb{P}^r$ is regular (or a morphism) if for every $x \in X$ there exist homogeneous polynomials of the same degree $F_0,\ldots, F_r$ in $q+1$ variables (homogeneous coordinates of $x$) such that the map $X \rightarrow \mathbb{P}^r$ given by $x \mapsto [F_0(x): \ldots : F_{r}(x)]$ is well-defined and agrees with $f$ on some open neighborhood of $x$. In particular for every $x \in X$ there has to be the local expression for $f$ of the above form such that ${\ F_j(x) \neq 0}$ for some $j$ .
\end{definition}
It can be seen that one can also use $r$ rational functions $f_i=F_i/F_j$, where $i \neq j$ to define regular maps.

The basic fact states that regular maps between quasi-projective varieties are continuous (see \cite{Sh13}).
Using the discussion above in the context of our GIT quotient $\mathbb{P}(\mathcal{H})_{ss} \rightarrow \mathbb{P}(\mathcal{H})_{ss}//G$, we have
\begin{enumerate}[(i)]
\item Projective spaces (e.g. $\mathbb{P}(\mathcal{H})$) are irreducible.
\item The set of semistable states $\mathbb{P}(\mathcal{H})_{ss}$ is an irreducible quasi-projective variety. Indeed, since $[\mathcal{N}\setminus \{0\}]$ is a closed set, $\mathbb{P}(\mathcal{H})_{ss}$ is an open subset of irreducible projective variety $\mathbb{P}(\mathcal{H})$. Obviously $\mathbb{P}(\mathcal{H})_{ss}=\mathbb{P}(\mathcal{H})_{ss} \cap \mathbb{P}(\mathcal{H})$ and $\mathbb{P}(\mathcal{H})$ is closed.
\item The GIT quotient $\pi:\mathbb{P}(\mathcal{H})_{ss} \rightarrow \mathbb{P}(\mathcal{H})_{ss}//G$ is a regular map between quasi-projective algebraic varieties, i.e. a morphism (see \cite{Br10} p. 10). In particular $\pi$ is continuous.
\item The set $\mathbb{P}(\mathcal{H})_{ss}//G$ is a projective variety (see \cite{Br10} p. 11, Proposition 1.29.). Moreover $\mathbb{P}(\mathcal{H})_{ss}//G$ is irreducible.  Thus, $\mathbb{P}(\mathcal{H})_{ss}//G$ is the irreducible quasi-projective variety.
\end{enumerate}

\section{Additional details on \texorpdfstring{$2 \times m \times n$}{2 x \textit{m} x \textit{n}} classes that are in the null-cone.}
\label{app:nullcone2mn}

In this appendix, we explicitly construct sequences of SL-operators $A_\alpha$, $B_\alpha$, and $C_\alpha$, such that $\lim_{\alpha \rightarrow \infty} A_\alpha \otimes B_\alpha \otimes C_\alpha \ket{\psi} = 0$ for $\ket{\psi} \in \mathbb{C}^2 \otimes \mathbb{C}^m \otimes \mathbb{C}^n$ corresponding to matrix pencils in KCF, in which certain combinations of blocks are present.
Observation \ref{obs:epsilon_nu_null-cone} in the main text does exactly that for matrix pencils (in KCF) that contain an $L_\epsilon$- as well as an $L^T_\nu$-block and thus shows that SLOCC classes corresponding to such matrix pencils are in the null-cone. Here, we show that a similar reasoning applies to matrix pencils whose KCF contains an $L_{\epsilon_1}$- and an $L_{\epsilon_2}$-block of different size, i.e., $\epsilon_1 \neq \epsilon_2$. Moreover, the reasoning also applies to matrix pencils that contain an $L_{\epsilon}$- and an $M$-block. Note however, that those SLOCC classes correspond to the nullcone follows already from \cite{BrLe19} together with Theorem \ref{thm:2mnsemistable}, as discussed in the main text.

Let us first consider an SLOCC class that corresponds to a matrix pencil whose KCF contains an $L_{\epsilon_1}$- and an $L_{\epsilon_2}$-block of different size, i.e., $\epsilon_1 \neq \epsilon_2$. Wlog we assume an ordering of the blocks such that the matrix pencil has the form $\mathcal{P} = L_{\epsilon_1} \oplus L_{\epsilon_2} \oplus \tilde{\mathcal{P}}$, where $\tilde{\mathcal{P}}$  is of size $[m- (\epsilon_1 +\epsilon_2)] \times [n- (\epsilon_1 +\epsilon_2 +2)]$ and contains the remaining blocks. As discussed above, such an SLOCC classes is in the null-cone (the same holds for two $L^T$ blocks of different size). Let us now explicitly construct a sequence of SL-operations which, in the limit, transforms states corresponding to such $\mathcal{P}$ to $0$.
Let us define
\begin{align}
B_\alpha' &= e^{-\frac{\epsilon_2}{\epsilon_2 - \epsilon_1} \alpha} \identity_{\epsilon_1} \oplus  e^{\frac{\epsilon_1}{\epsilon_2 - \epsilon_1} \alpha} \identity_{\epsilon_2} \text{ and} \\
C_\alpha' &=  e^{\frac{\epsilon_2}{\epsilon_2 - \epsilon_1} \alpha} \identity_{\epsilon_1 + 1} \oplus e^{-\frac{ \epsilon_2}{\epsilon_2 - \epsilon_1}\alpha} \identity_1 \oplus  e^{-\frac{\epsilon_1}{\epsilon_2 - \epsilon_1} \alpha} \identity_{\epsilon_2}
\end{align}
of sizes $(\epsilon_1+\epsilon_2) \times (\epsilon_1+\epsilon_2)$ and $(\epsilon_1+\epsilon_2 + 2) \times (\epsilon_1+\epsilon_2 + 2)$, respectively. We take $C_\alpha''$ as in the proof of Observation \ref{obs:epsilon_nu_null-cone}, replacing $\epsilon$ with $\epsilon_1 + 1$, in order to construct
$B_\alpha =  B_\alpha' \oplus \identity_{m-(\epsilon_1 + \epsilon_2)} $ and $C_\alpha = C_\alpha'' \left( C_\alpha' \oplus \identity_{n-(\epsilon_1 + \epsilon_2 + 2)} \right)$. Then, $\lim_{\alpha \rightarrow \infty} B_\alpha \mathcal{P} C_\alpha^T = 0$.
For two $L^T$ blocks instead of two $L$ blocks, the matrices can be constructed analogously.

The intuition behind this construction is similar as in Observation \ref{obs:epsilon_nu_null-cone}. As in Observation \ref{obs:epsilon_nu_null-cone},  $B'$ and $C'$ are constructed in order to damp one of the columns of the matrix pencil. Here, it is the $(\epsilon_1 + 1)$st column of the $L_{\epsilon_1} \oplus L_{\epsilon_2}$ part in $\mathcal{P}$, leaving the remaining columns invariant. Moreover, $C_\alpha''$ is used to damp all the remaining columns at the cost of boosting the $(\epsilon_1 + 1)$st column in such a way that the damping overcompensates the boosting as in Observation \ref{obs:epsilon_nu_null-cone}.

Let us finally consider an SLOCC class that corresponds to a matrix pencil whose KCF contains an $L_{\epsilon}$- and an $M$-block (with arbitrary eigenvalue and arbitrary size $e$). Again, such an SLOCC classes is in the null-cone, as discussed above (the same holds for an $L^T_\nu$- together with an $M$-block).
A sequence of SL-operations which, in the limit, transforms states corresponding to such matrix pencils to $0$ can be constructed as follows.
Consider an ordering of the blocks such that the matrix pencil has the form $\mathcal{P} = L_{\epsilon} \oplus M \oplus \tilde{\mathcal{P}}$, where $\tilde{\mathcal{P}}$ contains the remaining blocks.
Let us define
\begin{align}
B_\alpha' &=  e^{- \alpha} \identity_{\epsilon} \oplus e^{\frac{\epsilon}{e} \alpha} \identity_{e}  \text{ and} \\
C_\alpha' &=  e^{\alpha} \identity_{\epsilon} \oplus  \identity_{1} \oplus e^{-\frac{\epsilon}{e}  \alpha} \identity_{e}
\end{align}
of sizes $(\epsilon + e) \times (\epsilon + e)$ and $(\epsilon+ 1+e) \times (\epsilon + 1+e)$, respectively. We take $C_\alpha''$ as in the proof of Observation \ref{obs:epsilon_nu_null-cone} in order to construct $B_\alpha =  B_\alpha' \oplus \identity_{m-(\epsilon + e)} $ and $C_\alpha = C_\alpha'' \left( C_\alpha' \oplus \identity_{n-(\epsilon + 1 + e)} \right)$. Then, $\lim_{\alpha \rightarrow \infty} B_\alpha \mathcal{P} C_\alpha^T = 0$, which proves the statement. In case of an $L^T$- together with an $M$-block, the matrices can be constructed analogously.

\section{Proof of Lemma \ref{lemma:geometric}}
\label{app:geometric}
In this appendix, we restate and then prove Lemma \ref{lemma:geometric}.

\noindent {{\bf Lemma \ref{lemma:geometric}}{\bf.}} {\it
Let $\{\vec{v}_i\}_{i=1}^l$ be a set of $l$ distinct unit vectors in $\mathbb{R}^3$ with associated multiplicities $m_i \in \mathbb{N}$ sorted in descending order such that $\sum_{i=1}^l m_i = n$ for some $n \geq 3$ and $m_1 < \sum_{i =2}^l m_i$. Let $(\theta_i, \phi_i)$ denote the spherical coordinates of $\vec{v}_i$, i.e.,
\begin{align}
\vec{v_i} = \begin{pmatrix}
\sin \theta_i \cos \phi_i \\
\sin \theta_i \sin \phi_i\\
\cos \theta_i
\end{pmatrix}.
\end{align}
Moreover, let $f_\alpha$ be a (non-linear) transformation acting on the spherical coordinates of the vectors as
\begin{align}
f_{\alpha}: [0, \pi]\times [0, 2\pi] &\rightarrow [0, \pi]\times [0, 2\pi] \nonumber\\
(\theta_i, \phi_i) &\mapsto   ( 2 \arctan \left(\frac{1}{1-\alpha} \tan(\theta_i / 2) \right), \phi_i),
\end{align}
where $\alpha \in [0,1)$. Then, it is always possible to transform the vectors $\{\vec{v}_i\}_{i=1}^l$ to new set of vectors $\{\vec{v'}_i\}_{i=1}^l$ such that $\sum_i m_i \vec{v'}_i = 0$ through some simultaneous $SO(3)$ rotation of the vectors $\vec{v}_i$ followed by $f_{\alpha}$.
}

Let us first discuss some properties of the considered operation before we prove the lemma. First, note that the described operation is exactly the operation described in Section \ref{sec:geometric}. Recall that the non-linear part can be geometrically understood as hinging all $\vec{v}_i$ (except those pointing in direction $(0,0,1)^T$ or $(0,0,-1)^T$) down towards direction $(0,0,-1)^T$ and the $SO(3)$ operation can be understood as choosing some axis $(\Theta, \Phi)$ with respect to which the non-linear transformation is performed.

Let us now define a function $f$ for any set of distinct vectors $\{\vec{v}_i\}_{i=1}^l$ (with mutliplicities $m_i$) fulfilling the premises of Lemma \ref{lemma:geometric} as follows.
\begin{align}
\label{eq:f}
f : D &\rightarrow [0,n^2] \\
(\alpha,\Theta,\Phi) &\mapsto f(\alpha,\Theta,\Phi)=\left| \sum_i m_i \vec{v}'_i \right|^2,
\end{align}
where $D = [0, 1) \times [0, \pi] \times [0, 2 \pi]$ and $\vec{v}'_i$ are obtained by transforming the vectors  $\vec{v}_i$ as follows. First, the vectors $\vec{v}_i$ are rotated around the z-axis by the angle $\Phi$. Then they are rotated around the y-axis by $\Theta$. Finally, $f_\alpha$ is performed on them.
See Figure \ref{fig:operation} for a sketch of the operation.

\begin{figure}[H]
 \includegraphics[width=0.7 \columnwidth]{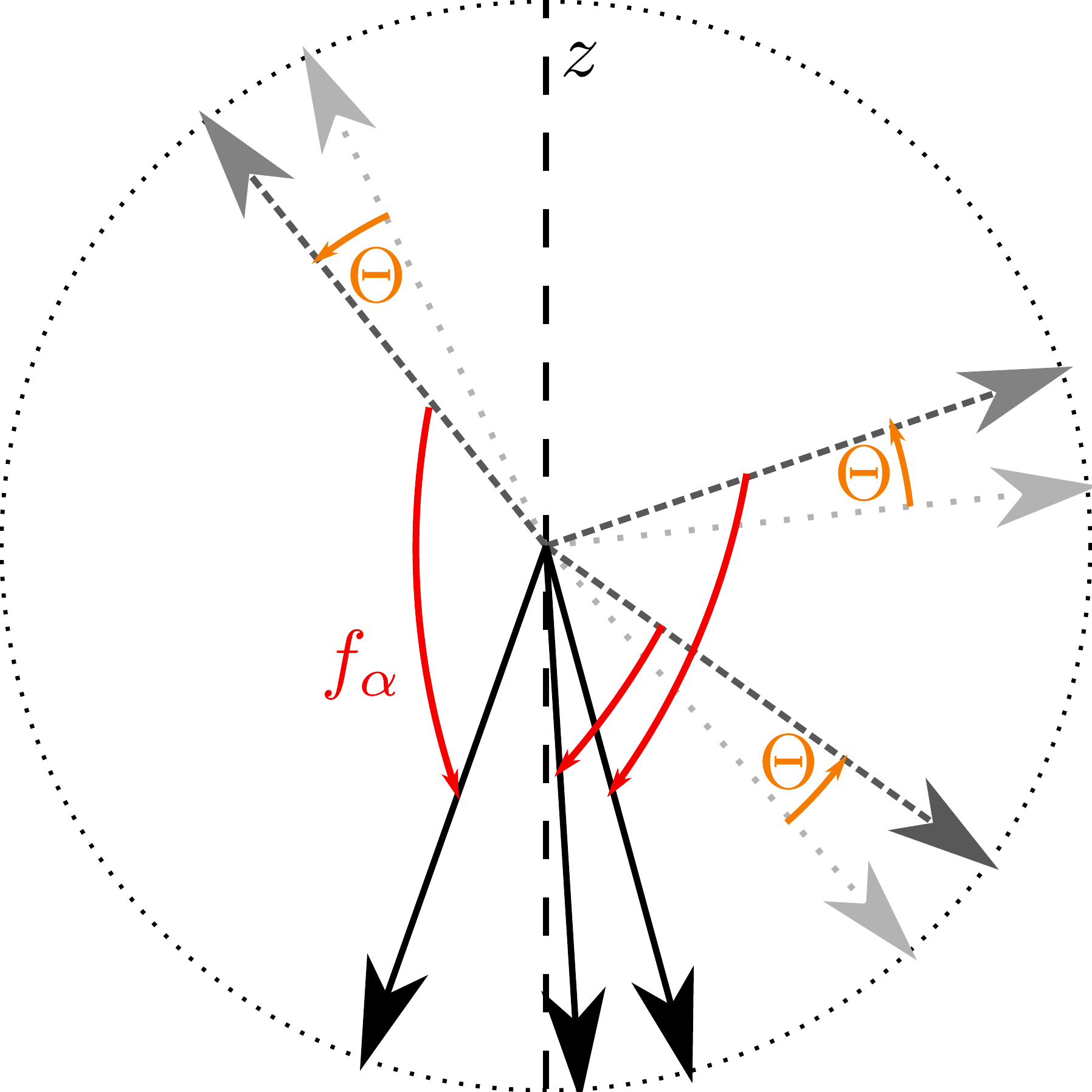}
  \caption{Sketch of the operation considered in Lemma \ref{lemma:geometric}. We depict a case in which $l=3$, $m_i=1$, and for simplicity all vectors (arrows) as well as the $z$-axis (dashed line) lie in the same plane. Initial vectors $\{\vec{v}_i\}$ of unit length (dotted arrows) are first jointly rotated (yielding the dashed arrows) followed by performing the non-linear operation $f_\alpha$ yielding the vectors $\{\vec{v'}_i\}$ (solid arrows), an operation that can be pictorially understood as hinging the vectors down towards negative $z$-direction. In the limit $\alpha \rightarrow 1$, $\{\vec{v'}_i\}$ become perfectly antiparallelly aligned with the $z$-axis. Vectors which would already be perfectly aligned with the $z$-axis after the initial rotation are exceptions to this, they would remain unchanged under $f_\alpha$ (not shown in the graphic).}
  \label{fig:operation}
\end{figure}

 Obviously, $f$ is bounded from below and above respectively by 0 (when the vectors $\vec{v}'_i$ sum to 0) and $n^2$ (when all vectors $\vec{v}'_i$ are aligned, which is not possible via an invertible operation, but nevertheless upper bounds $f$). In order to prove the lemma, we will show that $f$ actually attains 0, regardless of the initial orientations of the vectors, as long as the premises of Lemma \ref{lemma:geometric} are fulfilled.
To this end, let us define $\tilde{f}$ extending the function $f$ to the domain $\tilde{D} = [0, 1] \times [0, \pi] \times [0, 2 \pi]$ through
\begin{align}
\label{eq:ftilde}
\tilde{f}(\alpha,\Theta,\Phi) = \begin{cases}
f(\alpha,\Theta,\Phi), & \text{if } \alpha < 1\\
   (n- 2 m_i)^2 & \text{if } \alpha = 1 \text{ and } (\Theta, \Phi) = \\ &\qquad (\theta_i, \phi_i) \text{ for some } i \\
   n^2 & \text{otherwise}.
\end{cases}
\end{align}
Here, note that the condition $(\Theta, \Phi) = (\theta_i, \phi_i)$ means that the axis with respect to which the hinging operation is performed is exactly aligned with the vector $\vec{v}_i$. Moreover, note that by construction $\tilde{f}(1, \Theta, \Phi) = \lim_{\alpha \rightarrow 1} f(\alpha, \Theta, \Phi)$ for all $\Theta$, $\Phi$.

Let us now state some properties of the functions $f$ and $\tilde{f}$ in the following observation, which will allow us to prove Lemma \ref{lemma:geometric}.

\begin{observation}
\label{obs:properties}
The following statements hold.
\begin{enumerate}[(i)]
\item $f$ is continuous on $D$.
\item For all $x  \in D \setminus L_0(f)$ there exists some $x'$ in $D$ such that $f(x') < f(x)$, where $L_0(f) = \{x \in D : f(x) = 0\}$ is the zero-level set of the function.
\item For a fixed pair $\Theta$, $\Phi$ there exists an $\alpha_*$ such that $f(\alpha, \Theta, \Phi)$ is strictly monotonically increasing in $\alpha$  for all $\alpha > \alpha_*$. Note however, that there does not exist an $\alpha_*$ such that the statement holds for all $\Theta$, $\Phi$.
\item $\tilde{f}$ is lower semi-continuous on $\tilde{D}$. That is, for all $x \in \tilde{D}$ it holds that for all $\delta > 0$ there exists some $R > 0$ such that for all $x' \in \tilde{D}$ with $d(x,x') \leq R$, $\tilde{f}(x') \geq \tilde{f}(x) - \delta$.  Here, $d$ is a distance on $\tilde{D}$.
\item For any fixed $\Theta$ and $\Phi$, there exists some $\alpha \in [0,1)$ such that $ \tilde{f}(\alpha,\Theta,\Phi) < \tilde{f}(1,\Theta,\Phi)$.
\end{enumerate}
\end{observation}

With this observation, we are now in the position to prove Lemma \ref{lemma:geometric}. Afterwards, we will prove the observation.

\begin{proof}[Proof of Lemma \ref{lemma:geometric}]
As discussed above, in order to prove the lemma it suffices to show that $f$ as defined in Eq. (\ref{eq:f}) attains zero. To this end, we also make use of $\tilde{f}$ as defined in Eq. (\ref{eq:ftilde}) and Observation \ref{obs:properties}.
As $\tilde{f}$ is a lower semi-continuous function on a compact domain $\tilde{D}$ according to Observation \ref{obs:properties} (iv), Weierstrass' extreme value theorem guarantees that $\tilde{f}$ attains its infimum on $\tilde{D}$. Let us denote the point at which the infimum is attained by $x_0 = (\alpha_0, \Theta_0, \Phi_0) \in \tilde{D}$. Observation \ref{obs:properties} (v) guarantees that $\alpha_0 < 1$ and therefore $x_0 \in D$. Hence, also $f$ attains its infimum at $x_0$. Lastly, Observation \ref{obs:properties} (ii) guarantees that this infimum is $0$.
\end{proof}

\begin{proof}[Proof of Observation \ref{obs:properties}]
Proof of statement (i). To see that $f$ is continuous on $D$, it can first be verified that $f$ can be written as
\begin{align}
 f(\alpha,\Theta,\Phi)&= \left| \sum_i m_i  \frac{2 \Im\{(a x_i + b) (c x_i + d)^*\} }{|c x_i + d|^2 + |a x_i + b|^2} \right|^2 \nonumber\\
                                 &\qquad+ \left| \sum_i m_i  \frac{2 \Re\{(a x_i + b) (c x_i + d)^*\} }{|c x_i + d|^2 + |a x_i + b|^2} \right|^2 \nonumber\\
                                 &\qquad+ \left| \sum_i m_i \frac{|c x_i + d|^2 - |a x_i + b|^2}{|c x_i + d|^2 + |a x_i + b|^2} \right|^2,
\end{align}
where $a = \sqrt{\frac{1}{1 - \alpha}} e^{ - I \frac{\Phi}{2}} \cos \frac{\Theta}{2}$, $b = - \sqrt{\frac{1}{1 - \alpha}} e^{ I \frac{\Phi}{2}} \sin \frac{\Theta}{2}$, $c = \sqrt{1-\alpha} e^{ - I \frac{\Phi}{2}} \sin \frac{\Theta}{2}$, and $d = \sqrt{1-\alpha} e^{ I \frac{\Phi}{2}} \cos \frac{\Theta}{2}$. Let us remark here, that $a,b,c,d$ are actually the matrix elements of the corresponding operator $A$ as in Section \ref{sec:geometric} (which were denoted by $\alpha$, $\beta$, $\gamma$, and $\delta$, there). Note that $f$ is built by concatenation, addition, and multiplication of elementary functions that are all continuous, as well as rational functions, which are continuous except at points where the denominator is 0.
However, it can be easily verified that none of the denominators appearing in Eq. (\theequation) can attain 0. Thus, $f$ is continuous on $D$. Let us remark here, however, that $f$ is not uniformly continuous.

Proof of statement (ii). Let us first show the statement for $x=(0,0,0)$, which corresponds to $A = \identity$. Consider $x' = (\frac{\epsilon}{1+ \epsilon}, \Theta, \Phi)$, where $\epsilon > 0$, $\Theta = \arccos \sum_i m_i \cos \theta_i$, $\Phi =  \arcsin \frac{\sum_i m_i \sin \phi_i \sin \theta_i}{\sin \Theta}$. $x'$ corresponds to first rotating the vectors $\vec{v_i}$ such that the with multiplicities weighted vector sum of $\vec{v}_i$ points into direction $(0,0,1)^T$ followed by an $\epsilon$-hinging transformation.
 Taylor expansion of $f$ around $\epsilon = 0$ shows
\begin{align}
f(\frac{\epsilon}{1+ \epsilon},\Theta,\Phi) &= f(0, \Theta,\Phi) \nonumber\\
			& \qquad - 2 \epsilon  \left(\sum_i m_i \cos \tilde{\theta}_i \right) \left(\sum_i m_i \sin^2 \tilde{\theta}_i \right) \nonumber\\
			&\qquad  + \mathcal{O}(\epsilon^2),
\end{align}
where by $\tilde{\theta}_i$ we denote the spherical coordinates of the vectors after the rotation. It is clear that there exists an $\epsilon > 0$ such that $f(x') < f(x)$, unless either $\sum_i m_i \cos \tilde{\theta}_i= 0$, or $\sum_i m_i \sin^2 \tilde{\theta}_i =0$. In case $\sum_i m_i \cos \tilde{\theta}_i= 0$ the vectors already sum to 0, i.e., $f(x) = 0$. In the latter case, all vectors are pointing either in direction $(0,0,1)^T$ or $(0,0,-1)^T$. This cannot be the case as by assumption the vectors point in at least three distinct directions ($l \geq 3$).
Above we have proven statement (ii) for the special case $x = (0,0,0)$. Note that this shows that for any initial orientation of the vectors $\{\vec{v}_i\}$, one can find an operator $A$ such that the resulting vectors $\{\vec{v'}_i\}$ lead to a smaller value for $f$. This is due to the fact that $x = (0,0,0)$ corresponds to $A=\identity$. The situation $x \neq (0,0,0)$, i.e., $A\neq \identity$, can be easily dealt with by applying the above procedure to the vectors $\vec{v_i}$ obtained after an initial transformation defined by $x$, i.e., $A$. Then, $x'$ can be constructed by concatenating the $(\frac{\epsilon}{1 + \epsilon},\Theta,\Phi)$-transformation and the transformation defined by the initial $x$.

Proof of statement (iii). Let us first remark that this statement only holds if $m_1 < \sum_{i=2}^l m_i$, which holds true due to the assumption. Let us distinguish two cases. In the first case, none of the $\vec{v}_i$ coincides with the axis defined by $(\Theta, \Phi)$, i.e., there exists no $i$ such that $\phi_i = \Phi_i$ and $\theta_i = \Theta_i$. Hence, by increasing $\alpha$, the transformation $A$ hinges all of the vectors down with respect to the axis $(\Theta, \Phi)$. We will use the fact that once all of the vectors lie in the southern hemisphere with respect to the axis, increasing $\alpha$ further will only increase the squared norm of the vector sum, $f(\alpha,\Theta,\Phi)$. The $\alpha_*$, which the statement guarantees to exist, can thus be simply determined by $\alpha_* = 0$ in case all $\vec{v}_i$ already lie in the southern hemisphere, or
\begin{align}
\alpha_* = 1 - \tan(\theta/2),
\end{align}
where $\theta$ is the angle between the transformation axis and the vector $\vec{v}_i$ that is the most aligned with the axis, otherwise. Let us now deal with the second case, in which one $\vec{v}_i$ exactly coincides with the transformation axis. This vector is invariant under the transformation. $\alpha_*$ can then be determined following the same procedure as in the first case, completely disregarding the vector coinciding with the axis, i.e., it can be constructed by requiring that all of the remaining $l-1$ distinct vectors lie in the southern hemisphere.

Proof of statement (iv). Due to statement (i), $\tilde{f}$ is continuous in $D$. It remains to be shown that $\tilde{f}$ is lower semi-continuous in $\tilde{D} \setminus D$. Let $x = (1, \Theta, \Phi)\in \tilde{D} \setminus D$. We now distinguish the two cases $( \Theta, \Phi) \neq (\theta_i, \phi_i)$ for all $i$ and $( \Theta, \Phi) = (\theta_i, \phi_i)$ for some $i$.

Let us show that $\tilde{f}$ is actually continuous in points that belong to the first case. To this end, it suffices to show that for every $\delta>0$ there exists an $R$, such that for all $\epsilon_\alpha, \epsilon_\Theta, \epsilon_\Phi$ with $(1+\epsilon_\alpha, \Theta + \epsilon_\Theta, \Phi + \epsilon_\Phi) \in \tilde{D}$ and $\epsilon_\alpha^2 + \epsilon_\Theta^2 + \epsilon_\Phi^2 \leq R$, it holds that
\begin{align}
\label{eq:cont}
|\tilde{f} (1 + \epsilon_\alpha, \Theta + \epsilon_\Theta, \Phi + \epsilon_\Phi) - \tilde{f}(1,\Theta,\Phi)| \leq \delta.
\end{align}
 Let us denote the angle between the axis defined by $(\Theta, \Phi)$ and the vector $\vec{v}_i$ which is the most aligned with $(\Theta, \Phi)$ by $\theta$. Then let us first choose $R$ such that the angle between vectors $\vec{v}_i$ and directions $(\Theta + \epsilon_\Theta , \Phi + \epsilon_\Phi)$ is at least $\theta/2$ for all $\epsilon_\Theta, \epsilon_\Phi$. This makes sure that no vector $\vec{v}_i$ points into a direction that lies within (or is close to) the chosen $R$-region.
Now either Eq. (\theequation) is already satisfied with the chosen $R$, or we decrease $R$ in a second step such that the condition in Eq. (\theequation) is met.
To this end, let us find an $\alpha_{\text{min}} \in [0, 1)$ such that for all $(\Theta+\epsilon_\Theta, \Phi + \epsilon_\Phi)$ with $\epsilon_\Theta^2 + \epsilon_\Phi^2 \leq R$, it holds that $\tilde{f} (\alpha_{\text{min}}, \Theta + \epsilon_\Theta, \Phi + \epsilon_\Phi) \geq n^2-\delta$. This can be achieved as follows. First, one identifies the vector $(\Theta', \Phi')$ among $(\Theta+\epsilon_\Theta, \Phi + \epsilon_\Phi)$, which is the closest to some $\vec{v}_i$ (let us denote the corresponding angle between those two vectors by $\beta \geq \theta/2$). Then, choose $\alpha_{\text{min}}$ such that the projection of the transformed ${\vec{v'}_i}$ onto direction $(\Theta', \Phi')$ is antiparallely aligned with the direction $(\Theta', \Phi')$ and, moreover, the projected vector has a norm of at least $\sqrt{1-\delta/n^2}$. Choosing
\begin{align}
\alpha_{\text{min}} = \text{max}\left\{0,1- \tan\left( \frac{\beta}{2} \right) \tan \left(  \frac{1}{2} \arccos \sqrt{1-\frac{\delta}{n^2}} \right)\right\}
\end{align}
 suffices. Pictorially, this can be understood as hinging down the vector $\vec{v}_i$ sufficiently, as sketched in Figure \ref{fig:hinging}. Note that all other vectors are hinged down even more, i.e., each vector $\vec{v'}_i$ projected onto the $(\Theta', \Phi')$-axis has a norm of at least $\sqrt{1-\delta/n^2}$. Moreover, increasing $\alpha$ beyond $\alpha_{\text{min}}$ only increases the antiparallel alignment of the vectors with the $(\Theta', \Phi')$ direction. Thus, the norm of each of the projected vectors increases individually. These considerations indeed show that $\tilde{f} (\alpha \geq \alpha_{\text{min}}, \Theta + \epsilon_\Theta, \Phi + \epsilon_\Phi)  \geq n^2-\delta$.
 Now, note that in case $(1-\alpha_{\text{min}})^2 \geq R$, Eq. (\ref{eq:cont}) is already satisfied. Otherwise let us redefine $R= (1-\alpha_{\text{min}})^2$. Then, Eq.  (\ref{eq:cont}) holds for the new choice of $R$. This shows that $\tilde{f}$ is indeed continuous in the points $ (1, \Theta, \Phi)\in \tilde{D} \setminus D$ for which $( \Theta, \Phi) \neq (\theta_i, \phi_i)$ for all $i$.

\begin{figure}[H]
 \includegraphics[width=1.0 \columnwidth]{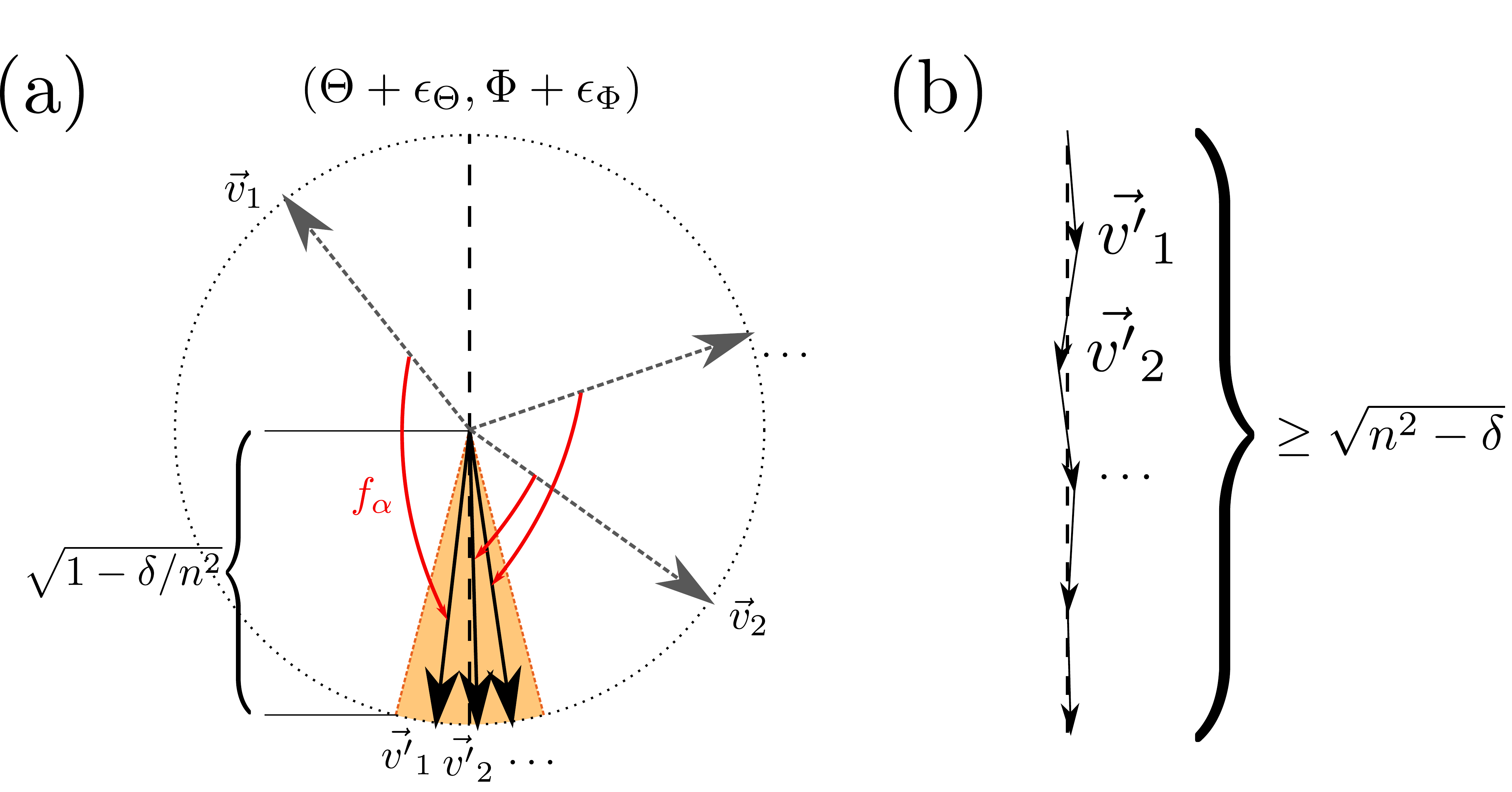}
  \caption{Figure (a) shows that $R$ can always be chosen s.t. for all $\epsilon_\alpha, \epsilon_\Theta, \epsilon_\Phi$ with $(1+\epsilon_\alpha, \Theta + \epsilon_\Theta, \Phi + \epsilon_\Phi) \in \tilde{D}$ and $\epsilon_\alpha^2 + \epsilon_\Theta^2 + \epsilon_\Phi^2 \leq R$, Eq. (\ref{eq:cont}) holds (first case). As explained in the main text, the vectors $\vec{v}_i$ are hinged down (along the axis $(\Theta + \epsilon_\Theta, \Phi +\epsilon_\Phi)$) sufficiently, i.e. with sufficiently small $\epsilon_\alpha$, such that all vectors are transformed into the shaded region, which defines the region for which Eq. (\ref{eq:cont}) is fulfilled. This is even more precisely depicted in Figure (b), which shows that the vectors in the shaded region in Figure(a) lead to a value of $f$, which is at least $n^2 - \delta$. Recall that $f(1,\Theta,\Phi) = n^2$. The second case works similarly, as explained in the main text.}
  \label{fig:hinging}
\end{figure}

Let us now show that $\tilde{f}$ is lower semi-continuous in points $ (1, \theta_i, \phi_i)$, i.e., points that belong to the second case. Recall that $\tilde{f}(1,  \theta_i, \phi_i) = (n - 2 m_i)^2$. Obviously, $\tilde{f}$ is not continuous, as, e.g., $\tilde{f}(1,  \theta_i + \epsilon_\Theta, \phi_i) = n^2$ for any small $\epsilon_\Theta \neq 0$. Let us follow a similar procedure as in the first case, however, let us completely disregard the vector $\vec{v}_i$ that is perfectly aligned with $(\Theta, \Phi)$. Instead, we make sure to find an $R$ such that the (with multiplicities weighted) norm of the sum of the projections onto the $(\theta_i + \epsilon_\Theta, \phi_i + \epsilon_\Phi)$-direction of the remaining vectors $\vec{v'}_j$ amounts to at least $\sqrt{(n- 2 m_i)^2 - \delta}$. Note that no matter where $\vec{v'}_i$ is pointing (actually, it will not be possible to make a statement about where it will point to), it holds that the whole (with multiplicities weighted) sum of $\vec{v'}_j$ will (within the relevant region around $(1, \theta_i, \phi_i)$) have a squared norm of at least $(n-2m_i)^2 - \delta$. In other words, for all $\delta$ we have constructed an $R$ such that  for all $\epsilon_\alpha, \epsilon_\Theta, \epsilon_\Phi$ such that $\epsilon_\alpha^2 + \epsilon_\Theta^2 + \epsilon_\Phi^2 \leq R$, $\tilde{f} (1 + \epsilon_\alpha, \Theta + \epsilon_\Theta, \Phi + \epsilon_\Phi) \geq \tilde{f}(1,\Theta,\Phi) - \delta$. This shows that $\tilde{f}$ is lower semi-continuous in the points $(1,\theta_i,\phi_i)$ and completes the proof of the fact that $\tilde{f}$ is lower semi-continuous on $\tilde{D}$.

Proof of statement (v). It is quite simple to see that the statement holds for the case $( \Theta, \Phi) \neq (\theta_i, \phi_i)$ for all $i$.
As $\sum m_i = n$, we have that $\left| \sum_i m_i \vec{v_i} \right|^2 = n^2$ only if all vectors are aligned. With the operation considered in the lemma, degeneracies of vectors cannot be created or lifted. Hence, $f(D) < n^2$ and thus also $\tilde{f}(D) < n^2$. It remains to be proven that in case $( \Theta, \Phi) = (\theta_i, \phi_i)$ for some $i$, moreover, $\tilde{f}(\alpha, \Theta, \Phi) < (n-2m_i)^2$ for some $\alpha \in [0,1)$. Similarly as in statement (iii), we have that for sufficiently large $\alpha \in [0,1)$ all the vectors except $\vec{v}_i$ lie in the southern hemisphere with respect to the axis $(\Theta, \Phi)$. Similarly as before, it can be seen that from then on it holds that $\tilde{f}(\alpha, \Theta, \Phi) < (n-2m_i)^2$ for all such sufficiently large $\alpha< 1$. This proves statement (v).
\end{proof}

\onecolumngrid
\section{Tables of SLOCC classes  in \texorpdfstring{$2 \times m \times n$}{2 x \textit{m} x \textit{n}}  and their orbit type for small system sizes}
\label{app:tables}

In this appendix, we provide complete tables for SLOCC classes in $2 \times m \times n$  systems for small system sizes, that is, up to $m \leq n \leq 5$, building on the characterization of SLOCC classes in $2 \times m \times n$ systems provided in \cite{ChMi10}. In particluar, we list the respective orbit types according to the charactrization presented in Section \ref{sec:2mncharacterization}. Already within such small system sizes, a rich variety of different behaviours manifests.
In the main text we have shown that strictly semistable states exist if and only if $m=n \geq 4$ (Theorem \ref{thm:2mnsemistable}). The systems for which we provide tables here start from the the well known 3-qubit system \cite{DuVi00}. Then, examples in which no strictly semistable states exist are covered (all system sizes up to $2 \times 3 \times 5$ as well as $2\times 4 \times 5$). Moreover, cases with a rich structure ($2 \times 4 \times 4$ and $2 \times 5 \times 5$), where all different orbit types are present, are covered. The tables also include an instance of a system, in which all classes are in the null-cone, $2 \times 3 \times 5$. Moreover, for $n = 2m$, it is known there exsits only a single SLOCC class containing fully entangled states \cite{DuSh10}, which is represented by two bipartite maximally entangled states shared by party A and C, as well as B and C, respectively. Within the here considered system sizes, this becomes apparent in $2 \times 2 \times 4$.

The tables give an exhaustive list of all SLOCC classes (families of SLOCC classes) of the considered dimensions in terms of the corresponding matrix pencil in KCF. One representative state (disregarding normalization for brievity) is given for each class, as well as the class' orbit types, and the maximal orbit dimensions \cite{BrLe19}. Moreover, we state how the orbit type has been identified. For strictly semistable classes, the polystable classes in their closure is identified.

\newpage

\subsection{\texorpdfstring{$2 \times 2 \times 2$}{2 x 2 x 2}}
The dimension of $\mathbb{P}(\mathcal{H})$ is 7, the maximal orbit dimension is 7.

\begingroup
\begin{table}[H]
\centering
\begin{tabular}{p{0.3cm}p{3.7cm}p{5.5cm}p{1.5cm}p{4cm}}
\multicolumn{1}{>{\arraybackslash}p{0.3cm}}{No.} 
    & \multicolumn{1}{>{\arraybackslash}p{3.7cm}}{Matrix Pencil} 
    & \multicolumn{1}{>{\arraybackslash}p{5.5cm}}{A representative}
    & \multicolumn{1}{>{\arraybackslash}p{1.5cm}}{Type}
    & \multicolumn{1}{>{\arraybackslash}p{4cm}}{Comments; Identified via}\\
 \hline
1 &  $\begin{pmatrix}
\lambda   & \mu  \\
   &  \lambda
\end{pmatrix}$ &  $\begin{aligned} \ket{\psi_1} & = \ket{0}\ket{01} \\ & + \ket{1}(\ket{00}+\ket{11})\end{aligned}$  & Null-cone &   Observation \ref{obs:diagonalreduction} and the corresponding diagonal matrix pencil (see Observation \ref{obs:diagonalreduction}) is separable in the splitting A|BC; Equivalent to three-qubit W-state.  \\

2 &  $\begin{pmatrix}
\lambda   &   \\
   &  \mu
\end{pmatrix}$ &  $\begin{aligned} \ket{\psi_2} & = \ket{0}\ket{11} \\ & + \ket{1}\ket{00}\end{aligned}$  & Stable &  Theorem \ref{thm:diagtype};  Equivalent to three-qubit GHZ-state. \\
\end{tabular}
\end{table}
\endgroup

\subsection{\texorpdfstring{$2 \times 2 \times 3$}{2 x 2 x 3}}
The dimension of $\mathbb{P}(\mathcal{H})$ is 11, the maximal orbit dimension is 11.

\begingroup
\begin{table}[H]
\centering
\begin{tabular}{p{0.3cm}p{3.7cm}p{5.5cm}p{1.5cm}p{4cm}}
\multicolumn{1}{>{\arraybackslash}p{0.3cm}}{No.} 
    & \multicolumn{1}{>{\arraybackslash}p{3.7cm}}{Matrix Pencil} 
    & \multicolumn{1}{>{\arraybackslash}p{5.5cm}}{A representative}
    & \multicolumn{1}{>{\arraybackslash}p{1.5cm}}{Type}
    & \multicolumn{1}{>{\arraybackslash}p{4cm}}{Comments; Identified via}\\
 \hline
1 &  $\begin{pmatrix}
\lambda   & \mu&  \\
   &  \lambda & \mu
\end{pmatrix}$ &  $\begin{aligned} \ket{\psi_1} & = \ket{0}(\ket{01} + \ket{12}) \\ & + \ket{1}(\ket{00}+\ket{11})\end{aligned}$  & Stable &   Shown in \cite{BrLe19}; Generic class.  \\

2 &  $\begin{pmatrix}
\lambda   & \mu &   \\
   &  & \lambda
\end{pmatrix}$ &  $\begin{aligned} \ket{\psi_2} & = \ket{0}\ket{01} \\ & + \ket{1}(\ket{00} + \ket{12})\end{aligned}$  & Null-cone &  Appendix \ref{app:nullcone2mn}.\\
\end{tabular}
\end{table}
\endgroup

\subsection{\texorpdfstring{$2 \times 2 \times 4$}{2 x 2 x 4}}
The dimension of $\mathbb{P}(\mathcal{H})$ is 15, the maximal orbit dimension is 15.

\begingroup
\begin{table}[H]
\centering
\begin{tabular}{p{0.3cm}p{3.7cm}p{5.5cm}p{1.5cm}p{4cm}}
\multicolumn{1}{>{\arraybackslash}p{0.3cm}}{No.} 
    & \multicolumn{1}{>{\arraybackslash}p{3.7cm}}{Matrix Pencil} 
    & \multicolumn{1}{>{\arraybackslash}p{5.5cm}}{A representative}
    & \multicolumn{1}{>{\arraybackslash}p{1.5cm}}{Type}
    & \multicolumn{1}{>{\arraybackslash}p{4cm}}{Comments; Identified via}\\
 \hline
1 &  $\begin{pmatrix}
\lambda   & \mu&  & \\
   & &   \lambda & \mu
\end{pmatrix}$ &  $\begin{aligned} \ket{\psi_1} & = \ket{0}(\ket{01} + \ket{13}) \\ & + \ket{1}(\ket{00}+\ket{12})\end{aligned}$  & Stable &   Equivalent to $\ket{\phi^+}_{AC}\ket{\phi^+}_{BC}$; Only SLOCC class containing fully entangled states. \\
\end{tabular}
\end{table}
\endgroup

\subsection{\texorpdfstring{$2 \times 3 \times 3$}{2 x 3 x 3}}
The dimension of $\mathbb{P}(\mathcal{H})$ is 17, the maximal orbit dimension is 17.

\begingroup
\begin{table}[H]
\centering
\begin{tabular}{p{0.3cm}p{3.7cm}p{5.5cm}p{1.5cm}p{4cm}}
\multicolumn{1}{>{\arraybackslash}p{0.3cm}}{No.} 
    & \multicolumn{1}{>{\arraybackslash}p{3.7cm}}{Matrix Pencil} 
    & \multicolumn{1}{>{\arraybackslash}p{5.5cm}}{A representative}
    & \multicolumn{1}{>{\arraybackslash}p{1.5cm}}{Type}
    & \multicolumn{1}{>{\arraybackslash}p{4cm}}{Comments; Identified via}\\
 \hline
1 &  $\begin{pmatrix}
\lambda   & \mu&  \\
   &    & \lambda\\
   &   & \mu
\end{pmatrix}$ &  $\begin{aligned} \ket{\psi_1} & = \ket{0}(\ket{01} + \ket{22}) \\ & + \ket{1}(\ket{00}+\ket{12})\end{aligned}$  & Null-cone &   Observation \ref{obs:epsilon_nu_null-cone}.   \\
2 &  $\begin{pmatrix}
\lambda   & \mu&  \\
   &  \lambda  & \mu \\
   &   & \lambda
\end{pmatrix}$ &  $\begin{aligned} \ket{\psi_2} & = \ket{0}(\ket{01} + \ket{12}) \\ & + \ket{1}(\ket{00}+\ket{11}+\ket{22})\end{aligned}$  & Null-cone &   Observation \ref{obs:diagonalreduction} and the corresponding diagonal matrix pencil (see Observation \ref{obs:diagonalreduction}) is separable in the splitting A|BC.    \\
3 &  $\begin{pmatrix}
\lambda   & \mu&  \\
   &  \lambda  &  \\
   &   & \lambda
\end{pmatrix}$ &  $\begin{aligned} \ket{\psi_3} & = \ket{0}\ket{01}  \\ & + \ket{1}(\ket{00}+\ket{11}+\ket{22})\end{aligned}$  & Null-cone &   \ditto \\
4 &  $\begin{pmatrix}
\lambda   & \mu&  \\
   &  \lambda  &  \\
   &   & \mu
\end{pmatrix}$ &  $\begin{aligned} \ket{\psi_4} & = \ket{0}(\ket{01} + \ket{22})  \\ & + \ket{1}(\ket{00}+\ket{11})\end{aligned}$  & Null-cone &   Observation \ref{obs:diagonalreduction} and Theorem \ref{thm:diagtype}.    \\

\end{tabular}
\end{table}
\endgroup

\begingroup
\begin{table}[H]
\centering
\begin{tabular}{p{0.3cm}p{3.7cm}p{5.5cm}p{1.5cm}p{4cm}}
5 &  $\begin{pmatrix}
\lambda   & &  \\
   &  \lambda  &  \\
   &   & \mu
\end{pmatrix}$ &  $\begin{aligned} \ket{\psi_5} & = \ket{0}\ket{22}  \\ & + \ket{1}(\ket{00}+\ket{11})\end{aligned}$  & Null-cone &   Theorem \ref{thm:diagtype}.    \\
6 &  $\begin{pmatrix}
\lambda   & &  \\
   &   \mu + \lambda  &  \\
   &   & \mu
\end{pmatrix}$ &  $\begin{aligned} \ket{\psi_6} & = \ket{0}(\ket{11} + \ket{22})  \\ & + \ket{1}(\ket{00}+\ket{11})\end{aligned}$  & Stable &   Theorem \ref{thm:diagtype}; Generic class \cite{HeGa18}.    \\
\end{tabular}
\end{table}
\endgroup

\subsection{\texorpdfstring{$2 \times 3 \times 4$}{2 x 3 x 4}}
The dimension of $\mathbb{P}(\mathcal{H})$ is 23, the maximal orbit dimension is 23.

\begingroup
\begin{table}[H]
\centering
\begin{tabular}{p{0.3cm}p{3.7cm}p{5.5cm}p{1.5cm}p{4cm}}
\multicolumn{1}{>{\arraybackslash}p{0.3cm}}{No.} 
    & \multicolumn{1}{>{\arraybackslash}p{3.7cm}}{Matrix Pencil} 
    & \multicolumn{1}{>{\arraybackslash}p{5.5cm}}{A representative}
    & \multicolumn{1}{>{\arraybackslash}p{1.5cm}}{Type}
    & \multicolumn{1}{>{\arraybackslash}p{4cm}}{Comments; Identified via}\\
 \hline
1 &  $\begin{pmatrix}
\lambda   & \mu&  & \\
   & \lambda   & \mu & \\
   &   & \lambda & \mu
\end{pmatrix}$ &  $\begin{aligned} \ket{\psi_1} & = \ket{0}(\ket{01} + \ket{12} + \ket{23}) \\ & + \ket{1}(\ket{00}+\ket{11}+\ket{22})\end{aligned}$  & Stable &   Shown in \cite{BrLe19}; Generic class.     \\

2 &  $\begin{pmatrix}
\lambda   & \mu&  & \\
   & \lambda   & \mu & \\
   &   &  & \lambda
\end{pmatrix}$ &  $\begin{aligned} \ket{\psi_2} & = \ket{0}(\ket{01} + \ket{12} ) \\ & + \ket{1}(\ket{00}+\ket{11}+\ket{23})\end{aligned}$  & Null-cone &    Appendix \ref{app:nullcone2mn}.     \\

3 &  $\begin{pmatrix}
\lambda   & \mu&  & \\
   &   & \lambda & \mu \\
   &   &  & \lambda
\end{pmatrix}$ &  $\begin{aligned} \ket{\psi_3} & = \ket{0}(\ket{01} + \ket{13} ) \\ & + \ket{1}(\ket{00}+\ket{12}+\ket{23})\end{aligned}$  & Null-cone &    \ditto    \\

4 &  $\begin{pmatrix}
\lambda   & \mu&  & \\
   &   & \lambda &  \\
   &   &  & \lambda
\end{pmatrix}$ &  $\begin{aligned} \ket{\psi_4} & = \ket{0}\ket{01}  \\ & + \ket{1}(\ket{00}+\ket{12}+\ket{23})\end{aligned}$  & Null-cone &   \ditto     \\

5 &  $\begin{pmatrix}
\lambda   & \mu&  & \\
   &   & \lambda &  \\
   &   &  & \mu
\end{pmatrix}$ &  $\begin{aligned} \ket{\psi_5} & = \ket{0}(\ket{01} + \ket{23} ) \\ & + \ket{1}(\ket{00}+\ket{12})\end{aligned}$  & Null-cone &    \ditto    \\
\end{tabular}
\end{table}
\endgroup

\subsection{\texorpdfstring{$2 \times 3 \times 5$}{2 x 3 x 5}}
\label{nolme}
The dimension of $\mathbb{P}(\mathcal{H})$ is 29, the maximal orbit dimension is 29.

\begingroup
\begin{table}[H]
\centering
\begin{tabular}{p{0.3cm}p{3.7cm}p{5.5cm}p{1.5cm}p{4cm}}
\multicolumn{1}{>{\arraybackslash}p{0.3cm}}{No.} 
    & \multicolumn{1}{>{\arraybackslash}p{3.7cm}}{Matrix Pencil} 
    & \multicolumn{1}{>{\arraybackslash}p{5.5cm}}{A representative}
    & \multicolumn{1}{>{\arraybackslash}p{1.5cm}}{Type}
    & \multicolumn{1}{>{\arraybackslash}p{4cm}}{Comments; Identified via}\\
 \hline
1 &  $\begin{pmatrix}
\lambda   & \mu&  & & \\
   & \lambda   & \mu & & \\
   &   &  & \lambda& \mu
\end{pmatrix}$ &  $\begin{aligned} \ket{\psi_1} & = \ket{0}(\ket{01} + \ket{12} + \ket{24}) \\ & + \ket{1}(\ket{00}+\ket{11}+\ket{23})\end{aligned}$  & Null-cone &   Appendix \ref{app:nullcone2mn};  Generic class \cite{HeGa18}.    \\

2 &  $\begin{pmatrix}
\lambda   & \mu&  & & \\
   &    & \lambda & \mu & \\
   &   &  & & \lambda
\end{pmatrix}$ &  $\begin{aligned} \ket{\psi_2} & = \ket{0}(\ket{01} + \ket{13} ) \\ & + \ket{1}(\ket{00}+\ket{12}+\ket{24})\end{aligned}$  & Null-cone &   Appendix \ref{app:nullcone2mn}.    \\
\end{tabular}
\end{table}
\endgroup

\subsection{\texorpdfstring{$2 \times 4 \times 4$}{2 x 4 x 4}}
\label{app:tables244}
The dimension of $\mathbb{P}(\mathcal{H})$ is 31, the maximal orbit dimension is 30.

\begingroup
\begin{table}[H]
\centering
\begin{tabular}{p{0.3cm}p{3.7cm}p{5.5cm}p{1.5cm}p{4cm}}
\multicolumn{1}{>{\arraybackslash}p{0.3cm}}{No.} 
    & \multicolumn{1}{>{\arraybackslash}p{3.7cm}}{Matrix Pencil} 
    & \multicolumn{1}{>{\arraybackslash}p{5.5cm}}{A representative}
    & \multicolumn{1}{>{\arraybackslash}p{1.5cm}}{Type}
    & \multicolumn{1}{>{\arraybackslash}p{4cm}}{Comments; Identified via}\\
 \hline
1 &  $\begin{pmatrix}
\lambda   & \mu  &   &   \\
   &  \lambda & \mu  &   \\
   &   &   & \lambda  \\
   &   &   &   \mu
\end{pmatrix}$ &  $\begin{aligned} \ket{\psi_1} & = \ket{0}(\ket{01}+\ket{12}+\ket{33}) \\ & + \ket{1}(\ket{00}+\ket{11}+\ket{23})\end{aligned}$  & Null-cone &  Observation \ref{obs:epsilon_nu_null-cone}.  \\

2 & $\begin{pmatrix}
\lambda   & \mu  &   &   \\
   &   & \lambda  &   \\
   &   &  \mu &   \\
   &   &   &   \lambda
\end{pmatrix}$ &$\begin{aligned} \ket{\psi_2} & = \ket{0}(\ket{01}+\ket{22}) \\ & + \ket{1}(\ket{00}+\ket{12}+\ket{33})\end{aligned}$ &Null-cone &   \ditto \\  
3 & $\begin{pmatrix}
 \lambda  & \mu  &   &   \\
   &   &  \lambda &   \\
   &   &  \mu & \lambda  \\
   &   &   &   \mu
\end{pmatrix}$ &$\begin{aligned}\ket{\psi_3} & = \ket{0}(\ket{01}+\ket{22}+\ket{33}) \\ & + \ket{1}(\ket{00}+\ket{12}+\ket{23})\end{aligned}$  & Null-cone &  \ditto  \\
\end{tabular}
\end{table}
\endgroup

\begingroup
\begin{table}[H]
\centering
\begin{tabular}{p{0.3cm}p{3.7cm}p{5.5cm}p{1.5cm}p{4cm}}

4 & $\begin{pmatrix}
\lambda   & \mu  &   &   \\
   & \lambda  & \mu  &   \\
   &   &  \lambda & \mu  \\
   &   &   &  \lambda
\end{pmatrix}$ & $\begin{aligned}\ket{\psi_4} & = \ket{0}(\ket{01}+\ket{12}+\ket{23}) \\ & + \ket{1}(\ket{00}+\ket{11}+\ket{22}+\ket{33})\end{aligned}$ & Null-cone &  Observation \ref{obs:diagonalreduction} and the corresponding diagonal matrix pencil (see Observation \ref{obs:diagonalreduction}) is separable in the splitting A|BC.   \\
5 & $\begin{pmatrix}
  \lambda & \mu  &   &   \\
   & \lambda  & \mu  &   \\
   &   &  \lambda &   \\
   &   &   &   \lambda
\end{pmatrix}$ & $\begin{aligned}\ket{\psi_5} & = \ket{0}(\ket{01}+\ket{12}) \\ & + \ket{1}(\ket{00}+\ket{11}+\ket{22}+\ket{33})\end{aligned}$ & Null-cone &  \ditto   \\
6 &$\begin{pmatrix}
  \lambda & \mu  &   &   \\
   &  \lambda &   &   \\
   &   &  \lambda & \mu  \\
   &   &   &  \lambda
\end{pmatrix}$  & $\begin{aligned}\ket{\psi_6} & = \ket{0}(\ket{01}+\ket{23}) \\ & + \ket{1}(\ket{00}+\ket{11}+\ket{22}+\ket{33})\end{aligned}$ & Null-cone &  \ditto  \\

7 & $\begin{pmatrix}
 \lambda  & \mu   &   &   \\
   &  \lambda &   &   \\
   &   & \lambda  &   \\
   &   &   &   \lambda
\end{pmatrix}$ &$\begin{aligned}\ket{\psi_7} & = \ket{0}\ket{01} \\ & + \ket{1}(\ket{00}+\ket{11}+\ket{22}+\ket{33})\end{aligned}$  & Null-cone &  \ditto   \\

8 & $\begin{pmatrix}
  \lambda & \mu  &   &   \\
   & \lambda  & \mu  &   \\
   &   &  \lambda &   \\
   &   &   &   \mu
\end{pmatrix}$ &$\begin{aligned} \ket{\psi_8} & = \ket{0}(\ket{01}+\ket{12}+\ket{33}) \\ & + \ket{1}(\ket{00}+\ket{11}+\ket{22})\end{aligned}$  & Null-cone & Observation \ref{obs:diagonalreduction} and Theorem \ref{thm:diagtype}. \\

9 & $\begin{pmatrix}
\lambda  & \mu  &   &   \\
   & \lambda  &   &   \\
   &   &  \lambda &   \\
   &   &   &   \mu
\end{pmatrix}$ & $\begin{aligned}\ket{\psi_9} & = \ket{0}(\ket{01}+\ket{33}) \\ & + \ket{1}(\ket{00}+\ket{11}+\ket{22})\end{aligned}$ & Null-cone &   \ditto   \\

10 & $\begin{pmatrix}
  \lambda &   &   &   \\
   & \lambda  &   &   \\
   &   &  \lambda &   \\
   &   &   &   \mu
\end{pmatrix}$ &$\begin{aligned}\ket{\psi_{10}} & = \ket{0}\ket{33} \\ & + \ket{1}(\ket{00}+\ket{11}+\ket{22})\end{aligned}$  &  Null-cone &  Theorem \ref{thm:diagtype}.  \\
11 & $\begin{pmatrix}
\lambda   & \mu   &   &   \\
   &  \lambda &   &   \\
   &   &  \mu & \lambda  \\
   &   &   &   \mu
\end{pmatrix}$ & $\begin{aligned}\ket{\psi_{11}} & = \ket{0}(\ket{01}+\ket{22}+\ket{33}) \\ & + \ket{1}(\ket{00}+\ket{11}+\ket{23})\end{aligned}$ & Strictly semistable &   Observation \ref{obs:diagonalreduction} and Theorem \ref{thm:diagtype}; Class 13 is in this class' closure.\\

12 & $\begin{pmatrix}
 \lambda  & \mu   &   &   \\
   &  \lambda &   &   \\
   &   &  \mu &   \\
   &   &   &   \mu
\end{pmatrix}$ & $\begin{aligned}\ket{\psi_{12}} & = \ket{0}(\ket{01}+\ket{22}+\ket{33}) \\ & + \ket{1}(\ket{00}+\ket{11})\end{aligned}$ & Strictly semistable & \ditto \\

13 & $\begin{pmatrix}
\lambda   &   &   &   \\
   &  \lambda &   &   \\
   &   & \mu  &   \\
   &   &   &  \mu
\end{pmatrix}$ & $\begin{aligned}\ket{\psi_{13}} & = \ket{0}(\ket{22}+\ket{33}) \\ & + \ket{1}(\ket{00}+\ket{11})\end{aligned}$ & Strictly polystable &  Theorem \ref{thm:diagtype}; Moreover, the given representative  is critical.    \\

14 & $\begin{pmatrix}
\lambda   & \mu   &   &   \\
   &  \lambda &   &   \\
   &   & \mu + \lambda  &   \\
   &   &   &   \mu
\end{pmatrix}$ &$\begin{aligned}\ket{\psi_{14}} & = \ket{0}(\ket{01}+\ket{22}+\ket{33}) \\ & + \ket{1}(\ket{00}+\ket{11}+\ket{22})\end{aligned}$  & Strictly semistable &  Observation \ref{obs:diagonalreduction} and Theorem \ref{thm:diagtype}; Class 13 is in this class' closure.  \\

15 & $\begin{pmatrix}
\lambda   &   &   &   \\
   &  \lambda &   &   \\
   &   &   \mu+ \lambda  &   \\
   &   &   &   \mu
\end{pmatrix}$ &$\begin{aligned}\ket{\psi_{15}} & = \ket{0}(\ket{22}+\ket{33}) \\ & + \ket{1}(\ket{00}+\ket{11}+\ket{22})\end{aligned}$  & Strictly semistable &    Theorem \ref{thm:diagtype}; Class 13 is in this class' closure.   \\
\end{tabular}
\end{table}
\endgroup

\begingroup
\begin{table}[H]
\centering
\begin{tabular}{p{0.3cm}p{3.7cm}p{5.5cm}p{1.5cm}p{4cm}}

16 & $\begin{pmatrix}
  \lambda &   &   &   \\
   & \mu+ \lambda   &   &   \\
   &   &  \mu &   \\
   &   &   &   x \mu + \lambda
\end{pmatrix}$ &$\begin{aligned}\ket{\psi_{16}} & = \ket{0}(\ket{11}+\ket{22}+x \ket{33}) \\ & + \ket{1}(\ket{00}+\ket{11}+\ket{33}),\end{aligned}$ \mbox{$x \neq 0, 1, \infty$}  & Stable & {Family of classes parametrized by free parameter $x$ fulfilling {$x \neq 0, 1, \infty$}; This family of SLOCC classes is the generic family, i.e., their union forms a full measure set of states \cite{HeGa18}; Stable due to Theorem \ref{thm:diagtype}.}
\end{tabular}
\end{table}
\endgroup

\subsection{\texorpdfstring{$2 \times 4 \times 5$}{2 x 4 x 5}}
The dimension of $\mathbb{P}(\mathcal{H})$ is 39, the maximal orbit dimension is 39.

\begingroup
\begin{table}[H]
\centering
\begin{tabular}{p{0.3cm}p{3.7cm}p{5.5cm}p{1.5cm}p{4cm}}
\multicolumn{1}{>{\arraybackslash}p{0.3cm}}{No.} 
    & \multicolumn{1}{>{\arraybackslash}p{3.7cm}}{Matrix Pencil} 
    & \multicolumn{1}{>{\arraybackslash}p{5.5cm}}{A representative}
    & \multicolumn{1}{>{\arraybackslash}p{1.5cm}}{Type}
    & \multicolumn{1}{>{\arraybackslash}p{4cm}}{Comments; Identified via}\\
 \hline
1 &  $\begin{pmatrix}
\lambda   & \mu  &   &   &\\
   &  \lambda & \mu  &   &\\
   &   & \lambda  & \mu  &\\
   &   &   &   \lambda & \mu
   \end{pmatrix}$ &  $\begin{aligned} \ket{\psi_1} & = \ket{0}(\ket{01}+\ket{12}+\ket{23}+\ket{34}) \\ & + \ket{1}(\ket{00}+\ket{11}+\ket{22}+\ket{33})\end{aligned}$  & Stable &  Shown in \cite{BrLe19}.  \\

2 &  $\begin{pmatrix}
\lambda   & \mu  &   &   &\\
   &    & \lambda  & \mu   &\\
   &   &   &   &\lambda \\
   &   &   &    & \mu
   \end{pmatrix}$ &  $\begin{aligned} \ket{\psi_2} & = \ket{0}(\ket{01}+\ket{13}+\ket{34}) \\ & + \ket{1}(\ket{00}+\ket{12}+\ket{24})\end{aligned}$  & Null-cone &  Observation \ref{obs:epsilon_nu_null-cone}.  \\

3 &  $\begin{pmatrix}
\lambda   & \mu  &   &   &\\
   &  \lambda  & \mu  &    &\\
   &   & \lambda  & \mu  & \\
   &   &   &    & \lambda
   \end{pmatrix}$ &  $\begin{aligned} \ket{\psi_3} & = \ket{0}(\ket{01}+\ket{12}+\ket{23}) \\ & + \ket{1}(\ket{00}+\ket{11}+\ket{22}+\ket{34})\end{aligned}$  & Null-cone &  Appendix \ref{app:nullcone2mn}.  \\

4 &  $\begin{pmatrix}
\lambda   & \mu  &   &   &\\
   &  \lambda  & \mu  &    &\\
   &   &   & \lambda  & \mu \\
   &   &   &    & \lambda
   \end{pmatrix}$ &  $\begin{aligned} \ket{\psi_4} & = \ket{0}(\ket{01}+\ket{12}+\ket{24}) \\ & + \ket{1}(\ket{00}+\ket{11}+\ket{23}+\ket{34})\end{aligned}$  & Null-cone &  \ditto  \\

5 &  $\begin{pmatrix}
\lambda   & \mu  &   &   &\\
   &  \lambda  & \mu  &    &\\
   &   &   & \lambda  &  \\
   &   &   &    & \lambda
   \end{pmatrix}$ &  $\begin{aligned} \ket{\psi_5} & = \ket{0}(\ket{01}+\ket{12}) \\ & + \ket{1}(\ket{00}+\ket{11}+\ket{23}+\ket{34})\end{aligned}$  & Null-cone &  \ditto  \\

6 &  $\begin{pmatrix}
\lambda   & \mu  &   &   &\\
   &  \lambda  & \mu  &    &\\
   &   &   & \lambda  &  \\
   &   &   &    & \mu
   \end{pmatrix}$ &  $\begin{aligned} \ket{\psi_6} & = \ket{0}(\ket{01}+\ket{12}+\ket{34}) \\ & + \ket{1}(\ket{00}+\ket{11}+\ket{23})\end{aligned}$  & Null-cone &  \ditto  \\

7 &  $\begin{pmatrix}
\lambda   & \mu  &   &   &\\
   &    & \lambda  & \mu   &\\
   &   &   & \lambda  & \mu  \\
   &   &   &    & \lambda
   \end{pmatrix}$ &  $\begin{aligned} \ket{\psi_7} & = \ket{0}(\ket{01}+\ket{23}+\ket{34}) \\ & + \ket{1}(\ket{00}+\ket{12}+\ket{23}+\ket{34})\end{aligned}$  & Null-cone &  \ditto  \\

8 &  $\begin{pmatrix}
\lambda   & \mu  &   &   &\\
   &    & \lambda  & \mu   &\\
   &   &   & \lambda  &   \\
   &   &   &    & \lambda
   \end{pmatrix}$ &  $\begin{aligned} \ket{\psi_8} & = \ket{0}(\ket{01}+\ket{23}) \\ & + \ket{1}(\ket{00}+\ket{12}+\ket{23}+\ket{34})\end{aligned}$  & Null-cone &  \ditto  \\

9 &  $\begin{pmatrix}
\lambda   & \mu  &   &   &\\
   &    & \lambda  &    &\\
   &   &   & \lambda  &   \\
   &   &   &    & \lambda
   \end{pmatrix}$ &  $\begin{aligned} \ket{\psi_9} & = \ket{0}\ket{01} \\ & + \ket{1}(\ket{00}+\ket{12}+\ket{23}+\ket{34})\end{aligned}$  & Null-cone &  \ditto  \\
10 &  $\begin{pmatrix}
\lambda   & \mu  &   &   &\\
   &    & \lambda  & \mu   &\\
   &   &   & \lambda  &   \\
   &   &   &    & \mu
   \end{pmatrix}$ &  $\begin{aligned} \ket{\psi_{10}} & = \ket{0}(\ket{01}+\ket{13}+\ket{34}) \\ & + \ket{1}(\ket{00}+\ket{12}+\ket{23})\end{aligned}$  & Null-cone &  \ditto  \\

\end{tabular}
\end{table}
\endgroup

\begingroup
\begin{table}[H]
\centering
\begin{tabular}{p{0.3cm}p{3.7cm}p{5.5cm}p{1.5cm}p{4cm}}
11 &  $\begin{pmatrix}
\lambda   & \mu  &   &   &\\
   &    & \lambda  &    &\\
   &   &   & \lambda  &   \\
   &   &   &    & \mu
   \end{pmatrix}$ &  $\begin{aligned} \ket{\psi_{11}} & = \ket{0}(\ket{01}+\ket{34}) \\ & + \ket{1}(\ket{00}+\ket{12}+\ket{23})\end{aligned}$  & Null-cone &  \ditto  \\

12 &  $\begin{pmatrix}
\lambda   & \mu  &   &   &\\
   &    & \lambda  &    &\\
   &   &   &   \mu+\lambda  &   \\
   &   &   &    & \mu
   \end{pmatrix}$ &  $\begin{aligned} \ket{\psi_{12}} & = \ket{0}(\ket{01}+\ket{23}+\ket{34}) \\ & + \ket{1}(\ket{00}+\ket{12}+\ket{23})\end{aligned}$  & Null-cone &  \ditto  \\
\end{tabular}
\end{table}
\endgroup

\subsection{\texorpdfstring{$2 \times 5 \times 5$}{2 x 5 x 5}}
\label{app:255}
The dimension of $\mathbb{P}(\mathcal{H})$ is 49, the maximal orbit dimension is 47.

\begingroup
\begin{table}[H]
\centering
\begin{tabular}{p{0.2cm}p{5.5cm}p{6.5cm}p{1.5cm}p{3cm}}
\multicolumn{1}{>{\arraybackslash}p{0.2cm}}{No.} 
    & \multicolumn{1}{>{\arraybackslash}p{5.5cm}}{Matrix Pencil} 
    & \multicolumn{1}{>{\arraybackslash}p{6.5cm}}{A representative}
    & \multicolumn{1}{>{\arraybackslash}p{1.5cm}}{Type}
    & \multicolumn{1}{>{\arraybackslash}p{3cm}}{Comments; Identified via}\\
 \hline
1 & $\begin{pmatrix}
\lambda   & \mu   &   &   &   \\
   &  \lambda & \mu  &   &   \\
   &   &  \lambda & \mu  &   \\
   &   &   &   &  \lambda \\
   &   &   &   &  \mu
\end{pmatrix}$ & $\begin{aligned}\ket{\psi_{1}} & = \ket{0}(\ket{01}+ \ket{12}+\ket{23}+\ket{44}) \\ & + \ket{1}(\ket{00}+\ket{11}+\ket{22}+\ket{34})\end{aligned}$ & Null-cone & Observation \ref{obs:epsilon_nu_null-cone}.  \\

2 & $\begin{pmatrix}
 \lambda  &\mu   &   &   &   \\
   &   & \lambda  &   &   \\
   &   &  \mu & \lambda  &   \\
   &   &   & \mu  & \lambda  \\
   &   &   &   &  \mu
\end{pmatrix}$ & $\begin{aligned}\ket{\psi_{2}} & = \ket{0}(\ket{01}+\ket{22}+\ket{33}+\ket{44} ) \\ & + \ket{1}(\ket{00}+\ket{12}+\ket{23}+\ket{34})\end{aligned}$ & Null-cone & \ditto \\

3 & $\begin{pmatrix}
\lambda   & \mu  &   &   &   \\
   &  \lambda & \mu  &   &   \\
   &   &   & \lambda  &   \\
   &   &   & \mu  & \lambda  \\
   &   &   &   &  \mu
\end{pmatrix}$ & $\begin{aligned}\ket{\psi_{3}} & = \ket{0}(\ket{01}+ \ket{12}+\ket{33}+\ket{44}) \\ & + \ket{1}(\ket{00}+\ket{11}+\ket{23}+\ket{34})\end{aligned}$ & Null-cone & \ditto \\

4 & $\begin{pmatrix}
 \lambda  & \mu  &   &   &   \\
   & \lambda  & \mu   &   &   \\
   &   &   & \lambda  &   \\
   &   &   & \mu  &   \\
   &   &   &   &  \lambda
\end{pmatrix}$ & $\begin{aligned}\ket{\psi_{4}} & = \ket{0}(\ket{01}+\ket{12}+\ket{33} ) \\ & + \ket{1}(\ket{00}+\ket{11}+\ket{23}+\ket{44})\end{aligned}$ & Null-cone & \ditto \\

5 & $\begin{pmatrix}
 \lambda  & \mu  &   &   &   \\
   &   &  \lambda &   &   \\
   &   &  \mu & \lambda   &   \\
   &   &   &  \mu &   \\
   &   &   &   &  \lambda
\end{pmatrix}$ & $\begin{aligned}\ket{\psi_{5}} & = \ket{0}(\ket{01}+\ket{22}+\ket{33} ) \\ & + \ket{1}(\ket{00}+\ket{12}+\ket{23}+\ket{44})\end{aligned}$ & Null-cone & \ditto  \\

6 & $\begin{pmatrix}
\lambda   & \mu   &   &   &   \\
   &   &  \lambda &   &   \\
   &   &  \mu &   &   \\
   &   &   & \lambda  &  \mu \\
   &   &   &   &  \lambda
\end{pmatrix}$ & $\begin{aligned}\ket{\psi_{6}} & =\ket{0}(\ket{01}+ \ket{22} + \ket{34}) \\ & + \ket{1}(\ket{00}+\ket{12}+\ket{33}+ \ket{44})\end{aligned}$ &Null-cone & \ditto \\
7 & $\begin{pmatrix}
\lambda   & \mu   &   &   &   \\
   &   &  \lambda &   &   \\
   &   &  \mu &   &   \\
   &   &   &  \lambda &   \\
   &   &   &   &  \lambda
\end{pmatrix}$ & $\begin{aligned}\ket{\psi_{7}} & = \ket{0}(\ket{01}+ \ket{22}) \\ & + \ket{1}(\ket{00}+\ket{12}+\ket{33}+ \ket{44})\end{aligned}$ & Null-cone & \ditto  \\
8 & $\begin{pmatrix}
\lambda   & \mu   &   &   &   \\
   &   &  \lambda &   &   \\
   &   &  \mu &   &   \\
   &   &   &  \lambda &   \\
   &   &   &   &  \mu
\end{pmatrix}$ & $\begin{aligned}\ket{\psi_{8}} & = \ket{0}(\ket{01}+ \ket{22} + \ket{44}) \\ & + \ket{1}(\ket{00}+\ket{12}+\ket{33})\end{aligned}$ & Null-cone & \ditto \\

\end{tabular}
\end{table}
\endgroup

\begingroup
\begin{table}[H]
\centering
\begin{tabular}{p{0.2cm}p{5.5cm}p{6.5cm}p{1.5cm}p{3cm}}

\hline
9 & $\begin{pmatrix}
 \lambda  & \mu  &   &   &   \\
   & \lambda  & \mu  &   &   \\
   &   &  \lambda &  \mu &   \\
   &   &   &  \lambda &  \mu \\
   &   &   &   & \lambda
\end{pmatrix}$ & $\begin{aligned}\ket{\psi_{9}} & = \ket{0}(\ket{01}+ \ket{12}+ \ket{23}+ \ket{34}) \\ & + \ket{1}(\ket{00}+\ket{11}+\ket{22}+\ket{33}+\ket{44})\end{aligned}$  & Null-cone & Observation \ref{obs:diagonalreduction} and the corresponding diagonal matrix pencil (see Observation \ref{obs:diagonalreduction}) is separable in the splitting A|BC. \\

10 & $\begin{pmatrix}
 \lambda  & \mu  &   &   &   \\
   & \lambda  &  \mu &   &   \\
   &   &  \lambda &  \mu &   \\
   &   &   &  \lambda &   \\
   &   &   &   & \lambda
\end{pmatrix}$ & $\begin{aligned}\ket{\psi_{10}} & = \ket{0}(\ket{01}+ \ket{12}+ \ket{23}) \\ & + \ket{1}(\ket{00}+\ket{11}+\ket{22}+\ket{33}+\ket{44})\end{aligned}$ & Null-cone & \ditto   \\

11 & $\begin{pmatrix}
 \lambda  & \mu  &   &   &   \\
   & \lambda  & \mu  &   &   \\
   &   &  \lambda &   &   \\
   &   &   &  \lambda & \mu  \\
   &   &   &   & \lambda
\end{pmatrix}$ & $\begin{aligned}\ket{\psi_{11}} & = \ket{0}(\ket{01}+ \ket{12}+ \ket{34}) \\ & + \ket{1}(\ket{00}+\ket{11}+\ket{22}+\ket{33}+\ket{44})\end{aligned}$ & Null-cone & \ditto \\

12 & $\begin{pmatrix}
 \lambda  & \mu  &   &   &   \\
   & \lambda  & \mu  &   &   \\
   &   &  \lambda &   &   \\
   &   &   &  \lambda &   \\
   &   &   &   & \lambda
\end{pmatrix}$ & $\begin{aligned}\ket{\psi_{12}} & = \ket{0}(\ket{01} + \ket{12}) \\ & + \ket{1}(\ket{00}+\ket{11}+\ket{22}+\ket{33}+\ket{44})\end{aligned}$ & Null-cone & \ditto \\

13 & $\begin{pmatrix}
 \lambda  & \mu  &   &   &   \\
   & \lambda  &   &   &   \\
   &   &  \lambda & \mu   &   \\
   &   &   &  \lambda &   \\
   &   &   &   & \lambda
\end{pmatrix}$ & $\begin{aligned}\ket{\psi_{13}} & = \ket{0}(\ket{01} + \ket{23}) \\ & + \ket{1}(\ket{00}+\ket{11}+\ket{22}+\ket{33}+\ket{44})\end{aligned}$ & Null-cone & \ditto \\

14 & $\begin{pmatrix}
 \lambda  & \mu  &   &   &   \\
   & \lambda  &   &   &   \\
   &   &  \lambda &   &   \\
   &   &   &  \lambda &   \\
   &   &   &   & \lambda
\end{pmatrix}$ & $\begin{aligned}\ket{\psi_{14}} & = \ket{0}\ket{01} \\ & + \ket{1}(\ket{00}+\ket{11}+\ket{22}+\ket{33}+\ket{44})\end{aligned}$ & Null-cone & \ditto \\

\hline
15 & $\begin{pmatrix}
\lambda   & \mu  &   &   &   \\
   &  \lambda &  \mu &   &   \\
   &   &  \lambda & \mu  &   \\
   &   &   &  \lambda &   \\
   &   &   &   &  \mu
\end{pmatrix}$ & $\begin{aligned}\ket{\psi_{15}} & = \ket{0}(\ket{01} + \ket{12} + \ket{23} + \ket{44} ) \\ & + \ket{1}(\ket{00}+\ket{11}+\ket{22}+\ket{33})\end{aligned}$ & Null-cone & Observation \ref{obs:diagonalreduction} and Theorem \ref{thm:diagtype}.  \\

16 & $\begin{pmatrix}
\lambda   &  \mu &   &   &   \\
   &  \lambda & \mu  &   &   \\
   &   &  \lambda &   &   \\
   &   &   &  \lambda &   \\
   &   &   &   &  \mu
\end{pmatrix}$ & $\begin{aligned}\ket{\psi_{16}} & = \ket{0}(\ket{01} +\ket{12} +  \ket{44} ) \\ & + \ket{1}(\ket{00}+\ket{11}+\ket{22}+\ket{33})\end{aligned}$ & Null-cone & \ditto  \\

17 & $\begin{pmatrix}
\lambda   & \mu  &   &   &   \\
   &  \lambda &   &   &   \\
   &   &  \lambda & \mu  &   \\
   &   &   &  \lambda &   \\
   &   &   &   &  \mu
\end{pmatrix}$ & $\begin{aligned}\ket{\psi_{17}} & = \ket{0}(\ket{01} + \ket{23} + \ket{44} ) \\ & + \ket{1}(\ket{00}+\ket{11}+\ket{22}+\ket{33})\end{aligned}$ & Null-cone & \ditto  \\

18 & $\begin{pmatrix}
\lambda   & \mu  &   &   &   \\
   &  \lambda &   &   &   \\
   &   &  \lambda &   &   \\
   &   &   &  \lambda &   \\
   &   &   &   &  \mu
\end{pmatrix}$ & $\begin{aligned}\ket{\psi_{18}} & = \ket{0}(\ket{01} + \ket{44} ) \\ & + \ket{1}(\ket{00}+\ket{11}+\ket{22}+\ket{33})\end{aligned}$ & Null-cone & \ditto  \\

19 & $\begin{pmatrix}
\lambda   &   &   &   &   \\
   &  \lambda &   &   &   \\
   &   &  \lambda &   &   \\
   &   &   &  \lambda &   \\
   &   &   &   &  \mu
\end{pmatrix}$ & $\begin{aligned}\ket{\psi_{19}} & = \ket{0}\ket{44}  \\ & + \ket{1}(\ket{00}+\ket{11}+\ket{22}+\ket{33})\end{aligned}$ & Null-cone & Theorem \ref{thm:diagtype}.  \\

\end{tabular}
\end{table}
\endgroup

\begingroup
\begin{table}[H]
\centering
\begin{tabular}{p{0.2cm}p{5.5cm}p{6.5cm}p{1.5cm}p{3cm}}
20 & $\begin{pmatrix}
 \lambda  & \mu  &   &   &   \\
   &  \lambda &  \mu &   &   \\
   &   &  \lambda &   &   \\
   &   &   &  \mu &  \lambda \\
   &   &   &   &  \mu
\end{pmatrix}$ & $\begin{aligned}\ket{\psi_{20}} & = \ket{0}(\ket{01} +\ket{12} +\ket{33}+\ket{44} ) \\ & + \ket{1}(\ket{00}+\ket{11}+ \ket{22}+ \ket{34})\end{aligned}$ & Null-cone & Observation \ref{obs:diagonalreduction} and Theorem \ref{thm:diagtype}. \\

21 & $\begin{pmatrix}
 \lambda  &  \mu &   &   &   \\
   &  \lambda &   &   &   \\
   &   &  \lambda &   &   \\
   &   &   &  \mu & \lambda  \\
   &   &   &   &  \mu
\end{pmatrix}$ & $\begin{aligned}\ket{\psi_{21}} & = \ket{0}(\ket{01} +\ket{33}+\ket{44} ) \\ & + \ket{1}(\ket{00}+\ket{11}+ \ket{22}+ \ket{34})\end{aligned}$ & Null-cone & \ditto  \\

22 & $\begin{pmatrix}
 \lambda  & \mu  &   &   &   \\
   &  \lambda & \mu  &   &   \\
   &   &  \lambda &   &   \\
   &   &   &  \mu &   \\
   &   &   &   &  \mu
\end{pmatrix}$ & $\begin{aligned}\ket{\psi_{22}} & = \ket{0}(\ket{01} +\ket{12} +\ket{33}+\ket{44} ) \\ & + \ket{1}(\ket{00}+\ket{11}+ \ket{22})\end{aligned}$ & Null-cone & \ditto \\

23 & $\begin{pmatrix}
 \lambda  &  \mu &   &   &   \\
   &  \lambda &   &   &   \\
   &   &  \lambda &   &   \\
   &   &   &  \mu &   \\
   &   &   &   &  \mu
\end{pmatrix}$ & $\begin{aligned}\ket{\psi_{23}} & = \ket{0}(\ket{01} + \ket{33}+\ket{44} ) \\ & + \ket{1}(\ket{00}+\ket{11}+ \ket{22})\end{aligned}$ & Null-cone & \ditto \\

24 & $\begin{pmatrix}
 \lambda  &   &   &   &   \\
   &  \lambda &   &   &   \\
   &   &  \lambda &   &   \\
   &   &   &  \mu & \lambda  \\
   &   &   &   &  \mu
\end{pmatrix}$ & $\begin{aligned}\ket{\psi_{24}} & = \ket{0}(\ket{33}+\ket{44} ) \\ & + \ket{1}(\ket{00}+\ket{11}+ \ket{22}+ \ket{34})\end{aligned}$ & Null-cone & \ditto \\

25 & $\begin{pmatrix}
 \lambda  &   &   &   &   \\
   &  \lambda &   &   &   \\
   &   &  \lambda &   &   \\
   &   &   &  \mu &   \\
   &   &   &   &  \mu
\end{pmatrix}$ & $\begin{aligned}\ket{\psi_{25}} & = \ket{0}(\ket{33}+\ket{44} ) \\ & + \ket{1}(\ket{00}+\ket{11}+ \ket{22})\end{aligned}$ & Null-cone & Theorem \ref{thm:diagtype}. \\

\hline
26 & $\begin{pmatrix}
 \lambda  & \mu  &   &   &   \\
   &  \lambda &  \mu &   &   \\
   &   &  \lambda &   &   \\
   &   &   & \mu + \lambda  &   \\
   &   &   &   &  \mu
\end{pmatrix}$ & $\begin{aligned}\ket{\psi_{26}} & = \ket{0}(\ket{01}+\ket{12}+\ket{33}+ \ket{44}) \\ & + \ket{1}(\ket{00}+\ket{11}+ \ket{22}+ \ket{33})\end{aligned}$ & Null-cone & Observation \ref{obs:diagonalreduction} and Theorem \ref{thm:diagtype}. \\

27 & $\begin{pmatrix}
 \lambda  & \mu   &   &   &   \\
   &  \lambda &   &   &   \\
   &   &  \lambda &   &   \\
   &   &   & \mu + \lambda  &   \\
   &   &   &   &  \mu
\end{pmatrix}$ & $\begin{aligned}\ket{\psi_{27}} & = \ket{0}(\ket{01}+\ket{33}+ \ket{44}) \\ & + \ket{1}(\ket{00}+\ket{11}+ \ket{22}+ \ket{33})\end{aligned}$ & Null-cone &  \ditto  \\

28 & $\begin{pmatrix}
 \lambda  &   &   &   &   \\
   &  \lambda &   &   &   \\
   &   &  \lambda &   &   \\
   &   &   & \mu + \lambda  &   \\
   &   &   &   &  \mu
\end{pmatrix}$ & $\begin{aligned}\ket{\psi_{28}} & = \ket{0}(\ket{33}+ \ket{44}) \\ & + \ket{1}(\ket{00}+\ket{11}+ \ket{22}+ \ket{33})\end{aligned}$ & Null-cone & Theorem \ref{thm:diagtype}. \\

29 & $\begin{pmatrix}
 \lambda  & \mu  &   &   &   \\
   &  \lambda &   &   &   \\
   &   &  \mu+\lambda  & \mu  &   \\
   &   &   &  \mu+\lambda &   \\
   &   &   &   &  \mu
\end{pmatrix}$ & $\begin{aligned}\ket{\psi_{29}} & =  \ket{0}(\ket{01}+ \ket{22}+\ket{23}+  \ket{33}+  \ket{44}) \\ & + \ket{1}(\ket{00}+\ket{11}+ \ket{22}+ \ket{33})\end{aligned}$ & Strictly semistable &  Observation \ref{obs:diagonalreduction} and Theorem \ref{thm:diagtype}; Moreover, $\ket{\psi_{29}}$ can be transformed to $\ket{\tilde{\psi}_{31}}$ asymptotically.\\

30 & $\begin{pmatrix}
 \lambda  &  \mu &   &   &   \\
   &  \lambda &   &   &   \\
   &   & \mu+\lambda  &   &   \\
   &   &   &  \mu +\lambda &   \\
   &   &   &   &  \mu
\end{pmatrix}$ & $\begin{aligned}\ket{\psi_{30}} & =  \ket{0}(\ket{01}+\ket{22}+ \ket{33}+ \ket{44}) \\ & + \ket{1}(\ket{00}+\ket{11}+ \ket{22}+ \ket{33})\end{aligned}$ & Strictly semistable &  \ditto  \\

\end{tabular}
\end{table}
\endgroup

\begingroup
\begin{table}[H]
\centering
\begin{tabular}{p{0.2cm}p{5.5cm}p{6.5cm}p{1.5cm}p{3cm}}
31 & $\begin{pmatrix}
 \lambda  &   &   &   &   \\
   &  \lambda &   &   &   \\
   &   & \mu+ \lambda  &   &   \\
   &   &   &  \mu + \lambda &   \\
   &   &   &   &  \mu
\end{pmatrix}$ & $\begin{aligned}\ket{\psi_{31}} & = \ket{0}(\ket{22}+ \ket{33}+ \ket{44}) \\ & + \ket{1}(\ket{00}+\ket{11}+ \ket{22}+ \ket{33})\end{aligned}$  $\begin{aligned} \ket{\tilde{\psi}_{31}}  & = \ket{0}(e^{i \phi}\ket{00}+ e^{i \phi} \ket{11}+ e^{-i \phi} \ket{22} \\ & \quad+e^{-i \phi} \ket{33}- \ket{44}) \\ & + \ket{1}(\ket{00}+\ket{11}+ \ket{22}+ \ket{33}+ \ket{44}),\end{aligned}$ where $\phi = -\frac{1}{2} \arccos{-\frac{7}{8}} $& Strictly polystable & Theorem \ref{thm:diagtype}; Moreover, $\ket{\tilde{\psi}_{31}}$ is critical and in the same SLOCC class.\\

\hline
32 & $\begin{pmatrix}
 \lambda  &  \mu &   &   &   \\
   &  \lambda  &   &   &   \\
   &   &  \mu + \lambda &   &   \\
   &   &   &  \mu  &   \\
   &   &   &   & x \mu + \lambda
\end{pmatrix}$ & $\begin{aligned}\ket{\psi_{32}} & =  \ket{0}(\ket{01} + \ket{22}+ \ket{33}+ x \ket{44}) \\ & +  \ket{1}(\ket{00}+ \ket{11}+ \ket{22}+ \ket{44}), \end{aligned}$ \newline where $x\neq0,1,\infty$ & {Strictly semistable} & Observation \ref{obs:diagonalreduction} and Theorem \ref{thm:diagtype}; This class' closure contains one class from family 33, the one with the same $x$.\\

33 & $\begin{pmatrix}
 \lambda  &   &   &   &   \\
   &  \lambda  &   &   &   \\
   &   &  \mu + \lambda &   &   \\
   &   &   &  \mu  &   \\
   &   &   &   & x \mu + \lambda
\end{pmatrix}$ & $\begin{aligned}\ket{\psi_{33}} & = \ket{0}(\ket{22}+ \ket{33}+ x \ket{44} \\ & +\ket{1}(\ket{00}+ \ket{11}+ \ket{22}+ \ket{44})),\end{aligned}$\newline where $x\neq0,1,\infty$ &{Strictly polystable} & Theorem \ref{thm:diagtype}.  \\

\hline
34 & $\begin{pmatrix}
 \lambda  &   &   &   &   \\
   &  \mu + \lambda &   &   &   \\
   &   &  \mu &   &   \\
   &   &   &  x_1 \mu + \lambda &   \\
   &   &   &   & x_2 \mu + \lambda
\end{pmatrix}$ & $\begin{aligned}\ket{\psi_{34}} & = \ket{0}(\ket{11}+\ket{22}+ x_1 \ket{33}+ x_2 \ket{44}) \\ & + \ket{1}(\ket{00}+ \ket{11}+ \ket{33}+ \ket{44}), \end{aligned}$ \newline where $\{0,1,\infty,x_1,x_2\} $ has a cardinality of five & Stable & Family of classes parametrized by free parameters $x_1$ and $x_2$ such that the set $\{0, 1, \infty, x_1, x_2\}$ has a cardinality of five; This family of SLOCC classes is the generic family, i.e., their union forms a full measure set of states \cite{HeGa18}; Stable due to Theorem \ref{thm:diagtype}. \\
\end{tabular}
\end{table}
\endgroup

\end{document}